\RequirePackage{float} 
\documentclass[a4paper,fleqn]{cas-sc}

\usepackage[numbers,sort&compress]{natbib}
\usepackage{amsthm}
\usepackage{algorithm}
\usepackage{algpseudocode}
\usepackage{microtype}
\usepackage[caption=false,font=small]{subfig}

\graphicspath{{Figures/}}

\def\bs{\boldsymbol}
\def\mf{\mathbf}
\def\mb{\mathbb}
\def\mc{\mathcal}

\newtheorem{theorem}{Theorem}
\newtheorem{remark}{Remark}

\makeatletter
\let\casappendix\appendix
\renewcommand\appendix{%
  \casappendix
  \setcounter{equation}{0}%
  \gdef\thesection{\appendixname~\@Alph\c@section}%
  \gdef\thefigure{\@Alph\c@section.\arabic{figure}}%
  \gdef\thetable{\@Alph\c@section.\arabic{table}}%
  \@addtoreset{equation}{section}%
  \gdef\theequation{\@Alph\c@section.\arabic{equation}}}
\makeatother

\ExplSyntaxOn
\RenewDocumentCommand \orcidauthor { m m }
  {
    \seq_gput_right:Nn \g_stm_orcid_seq
      { \href { https://orcid.org/#1 } { \ttfamily #1 } \parsename { #2 } \space (\eadauthor) }
  }
\ExplSyntaxOff

\begin{document}
\let\WriteBookmarks\relax
\def\floatpagepagefraction{1}
\def\textpagefraction{.001}

\shorttitle{Parameter-affine Dynamic Mode Decomposition}
\shortauthors{A. Chakrabarti et~al.}

\title[mode = title]{Parameter-affine Dynamic Mode Decomposition for Predicting Nonlinear Systems Across Parameter Space}

\author[1]{Ananda Chakrabarti}[orcid=0009-0004-3686-713X]
\cormark[1]
\ead{chakrabarti.44@osu.edu}

\affiliation[1]{organization={ElectroScience Laboratory and Department of Electrical and Computer Engineering, The Ohio State University},
  addressline={Columbus, OH 43212},
  country={USA}}

\author[1]{Haitham H. Saleh}[orcid=0009-0001-1688-1709]
\ead{saleh.175@osu.edu}

\author[2]{Indranil Nayak}[orcid=0000-0003-4702-941X]
\ead{nayaknil@stanford.edu}

\affiliation[2]{organization={SLAC National Accelerator Laboratory, Stanford University},
  addressline={Menlo Park, CA 94025},
  country={USA}}

\author[1]{Balasubramaniam Shanker}[orcid=0000-0003-4712-8186]
\ead{shanker@ece.osu.edu}

\author[1]{Fernando L. Teixeira}[orcid=0000-0002-1562-1462]
\ead{teixeira.5@osu.edu}

\author[1,3]{Debdipta Goswami}[orcid=0000-0002-5142-1222]
\ead{goswami.78@osu.edu}

\affiliation[3]{organization={Department of Mechanical and Aerospace Engineering, The Ohio State University},
  addressline={Columbus, OH 43210},
  country={USA}}

\cortext[1]{Corresponding author.}

\begin{abstract}
Predicting how a nonlinear system behaves at different values of its underlying physical parameters, such as a fluid's viscosity or a plasma system's field strength, is a central problem in data-driven reduced-order modeling (ROM) and scientific machine learning. Dynamic mode decomposition (DMD) is a ROM technique that builds inexpensive linear surrogate models of nonlinear dynamics from snapshot data using least-squares (LS) regression. However, a DMD model is valid only at the parameter values it was trained on, and must be retrained from scratch elsewhere. Existing parametric ROM methods generally fit separate models at the training parameters, then interpolate their modes, eigenvalues, reduced operators or other quantities to a target parameter. This requires the interpolated quantities to vary smoothly across the parameter space, and this assumption becomes unreliable when training parameter samples are sparse or the parameter space is multi-dimensional. We present parameter-affine dynamic mode decomposition (\texttt{paDMD}), which embeds known parametrization functions into the DMD LS regression itself. It learns a single parameter-affine DMD surrogate model, and can be evaluated across the parameter space without retraining. We demonstrate \texttt{paDMD} on flow past a cylinder, a fluidic pinball system, and two kinetic plasma systems: an oscillating electron beam and a virtual cathode oscillator, simulated by an electromagnetic particle-in-cell (EMPIC) method, each with varying physical parameters. Against state-of-the-art parametric ROM methods, \texttt{paDMD} is consistently more accurate over long prediction horizons, with sparse training samples and in multi-dimensional parameter settings.
\end{abstract}

\begin{highlights}
\item We introduce parameter-affine DMD (\texttt{paDMD}) for accurate prediction of nonlinear systems across their parameter space without retraining at each parameter value.
\item Unlike prior methods, \texttt{paDMD} embeds the system's parameter-affine structure directly into the DMD regression.
\item Our model outperforms other state-of-the-art parametric reduced-order modeling (ROM) methods, and is more robust to limited training data and multi-dimensional parameter spaces.
\item Validated on nonlinear systems, e.g., fluid dynamics and kinetic plasma systems simulated with EMPIC.
\end{highlights}


\maketitle

\section{Introduction} \label{sec:intro}
\noindent
Developments in machine learning (ML) have enabled routine data-driven modeling of complex systems across a wide range of applications, from playing board games \cite{silver2017mastering}, to speech recognition \cite{hinton2012deep}, and discovering new mathematical algorithms \cite{fawzi2022discovering}. While purely data-driven approaches can be powerful, they often struggle to model complex physical systems where data might be scarce and the governing equations and boundary conditions encode important physical and structural constraints. This has given rise to a family of physics-constrained ML methods \cite{raissi2019physics, greydanus2019hamiltonian, cranmer2020lagrangian,nayak2025temporally} that incorporate knowledge of the underlying dynamics directly into the learning pipeline for accurate system modeling. Despite this progress, modeling and control of complex, high-dimensional nonlinear dynamics remain a central challenge in science and engineering. Many systems of practical interest, such as those involving complex fluid flows and nonlinear plasma dynamics, can be studied through high-fidelity, first-principles simulations. These simulations provide rich spatiotemporal data, but their heavy computational cost makes repeated evaluations across the parameter space impractical. The curse of dimensionality further exacerbates this issue when the parameter space itself is high-dimensional.

Data-driven reduced-order models (ROMs) address these bottlenecks by learning low-rank surrogates directly from data. Such models reproduce the dynamics at a small fraction of the cost of high-fidelity simulations. Among data-driven ROM techniques, dynamic mode decomposition (DMD) \cite{schmid2010, Rowley2009} has become a popular choice, since it is simple to implement and yields a linear representation of the dynamics, which is well-suited for prediction, analysis, and control. DMD also has intrinsic connections to Koopman operator theory \cite{Koopman1931}, which offers a linear perspective on inherently nonlinear systems by operating in an infinite-dimensional space of observable functions. In practice, this operator is approximated using a finite-dimensional projection onto an invariant subspace spanned by a set of observables. DMD obtains such a projection by taking the observables to be the system states themselves. 

Several variants of the standard DMD algorithm have been developed \cite{tu2013dynamic, schmid2022dynamic}. In extended dynamic mode decomposition (EDMD) \cite{Williams2015}, hand-picked dictionaries of observable functions are used to learn the finite-dimensional projection of the Koopman operator. Hankel DMD \cite{kamb2020time, brunton2017chaos} uses time-delay coordinates of the states as vector-valued observables. This approach has connections to Takens' embedding theorem \cite{takens2006detecting} and often yields more accurate predictions than standard DMD. Other Koopman-based modeling frameworks automate the choice of observables entirely using neural networks (NNs), resulting in a family of Koopman autoencoders (KAEs) \cite{lusch2018deep, azencot2020forecasting, nayak2025temporally} for learning nonlinear systems. Online variants of DMD \cite{hemati2014dynamic, zhang2019online, kale2025efficient} learn and recursively update the DMD operator as new data becomes available. Physics-informed DMD \cite{baddoo2023physics} instead constrains the learned operator to respect known structural properties of the dynamics, such as symmetries, conservation laws or causality. Koopman/DMD methods have also been extended to control systems, yielding linear \cite{proctor2016dynamic, otto2019linearly} or bilinear \cite{Goswami2017, Goswami2021, chakrabarti2025temporally} models of nonlinear control systems. These models enable efficient, convex algorithms from linear systems theory to be applied to nonlinear control problems \cite{korda2018, kaiser2021data, mauroy2020koopman}. Common to all of these methods is a requirement on the \emph{persistency of excitation} in the training data \cite{ljung1999system, willems2005note}. That is, the snapshots must be rich enough to sufficiently excite the dynamics being identified. For autonomous systems this richness comes from the trajectory itself, while for control systems it is also supplied by the input signal, which the user is free to vary along the trajectory. Starting with fluid dynamics \cite{schmid2011applications, Rowley2009}, DMD and its Koopman-based variants have been successfully applied to a wide variety of domains, such as plasma dynamics \cite{taylor2018dynamic, nayak2024accelerating}, electromagnetic cavities \cite{nayak2023fly}, power networks \cite{susuki2011nonlinear}, epidemiology  \cite{proctor2015discovering}, financial markets \cite{mann2016dynamic}, robotics \cite{abraham2019active, folkestad2022koopnet}, traffic networks \cite{Avila2020}, neuroscience \cite{brunton2016extracting, raut2025arousal} among others, due to their simplicity and generalizability. In this paper, while we build on the bilinear Koopman framework, we restrict our attention to autonomous systems and focus on prediction rather than control.

A key limitation of data-driven ROM methods like DMD and its many variants is that they are typically trained on data for a single, fixed set of system parameters. In many applications, however, the dynamics depend strongly on governing physical parameters, such as the viscosity in fluid dynamics or externally applied fields in plasma systems. If any such parameter changes, the model must be retrained from scratch. Collecting new training data and constructing a separate model for every parameter set of interest quickly becomes computationally prohibitive. This motivates \emph{parametric} ROM and DMD algorithms that generalize across a continuum of parameters using data from only a small set of training samples. The central challenge is to introduce parameter dependence without sacrificing the efficiency and interpretability that make these methods attractive in the first place.

Existing parametric methods address this issue in several ways. A popular approach \cite{sayadi2015parametrized} forms a single global regression problem by \emph{stacking} snapshot data from multiple parameters and computing one DMD model, whose modes are then interpolated across parameters. However, such stacking becomes costly as the training set grows. Other methods compute separate ROMs at each training parameter and interpolate the \emph{eigenpairs}, \emph{ROM operators}, and/or other quantities \cite{huhn2023parametric} using standard interpolation techniques to obtain models at new parameters. One method \cite{andreuzzi2023dynamic}, for instance, projects snapshots into a POD space, forecasts reduced trajectories with DMD (or its variants), and reconstructs across the parameter space by interpolating the reduced trajectories.

Some methods interpolate quantities of interest via \emph{Grassmann manifold interpolation}. One such method \cite{xu2024parametric} interpolates reduced matrices on Grassmann manifolds with parallelization. However, this method focuses exclusively on linear parameter varying systems and requires the explicit mathematical expressions of the parametrized system matrices. Therefore, it is not a data-driven method and cannot be used to model nonlinear systems. Another method \cite{lu2023drips} performs an orthogonal Procrustes realignment of the ROM bases to a chosen reference configuration across the parameters. Then, it combines Grassmann-manifold interpolation of the ROM bases with matrix-manifold interpolation of the ROM operators. A more recent method \cite{du2025interpolation} introduces explicit schemes for interpolating \emph{DMD modes} via subspace angles and Grassmann-manifold methods, together with mode realignment, while interpolating eigenvalues separately using standard techniques. A related class of methods keeps the model fit and the parameter dependence in two separate stages. For instance, ROMs fitted to POD coefficients at each training parameter are interpolated by cubic splines and Gaussian process regression \cite{lin2023dmdrom}. Related approaches then couple this with either $k$-nearest-neighbors regression \cite{gao2022knndmd}, or build on a tensor train with the space, time, and parameter factors kept separated. The parameter factor is then mapped by Gaussian process regression \cite{lin2025tdmdgpr} or a neural operator branch network \cite{chen2025tdmddeeponet}. 

While these interpolation-based approaches can reduce computational cost relative to other parametric methods like stacked DMD \cite{sayadi2015parametrized}, they still have noticeable disadvantages. For instance, they work in two stages: individual ROMs must first be learned at the training parameter samples, and then they are interpolated across the parameter space. Moreover, they are highly sensitive to the model rank and the choice of training parameter samples, and they rely on the restrictive assumption that the quantities of interest being interpolated vary smoothly with the parameters. This sensitivity is amplified when the number of training parameter samples is small or when the parameter space is multi-dimensional, where interpolation can become unstable and may even produce models that fail to converge during prediction. Furthermore, where the quantities interpolated are DMD modes or eigenvalues, these emerge from an eigendecomposition in no fixed order, and so they must be matched and sorted carefully across the parameters before they can be interpolated. An NN-based approach is also proposed in \cite{guo2025learning}, where the Koopman operator is expressed as a function of the parameters and learned directly using a neural network. However, the demonstrations reported there use considerably larger parameter sample budgets than the regime we study here. Moreover, the black-box nature of all these interpolation/NN based parametric ROM methods makes exploration and optimization over the parameter space difficult, if not impractical.

These difficulties share a common root in the interpolation-based methods: the dynamics are fitted first and the parameter dependence is supplied afterwards, so the resulting model inherits the fragility of whichever interpolation is used to provide it. The neural-network approach of \cite{guo2025learning} differs in this respect, since the parametric operator and the dictionary are trained together, at the cost of a substantially heavier training problem. This paper presents \emph{parameter-affine dynamic mode decomposition} (\texttt{paDMD}), a parameter-affine DMD framework designed to exploit the affine parametric structure of the underlying dynamics. The main idea is to incorporate parameter dependence directly into the regression step used to identify the DMD evolution operator, rather than relying on interpolation of eigenpairs, DMD modes or reduced operators. By learning a parametric Koopman surrogate ROM that is affine in known parameter functions, \texttt{paDMD} preserves the computational workflow of bilinear DMD \cite{Goswami2017, Goswami2021} while avoiding many of the pitfalls of other parametric DMD methods. The resulting model can be evaluated across parameters using operations similar to those in bilinear DMD. However, unlike bilinear Koopman methods, \texttt{paDMD} focuses on \emph{parametrized systems} rather than control systems. This matters for system identification: a control input varies continually along a trajectory, which can provide persistency of excitation when it is suitably designed, whereas physical parameters like viscosity and field strength are often constant along each trajectory, so the excitation in the parameter domain is often one or two orders of magnitude smaller than in the time domain. This scarcity of parametric variation is a key challenge for system identification in the parametric setting, which is not present in the control setting. Another parametric ROM method \cite{mcquarrie2023nonintrusive} has a similar workflow in that the known parametric structure is embedded directly into the regression problem. However, this method learns a continuous-time polynomial ROM whose form should mirror that of the governing PDE, and therefore requires denser sampling and access to time-derivative data of the ROM states, which can be difficult to obtain and introduce numerical errors. Our \texttt{paDMD} method, on the other hand, works entirely in discrete-time before projecting back to the full-state and continuous-time using the DMD modes and eigenvalues.

We demonstrate \texttt{paDMD} on several representative benchmark systems in nonlinear dynamics. First, we consider a two-dimensional incompressible flow past a cylinder in two separate settings: one in which the fluid viscosity is varied, and one in which the cylinder diameter is varied, the latter constituting a geometric variation of the flow domain. Second, we consider a more challenging and strongly nonlinear system in fluid dynamics: the fluidic pinball, driven by three independently rotating cylinders, with the varying angular velocity of the three cylinders being our parameters of interest. Third, we model an electron beam oscillating under a transverse magnetic field, with variation first in the magnetic field amplitude and subsequently in both the magnetic field amplitude and the electron injection rate. Finally, we consider a more challenging kinetic plasma system: the virtual cathode oscillator, with variation in the electron injection rate. The kinetic plasma systems were simulated using a charge-conserving electromagnetic particle-in-cell (EMPIC) solver \cite{nayak2021koopman, na2016local}. In all settings, we benchmark against state-of-the-art parametric ROM methods: \emph{stacked parametric DMD} \cite{sayadi2015parametrized}, \emph{dimension reduction and interpolation in parameter space (DRIPS)} \cite{lu2023drips}, \emph{parametric operator inference (pOpInf)} \cite{mcquarrie2023nonintrusive}, and \emph{reduced Koopman operator interpolation (rKOI)} \cite{huhn2023parametric}. The results show that \texttt{paDMD} achieves improved predictive accuracy over the other interpolation-based alternatives and exhibits greater robustness to limited training data and multi-dimensional parameter spaces, all while maintaining a simple and intuitive training and evaluation mechanism.

The remainder of the paper is organized as follows. Section~\ref{sec:background} reviews Koopman operator theory and standard DMD. Section~\ref{sec:padmd} presents the theory behind the proposed \texttt{paDMD} method as well as the algorithm itself. Section~\ref{sec:results} reports numerical results on the flow past cylinder, the fluidic pinball, and the plasma systems of oscillating electron beam and virtual cathode oscillations. \ref{app:sensitivity} reports the sensitivity of every method to measurement noise and to the number of training parameter samples, and \ref{app:timing} compares the computational cost of the methods across all the systems. Finally, Section~\ref{sec:conclusion} summarizes the paper and discusses future research directions.

\section{Background} \label{sec:background}

\subsection{Koopman operator theory: an overview}\label{subsec:koopman_overview}

Consider a nonlinear dynamical system $\mathbf{f}_{\boldsymbol{\theta}}:\mathbb{X}\rightarrow\mathbb{X}$ evolving on a compact manifold $\mathbb{X} \subseteq \mathbb{R}^{n}$ and parametrized by $\boldsymbol{\theta} \in \mb{R}^p$,
\begin{equation}\label{eq:dyn_sys}
\dot{\mf{x}} = \mf{f}_{\boldsymbol{\theta}}(\mf{x}),
\end{equation}
where $\mf{x} \in \mathbb{X}$ is the state. Let $\mf{\Phi}_{\boldsymbol{\theta}}(t,\mf{x}_0)$ denote the flow map of \eqref{eq:dyn_sys} at time $t>0$, starting from an initial state $\mf{x}_0$. A measurable function $\varphi:\mb{X}\rightarrow\mathbb{C}$ is called an \emph{observable} of the system \eqref{eq:dyn_sys}, and we let $\mathcal{F}$ denote the space of all such observables. The continuous-time Koopman operator $\mathcal{K}_{\boldsymbol{\theta}}^{t}:\mathcal{F}\rightarrow\mathcal{F}$ is then defined as
\begin{equation}\label{eq:koopman_def}
(\mathcal{K}_{\boldsymbol{\theta}}^{t}\varphi)(\cdot)=\varphi\circ\mathbf{\Phi}_{\boldsymbol{\theta}}(t,\cdot),
\end{equation}
where $\circ$ denotes function composition. Unlike the original nonlinear system, which describes the time evolution of the true state $\mf{x}$, the Koopman operator acts on observables $\varphi$ and is \emph{linear} in its argument. As such, $\mathcal{K}_{\boldsymbol{\theta}}^{t}$ can be characterized by its eigenvalues and eigenfunctions. A function $\phi:\mathbb{X}\rightarrow\mathbb{C}$ is an eigenfunction of $\mathcal{K}_{\boldsymbol{\theta}}^{t}$ with corresponding eigenvalue $\lambda\in\mathbb{C}$ if
\begin{equation}\label{eq:koopman_eig}
(\mathcal{K}_{\boldsymbol{\theta}}^{t}\phi)(\cdot)=e^{\lambda t}\phi(\cdot).
\end{equation}

The infinitesimal generator of the Koopman operator, defined as $\lim_{t\to 0}\dfrac{\mathcal{K}_{\boldsymbol{\theta}}^{t}-I}{t}$ (with $I$ being the identity operator), coincides with the Lie derivative $L_{\mathbf{f}_{\boldsymbol{\theta}}}=\mathbf{f}_{\boldsymbol{\theta}}\cdot\nabla$ along $\mathbf{f}_{\boldsymbol{\theta}}$ \cite{Mauroy2016}. This generator satisfies the eigenvalue equation
\begin{equation}\label{eq:generator_eig}
L_{\mathbf{f}_{\boldsymbol{\theta}}}\phi=\lambda\phi.
\end{equation}

Consequently, the time-varying observable $\psi(t,\mathbf{x}) \triangleq \mathcal{K}_{\boldsymbol{\theta}}^{t}\varphi(\mathbf{x})$ is the solution of the partial differential equation (PDE) \cite{Mauroy2016}
\begin{equation}\label{eq:koopman_pde}
\begin{aligned}
\frac{\partial \psi}{\partial t} &= L_{\mathbf{f}_{\boldsymbol{\theta}}}\psi,\\
\psi(0,\mathbf{x}) &= \varphi(\mathbf{x}).
\end{aligned}
\end{equation}

While this linearity of the Koopman operator is a clear advantage, $\mathcal{K}_{\boldsymbol{\theta}}^{t}$ is also infinite-dimensional and admits an infinite number of eigenfunctions. In particular, if $\phi_{1}$ and $\phi_{2}$ are eigenfunctions of $\mathcal{K}_{\boldsymbol{\theta}}^{t}$ with eigenvalues $\lambda_{1}$ and $\lambda_{2}$ respectively, then $\phi_{1}^{k}\phi_{2}^{l}$ is also an eigenfunction with eigenvalue $k\lambda_{1}+l\lambda_{2}$ for all $k,l\in\mathbb{N}$. Moreover, being infinite-dimensional, the Koopman operator may also contain continuous and residual spectra with associated generalized eigendistributions \cite{Mohr2014}. In this paper, however, we restrict our attention to the point spectrum of the Koopman operator.

Let $\mathbf{g}(\cdot)\in\mathcal{F}^{p}$, $p\in\mathbb{N}$, be a vector-valued observable. It can be expressed in terms of the Koopman eigenfunctions $\phi_{i}(\cdot)$ as
\begin{equation}\label{eq:koopman_modes_expansion}
\mathbf{g}(\cdot)=\sum_{i=1}^{\infty}\phi_{i}(\cdot)\mathbf{v}_{i}^{g},
\end{equation}
where the coefficients $\mathbf{v}_{i}^{g}\in\mathbb{R}^{p}$, $i=1,2,\ldots$, are known as the \emph{Koopman modes} of $\mathbf{g}(\cdot)$. These modes represent the projection of the observable onto the space spanned by the Koopman eigenfunctions \cite{Budisic2012}. While the Koopman eigenvalues and eigenfunctions are intrinsic properties of the dynamics, the Koopman modes depend on the choice of observable.

\begin{remark}
As noted earlier, the Koopman operator, being infinite-dimensional and possibly non-compact, may contain continuous and residual spectra. The continuous spectrum corresponds to the unbounded resolvent operator $(\lambda I - \mc{K}_{\boldsymbol{\theta}}^t)^{-1}$ with a dense domain, and neither it nor the residual spectrum can be represented by a finite set of modes and eigenvalues; so neither is captured by the modal reconstruction \eqref{eq:koopman_modes_expansion} on which DMD is based. The modal reconstruction therefore represents only the point spectrum of the Koopman operator. Removing the initial transient, as we do in every system modeled in this work, reduces the influence of transient behavior, but it does not by itself establish that the underlying operator has a pure point spectrum \cite{Mauroy2016}. Where continuous or residual spectral components are present, they are not captured by a finite set of modes and may limit the accuracy of the reconstruction and of long-horizon prediction. The results reported below assess the accuracy actually obtained over the horizons considered. Note that a continuum of eigenvalues does not denote a continuous spectrum.
\end{remark}

\subsection{Approximating the Koopman Operator via DMD}\label{subsec:dmd}

Dynamic mode decomposition (DMD) is a data-driven method for approximating the Koopman operator and extracting dominant Koopman modes directly from time-series snapshot data via matrix factorization \cite{Rowley2009, schmid2010}. Unlike other methods such as EDMD and KAEs, DMD uses the identity observable (i.e., the state itself, meaning $\varphi(\cdot)$ is an identity map) and fits a best linear map that advances snapshots forward in time.

Consider state snapshots $\mathbf{x}_k\in\mathbb{R}^n$ sampled uniformly with sampling interval $\Delta t$, i.e.,
\begin{equation}\label{eq:dmd_snapshots}
\mathbf{x}_k \triangleq \mathbf{x}(t_0+k\Delta t),\qquad k=0,1,\dots,T.
\end{equation}
This is a record of $T+1$ states, giving $T$ snapshot pairs. Where we quote a number of training snapshots in Section~\ref{sec:results}, we mean the number of pairs $T$.

We concatenate them to form snapshot matrices $\mf{X}\in\mathbb{R}^{n\times T}$ and $\mf{X}^+\in\mathbb{R}^{n\times T}$, where $\mf{X}^+$ is one time snapshot ahead of $\mf{X}$:

\begin{equation}\label{eq:dmd_data_mats}
\mf{X}=\begin{bmatrix}\mathbf{x}_0 & \cdots & \mathbf{x}_{T-1}\end{bmatrix},
\qquad
\mf{X}^+=\begin{bmatrix}\mathbf{x}_1 & \cdots & \mathbf{x}_{T}\end{bmatrix}.
\end{equation}

DMD seeks to find a best-fit linear operator $\mathbf{A}\in\mathbb{R}^{n\times n}$ satisfying $\mf{X}^+\approx \mathbf{A}\mf{X}$ using least-squares (LS) approximation:
\begin{equation}\label{eq:dmd_ls}
\mf{X}^+\approx \mathbf{A}\mf{X}
\implies
\mathbf{A}
=
\arg\min_{\mc{A}\in\mathbb{R}^{n\times n}}
\left\|\mf{X}^+-\mc{A}\mf{X}\right\|_F^2 \approx \mf{X}^+ \mf{X}^{\dagger}
\end{equation}
where $\dagger$ denotes the Moore-Penrose pseudoinverse. A standard computationally efficient implementation is obtained via rank truncation. Compute the singular value decomposition (SVD) of $\mf{X}$ and retain the leading $r$ singular values:
\begin{equation}\label{eq:dmd_svd_trunc}
\mf{X} = \mf{U}\mf{\Sigma}\mf{V}^\top\approx \mathbf{U}_r\mathbf{\Sigma}_r\mathbf{V}_r^{\top},
\end{equation}
where $\mathbf{U}_r\in\mathbb{R}^{n\times r}$, $\mathbf{\Sigma}_r\in\mathbb{R}^{r\times r}$, and $\mathbf{V}_r\in\mathbb{R}^{T\times r}$.
The reduced Koopman operator $\tilde{\mathbf{A}}\in\mathbb{R}^{r\times r}$ is then approximated as
\begin{equation}\label{eq:dmd_Atilde}
\tilde{\mathbf{A}}
\triangleq
\mathbf{U}_r^{\top}\mathbf{A}\mathbf{U}_r
\approx
\mathbf{U}_r^{\top}\mf{X}^+\mathbf{V}_r\mathbf{\Sigma}_r^{-1}.
\end{equation}

We eigendecompose $\tilde{\mathbf{A}}$ as $\tilde{\mathbf{A}} \mf{W}=\mf{W}\mf{\Lambda}$, where $\mf{\Lambda}=\text{diag}(\lambda_1,~\dots,~\lambda_r)$ contains the eigenvalues of $\tilde{\mathbf{A}}$, and $\mf{W}=[\mf{w}_1,~\dots,~\mf{w}_r]$ are the corresponding eigenvectors. The full-state DMD modes $\mf \Phi$ are given by
\begin{equation}\label{eq:dmd_modes}
\mf{\Phi} \triangleq \mathbf{U}_r\mathbf{W}.
\end{equation}
Since the DMD eigenvalues we obtained are in discrete-time, we convert them into continuous-time eigenvalues by taking the matrix logarithm $\mf{\Omega} \triangleq \frac{1}{\Delta t} \ln (\mf{\Lambda})$, with $\mf{\Omega} = \text{diag}(\omega_1,~\dots,~\omega_r)$. Thus, given an initial condition $\mathbf{x}(t_0)$, the continuous-time DMD reconstruction or prediction is given by
\begin{equation}\label{eq:dmd_reconstruction_ct}
\hat{\mathbf{x}}(t)\approx
\mf{\Phi}e^{\mf{\Omega}(t-t_0)}\mf{\Phi}^\dagger \mf{x}(t_0).
\end{equation}

This provides an explicit low-dimensional linear model and modal decomposition for analyzing and reconstructing the system dynamics from data. The complete procedure is summarized in Algorithm~\ref{alg:dmd}.

\begin{algorithm}[!h] \small
\caption{DMD}
\label{alg:dmd}
\textbf{Input:} Snapshot data $\{\mathbf{x}_k\}_{k=0}^{T}$, sampling interval $\Delta t$, model rank $r$. \\
\textbf{Output:} Reduced DMD operator $\tilde{\mathbf{A}}$, DMD modes $\mathbf{\Phi}$, and continuous-time eigenvalues $\mathbf{\Omega}$.
\begin{algorithmic}[1]
\State $\mathbf{X} \gets [\mathbf{x}_0~\cdots~\mathbf{x}_{T-1}]$, $\quad\mathbf{X}^+ \gets [\mathbf{x}_1~\cdots~\mathbf{x}_{T}]$ \Comment{DMD regression matrices}
\Procedure{\textsc{Compute \texttt{DMD}}}{$\mathbf{X}^+, \mathbf{X}, r, \Delta t$}
    \State $\mathbf{U}_r,\mathbf{\Sigma}_r,\mathbf{V}_r \gets \text{SVD}(\mathbf{X}, r)$ \Comment{Truncated rank-$r$ SVD of $\mathbf{X}$}
    \State $\tilde{\mathbf{A}} \gets \mathbf{U}_r^\top \mathbf{X}^+ \mathbf{V}_r \mathbf{\Sigma}_r^{-1}$ \Comment{Reduced DMD/Koopman operator}
    \State $\mathbf{W},\mathbf{\Lambda} \gets \text{EIG}(\tilde{\mathbf{A}})$ \Comment{Eigendecomposition of $\tilde{\mathbf{A}}$}
    \State $\mathbf{\Phi} \gets \mf U_r\mathbf{W}$ \Comment{Full-state DMD modes}
    \State $\mathbf{\Omega} \gets \frac{1}{\Delta t}\ln(\mathbf{\Lambda})$ \Comment{Continuous-time eigenvalues}
\EndProcedure
\end{algorithmic}
\end{algorithm}

To better capture more complex nonlinear dynamics where the state snapshots alone may not span a Koopman invariant subspace, one may employ \emph{Hankel DMD} which constructs a vector-valued observable using delay coordinates. Specifically, for a chosen number of delay stacks $d\in\mathbb{N}$ and sampling interval $\Delta t$, we define the delay-embedded observable
$\mathbf{z}_k
\triangleq
\begin{bmatrix}
\mathbf{x}_k^\top,~
\mathbf{x}_{k-1}^\top,~
\cdots,~
\mathbf{x}_{k-(d-1)}^\top
\end{bmatrix}^\top
\in \mathbb{R}^{dn},
$
where the stacking is vertical and indices are assumed valid (i.e., $k\ge d-1$). Hankel DMD then applies the standard DMD algorithm of \eqref{eq:dmd_data_mats}--\eqref{eq:dmd_reconstruction_ct} to snapshot matrices formed from the delay-embedded data.

This time-delay lifting is motivated by results from Takens' embedding theorem \cite{takens2006detecting, brunton2017chaos}, and often yields improved modal representations for complex nonlinear systems when the original state snapshots are insufficient to linearly represent the dominant dynamics.
To obtain the reconstruction $\hat{\mathbf{x}}(t)$ of the original state, we evaluate \eqref{eq:dmd_reconstruction_ct} in the delay coordinates, with the modal amplitudes obtained from the full time-delayed initial condition $\mf z_0$, and then extract the first $n$ entries of the result.

\begin{remark}\label{rem:dmd_rank}
The model rank $r$ in \eqref{eq:dmd_svd_trunc} decides how many of the dominant singular values of $\mf{X}$ are retained. A common choice for selecting $r$ is the energy criterion
\begin{equation}\label{eq:energy_criterion}
r = \min\left\{k ~:~ \frac{\sum_{j=1}^{k}\sigma_j^2}{\sum_{j}\sigma_j^2} \ge 1-\epsilon\right\},
\end{equation}
where $\sigma_j$ are the diagonal entries of $\mf{\Sigma}$ and $\epsilon$ is a small user-supplied quantity. Since no single value of $\epsilon$ suits every problem, in practice $r$ is tuned and chosen heuristically for each problem by inspecting the decay of the singular values and the resulting prediction accuracy. Too small a rank discards dynamically relevant modes, while too large a rank introduces spurious, unstable eigenvalues over long prediction horizons and results in overfitting.
\end{remark}

\begin{remark}\label{rem:tikhonov}
To further ensure stable predictions, we can use ridge regression/Tikhonov regularization \cite{golub1999tikhonov}. Here, the matrix $\mathbf{A}$ can be obtained by solving the following regularized least-squares problem:
\begin{equation}\label{eq:dmd_ls_tik}
\mathbf{A}=\arg \min_{\mc{A}\in \mathbb R^{n\times n}}
\left\|\mathbf{X}^+ - \mc{A} \mf X\right\|_F^2
+ \lambda\left\|\mc{A} \right\|_F^2,
\end{equation}
where $\lambda > 0$ is a regularization parameter and $\|\cdot\|_F$ denotes the Frobenius norm. This leads to the closed-form solution
\begin{equation}\label{eq:dmd_closed_form_tik}
\mathbf{A}
= \mathbf{X}^+\mf X^\top\left(\mf X \mf X^\top + \lambda\mathbf{I}\right)^{-1}.
\end{equation}
Equivalently, in terms of the SVD of $\mf X$, this replaces each inverse singular value $1/\sigma_j$ by the filter $\sigma_j/(\sigma_j^{2}+\lambda)$, which leaves the dominant directions essentially unchanged while damping the smallest retained ones. The regularization term can improve numerical stability, and $\lambda$ is a tunable hyperparameter. The rest of the DMD workflow can be adapted as shown in Algorithm~\ref{alg:dmd}.
\end{remark}

\section{Parameter-affine dynamic mode decomposition ({\texttt{paDMD}})}
\label{sec:padmd}

The standard DMD procedure above constructs an approximation of the finite-dimensional Koopman operator from snapshot data generated by a dynamical system with a \emph{fixed} parameter $\boldsymbol{\theta}$. This requires new training snapshot data at each parameter value of interest in order to construct a corresponding DMD model, which may be computationally expensive or infeasible when high-fidelity simulations or experiments are involved. This motivates \emph{parametric DMD} methods, whose goal is to leverage snapshot data collected at a finite set of training parameter samples to efficiently construct DMD-based reduced-order models that generalize across the parameter space.

\subsection{Parametric modeling for parameter-affine systems}\label{subsec:param_affine_interp}

Consider a dynamical system parametrized by $\boldsymbol{\theta}$ in an affine form
\begin{equation}\label{eq:param_affine_dyn}
\dot{\mathbf{x}}=\mathbf{f}_{\boldsymbol{\theta}}(\mathbf{x})=\mathbf{f}(\mathbf{x})+\sum_{i=1}^{m}\mathbf{g}_i(\mathbf{x})\,h_i(\boldsymbol{\theta}),
\qquad
\mathbf{x}(t_0)=\mathbf{x}_0,
\end{equation}
where $\mathbf{x}\in\mathbb{X}\subseteq\mathbb{R}^{n}$, $\boldsymbol{\theta}\in\mathbb{R}^{p}$, and the parametrization functions $h_i:\mathbb{R}^{p}\rightarrow\mathbb{R}$ for $i=1,\ldots,m$ enter the dynamics affinely through the vector fields $\mathbf{g}_i$. Applying \eqref{eq:koopman_pde} to \eqref{eq:param_affine_dyn} yields the evolution of the following PDE \cite{Goswami2017}:
\begin{equation}\label{eq:param_forced_pde}
\begin{aligned}
\frac{\partial \psi}{\partial t} &= L_{\mathbf{f}}\psi+\sum_{i=1}^{m} h_i(\boldsymbol{\theta})\,L_{\mathbf{g}_i}\psi,\\
\psi(t_0,\mathbf{x}) &= \varphi(\mathbf{x}),
\end{aligned}
\end{equation}
where $L_{\mathbf{g}_i}\triangleq \mathbf{g}_i\cdot\nabla$ are the Lie derivatives along $\mf{g}_i$ for $i=1,\ldots,m$, and are therefore linear operators on the space of $\psi$. The parametrized forced PDE \eqref{eq:param_forced_pde} provides a representation of the parametric system \eqref{eq:param_affine_dyn} in a lifted latent space that is bilinear with respect to the parametric functions $h_i(\boldsymbol{\theta})$.

Let $\mathcal{F}$ be the closure of the span of a countable collection of real-valued observables $\{\varphi_1(\cdot), \ \varphi_2(\cdot), \ \ldots \ \}$. We define the lifted state $\mathbf{z}=\varphi(\mathbf{x})\in\mathcal{Z}\subseteq\mathbb{R}^{N}$, with $\varphi(\mathbf{x})\triangleq\begin{bmatrix}\varphi_1(\mathbf{x}) \quad\varphi_2(\mathbf{x}) \ \ldots \ \end{bmatrix}^{\top}$ and $N$ being its cardinality. Further, assume there exists a nonlinear reconstruction map $\varphi^{-1}:\mathcal{Z}\rightarrow\mathbb{X}$ satisfying $\varphi^{-1}(\mathbf{z})=\mathbf{x}$. The main result is summarized in the following theorem.

\begin{theorem}\label{thm:bilinear_lift}
Assume that $\mathcal{F}$ is invariant under $L_{\mathbf{f}}$ and $\{L_{\mathbf{g}_i}\}_{i=1}^{m}$, i.e., $L_{\mathbf{f}}\mathcal{F}\subset\mathcal{F}$ and $L_{\mathbf{g}_i}\mathcal{F}\subset\mathcal{F}$ for all $i\in\{1,\ldots,m\}$. Then the lifted dynamics admit a bilinear representation:
\begin{equation}\label{eq:bilinear_lift}
\dot{\mathbf{z}}
=
\left.\frac{\partial \varphi(\mathbf{x})}{\partial \mathbf{x}}
\left[
\mathbf{f}(\mathbf{x})+\sum_{i=1}^{m}\mathbf{g}_i(\mathbf{x})\,h_i(\boldsymbol{\theta})
\right]\right|_{\mathbf{x}=\varphi^{-1}(\mathbf{z})}
=
\mathbf{A}\mathbf{z}+\sum_{i=1}^{m} h_i(\boldsymbol{\theta})\,\mathbf{B}_i\mathbf{z},
\end{equation}
where
$\mathbf{A}\mathbf{z}\triangleq \left.L_{\mathbf{f}}\varphi(\mathbf{x})\right|_{\mathbf{x}=\varphi^{-1}(\mathbf{z})}$
and
$\mathbf{B}_i\mathbf{z}\triangleq \left.L_{\mathbf{g}_i}\varphi(\mathbf{x})\right|_{\mathbf{x}=\varphi^{-1}(\mathbf{z})}$.
\end{theorem}

\begin{proof}
Differentiating $\mathbf{z}=\varphi(\mathbf{x})$ along the flow map of \eqref{eq:param_affine_dyn} gives
\begin{equation*}
\dot{\mf z} = \frac{\partial\varphi(\mf x)}{\partial\mf x}\dot{\mf x}= \frac{\partial\varphi(\mf x)}{\partial\mf x}\Big[\mf f(\mf x)+\textstyle\sum_i \mf g_i(\mf x)\,h_i(\bs\theta)\Big]
= L_{\mf f}\varphi(\mf x)+\textstyle\sum_i h_i(\bs\theta)\,L_{\mf g_i}\varphi(\mf x),
\end{equation*}
where the $h_i(\boldsymbol{\theta})$ factor out of the time derivative because they do not vary with time along the trajectory. Neither $L_{\mathbf f}$ nor any $\{L_{\mathbf g_i}\}_{i=1}^m$ depends on $\boldsymbol\theta$, since the parameter enters \eqref{eq:param_affine_dyn} only through the scalars $h_i(\boldsymbol\theta)$. By the invariance assumption, $L_{\mathbf{f}}\varphi \in \mc F$ and $L_{\mathbf{g}_i}\varphi \in \mc F, ~\forall i \in \{1,\dots,m\}$. Thus, there exist constant operators $\mf A$ and $\{\mf B_i\}_{i=1}^m$, independent of $\boldsymbol{\theta}$, such that $L_{\mf f}\varphi=\mf A\varphi$ and $L_{\mf g_i}\varphi=\mf B_i\varphi$. Evaluating at $\mathbf{x}=\varphi^{-1}(\mathbf{z})$ and substituting $\varphi(\mathbf{x})=\mathbf{z}$ yields \eqref{eq:bilinear_lift}. This argument follows that of Theorem 1 in \cite{Goswami2021} with the control inputs $u_i(t)$ replaced by the parametrization functions $h_i(\boldsymbol{\theta})$. We further note that a time-varying $u_i(t)$ renders the lifted operator $\mathbf A+\sum_i u_i(t)\mathbf B_i$ time-varying, whereas a fixed $\boldsymbol{\theta}$ renders it time-invariant.
\end{proof}

We thus obtain a parametric model $\dot{\mf z} = \mf K^{(\boldsymbol{\theta})} \mf z$, where $\mf K^{(\boldsymbol{\theta})} \triangleq \mathbf{A}+\sum_{i=1}^{m} h_i(\boldsymbol{\theta})\,\mathbf{B}_i$ is the parameter-affine Koopman generator. The corresponding Koopman operator over a sampling interval $\Delta t$ is its exponential $e^{\mf K^{(\boldsymbol{\theta})}\Delta t}$, and it is this quantity that is approximated by the discrete-time surrogate model identified below.

\begin{remark}\label{rem:truncation}
The map $\varphi$ is assumed to be injective. The map $\varphi^{-1}$ is a minimal representation that uses only the essential components of $\mathbf{z}_k$ to reconstruct $\mathbf{x}_k$. As discussed earlier, in standard DMD we have $\mf{z}_k = \varphi(\mf{x}_k)=\mf{x}_k$, so that $\varphi=\varphi^{-1}=I$. In Hankel DMD, $\mathbf{z}_k$ is the delay-embedded state $\begin{bmatrix}
\mathbf{x}_k^\top,~
\mathbf{x}_{k-1}^\top,~
\cdots,~
\mathbf{x}_{k-(d-1)}^\top
\end{bmatrix}^\top$, and $\varphi^{-1}$ simply extracts its first $n$ components. When EDMD is used, in practice $\mathbf{z}$ is chosen as $\mathbf{z}=\{\varphi_1(\mathbf{x}),\ldots,\varphi_N(\mathbf{x})\}$ with $N\gg n$, ideally ensuring vanishing residuals. Here, $\mf x_k$ is often included in $\mf z_k$, which preserves injectivity since $\varphi^{-1}$ reduces to the linear projection onto those components. Establishing the invariance of this finite-dimensional functional space $\mathrm{span}\{\varphi_1(\cdot),\ldots,\varphi_N(\cdot)\}$ is considerably more challenging and is typically done using neural networks.
\end{remark}

\subsection{Learning the parameter-affine DMD operator from data}\label{subsec:padmd_desc}

Here we present \emph{parameter-affine dynamic mode decomposition (\texttt{paDMD})} to estimate a parameter-dependent finite-dimensional DMD/Koopman operator from finite samples of \emph{state} and \emph{parameter} data.
We consider the parameter-affine continuous-time system introduced in Section~\ref{subsec:param_affine_interp}, where the parameter enters through known functions $\{h_i(\boldsymbol{\theta})\}_{i=1}^{m}$.
Since identifying continuous-time generators typically requires derivative information, which is difficult to obtain in practice, we instead learn a discrete-time predictor from snapshot pairs while preserving an affine-in-parameter structure.

Over a sampling interval $\Delta t$, the Koopman bilinear dynamics satisfies
\begin{equation}\label{eq:padmd_flow}
\mathbf{z}(t+\Delta t)
= e^{\big(\mathbf{A}+\sum_{i=1}^{m} h_i(\boldsymbol{\theta})\,\mathbf{B}_{i}\big)\Delta t}\mathbf{z}(t).
\end{equation}

A first-order Euler discretization yields a parameter-affine discrete-time approximation
\begin{equation}\label{eq:padmd_euler}
\mathbf{z}(t+\Delta t)
\approx
\Big(\tilde{\mathbf{A}}+\sum_{i=1}^{m} h_i(\boldsymbol{\theta})\,\tilde{\mathbf{B}}_i\Big)\mathbf{z}(t),
\quad
\tilde{\mathbf{A}}\triangleq \mathbf{I}+\mathbf{A}\Delta t,\quad
\tilde{\mathbf{B}}_i\triangleq \mathbf{B}_{i}\Delta t.
\end{equation}

This provides a framework for estimating $(\tilde{\mathbf{A}},\{\tilde{\mathbf{B}}_i\}_{i=1}^{m})$ directly from data using least-squares approximation, assuming the $h_i(\boldsymbol{\theta})$ functions are known.

We assume that we are given $L$ training trajectories corresponding to parameter samples $\{\boldsymbol{\theta}^{(\ell)}\}_{\ell=1}^{L}$, with uniformly sampled snapshots
\begin{equation}\label{eq:padmd_snapshots}
\mathbf{x}_k^{(\ell)} \triangleq \mathbf{x}^{(\ell)}(t_0+k\Delta t),
\qquad k=0,1,\dots,T,
\qquad \ell=1,\dots,L.
\end{equation}

Let $\mathbf{h}(\boldsymbol{\theta})\triangleq \begin{bmatrix}h_1(\boldsymbol{\theta}) & \cdots & h_m(\boldsymbol{\theta})\end{bmatrix}^{\top}\in\mathbb{R}^{m}$.
To encode the parameter-affine structure linearly in the \texttt{paDMD} least-squares regression, we introduce
\begin{equation}\label{eq:padmd_kron_lift}
\boldsymbol{\psi}(\mathbf{x},\boldsymbol{\theta})
\triangleq
\begin{bmatrix}
1\\
\mathbf{h}(\boldsymbol{\theta})
\end{bmatrix}\otimes\mf x
=
\begin{bmatrix}\mf{x}\\h_1(\boldsymbol{\theta})\mathbf{x}\\ \vdots\\ h_m(\boldsymbol{\theta})\mathbf{x}\end{bmatrix}
\in\mathbb{R}^{n+mn}.
\end{equation}

For each trajectory $\{\mf x_k^{(\ell)}\}_{k=0}^T$, we define the snapshot matrices
\begin{equation*}
\begin{aligned}
\mathbf{\Psi}^{(\ell)}
&\triangleq
\begin{bmatrix}
\boldsymbol{\psi}(\mathbf{x}_0^{(\ell)},\boldsymbol{\theta}^{(\ell)}) &
\boldsymbol{\psi}(\mathbf{x}_1^{(\ell)},\boldsymbol{\theta}^{(\ell)}) &
\cdots &
\boldsymbol{\psi}(\mathbf{x}_{T-1}^{(\ell)},\boldsymbol{\theta}^{(\ell)})
\end{bmatrix}
\in\mathbb{R}^{(n+mn)\times T},\\
{\mathbf{X}^+}^{(\ell)}
&\triangleq
\begin{bmatrix}
\mathbf{x}_1^{(\ell)} & \cdots & \mathbf{x}_{T}^{(\ell)}
\end{bmatrix}
\in\mathbb{R}^{n\times T},
\end{aligned}
\end{equation*}
and concatenate them to form
\begin{equation*}
\mathbf{\Psi}
\triangleq
\begin{bmatrix}
\mathbf{\Psi}^{(1)} & \cdots & \mathbf{\Psi}^{(L)}
\end{bmatrix}
\in\mathbb{R}^{(n+mn)\times LT}
~\text{and}~
\mathbf{X}^+
\triangleq
\begin{bmatrix}
{\mathbf{X}^+}^{(1)} & \cdots & {\mathbf{X}^+}^{(L)}
\end{bmatrix}\in\mathbb{R}^{n\times LT}.
\end{equation*}

We seek matrices $\tilde{\mathbf{A}}\in\mathbb{R}^{n\times n}$ and $\tilde{\mathbf{B}} \triangleq \begin{bmatrix} \tilde{\mathbf{B}}_1 & \tilde{\mathbf{B}}_2 & \ldots & \tilde{\mathbf{B}}_m \end{bmatrix}\in\mathbb{R}^{n\times mn}$, where $\tilde{\mf{B}}_i \in \mb{R}^{n \times n}, ~i \in \{ 1,\ldots,m\}$, such that $\mathbf{X}^+
\approx
\begin{bmatrix}
\tilde{\mathbf{A}} & \tilde{\mathbf{B}}
\end{bmatrix}
\mathbf{\Psi}$,
i.e.,
\begin{equation}\label{eq:padmd_ls}
\begin{bmatrix}
\tilde{\mathbf{A}} & \tilde{\mathbf{B}}
\end{bmatrix}
=
\arg \min_{\substack{\mc{A}\in\mathbb{R}^{n\times n}\\ \mc{B}\in\mathbb{R}^{n\times mn}}}
\left\|
\mathbf{X}^+ -
\begin{bmatrix}
\mc{A} & \mc{B}
\end{bmatrix}
\mathbf{\Psi}
\right\|_F^2
\approx
\mf{X}^+ \mf{\Psi}^{\dagger}
\end{equation}

Persistency of excitation in the parametric functions is necessary to learn the operators $(\tilde{\mathbf A},\{\tilde{\mathbf B}_i\}_{i=1}^{m})$ accurately. Since $\mathbf h(\boldsymbol\theta^{(\ell)})$ is constant along the trajectory $\{\mathbf x_k^{(\ell)}\}_{k=0}^{T}$, the data constrain these operators only through the vectors $[1~h_1(\boldsymbol\theta^{(\ell)})~\dots~h_m(\boldsymbol\theta^{(\ell)})]^{\top}$, $\ell=1,\dots,L$. Let these vectors form the rows of $\mf H\in\mathbb{R}^{L\times(m+1)}$. Separately identifying $(\tilde{\mathbf A},\{\tilde{\mathbf B}_i\}_{i=1}^{m})$ requires $\mathrm{rank}(\mf H)=m+1$, regardless of the number of snapshots $T$ contained in each trajectory. This is a condition on the parameter design alone, and it is necessary rather than sufficient, as the training trajectories must also excite the state directions retained by the truncation. Because it is a rank condition, it does not favor any particular arrangement of samples, and both structured and scattered designs can satisfy it. In our experiments we spread the training samples across the range of each parametric function. This is an experimental-design choice intended to improve parameter coverage, and we verify that $\mathrm{rank}(\mf H)=m+1$ holds for every training set used. We follow this in all the experiments reported in Section~\ref{sec:results}, \ref{app:sensitivity} and \ref{app:timing}.

\begin{remark}\label{rem:padmd_known_h}
The \texttt{paDMD} method is physics-informed in a specific and limited sense: we only assume the parametrization functions $\{h_i(\boldsymbol{\theta})\}_{i=1}^{m}$, which may be nonlinear in $\boldsymbol{\theta}$, are known \emph{a priori}. The vector fields $\mathbf{f}$ and $\mathbf{g}_i$ in \eqref{eq:param_affine_dyn} remain unknown, so the Koopman generators, DMD modes and eigenvalues are all learned from data. This is often natural in physics-driven settings such as fluid dynamics or plasma science where the physical role of each parameter is known even when the governing equations are not solved directly \cite{otto2019linearly, nayak2021koopman}.
\end{remark}

\begin{remark}\label{rem:padmd_normalization}
Normalizing the state or input variables before the regression step is common practice in DMD and other data-driven ROM methods. It is especially important for parametric ROMs, which often combine quantities of very different physical origin and magnitude in a single regression. Without normalization, the fit is dominated by whichever term is largest in magnitude rather than by its influence on the dynamics. In \texttt{paDMD}, we therefore normalize each parametrization function $h_i(\boldsymbol{\theta})$ to the range $[0,c]$ before the least-squares approximation. The value of $c$ sets how strongly the parametric variation is expressed in $\bs\Psi$, and thus how well the parameter-dependent operators $\{\tilde{\mf B}_i\}$ are recovered by the least-squares fit. It therefore carries the persistency of excitation in the training parameter samples into the learned operators. Because each problem has different parameters influencing the dynamics in different ways, $c$ is chosen heuristically for each problem, and can be regarded as a tunable hyperparameter.
\end{remark}

For computational efficiency, we use a rank-$\tilde r$ truncated SVD of $\mathbf{\Psi}$ (by retaining the first $\tilde{r}$ largest singular values of $\mf{\Psi}$) to compute \eqref{eq:padmd_ls}, i.e.,
\begin{equation*}
\mathbf{\Psi}=\tilde{\mathbf{U}} \tilde{\mathbf{\Sigma}} \tilde{\mathbf{V}}^{\top}
\approx
\tilde{\mathbf{U}}_r\tilde{\mathbf{\Sigma}}_r\tilde{\mathbf{V}}_r^{\top}.
\end{equation*}

We partition $\tilde{\mathbf{U}}_r$ into its first $n$ rows and remaining $mn$ rows as
\begin{equation*}
\tilde{\mathbf{U}}_r=
\begin{bmatrix}
\tilde{\mathbf{U}}_0\\
\tilde{\mathbf{U}}_h
\end{bmatrix},
~
\tilde{\mathbf{U}}_0\in\mathbb{R}^{n\times \tilde r},
~
\tilde{\mathbf{U}}_h\in\mathbb{R}^{mn\times \tilde r}.
\end{equation*}

Then the least-squares solution of (\ref{eq:padmd_ls}) is obtained as
\begin{equation}\label{eq:padmd_closed_form_A_B}
\tilde{\mathbf{A}}
=
\mathbf{X}^+\,\tilde{\mathbf{V}}_r\tilde{\mathbf{\Sigma}}_r^{-1}\tilde{\mathbf{U}}_0^{\top} \in \mb{R}^{n\times n},
\qquad
\tilde{\mathbf{B}}
=
\mathbf{X}^+\,\tilde{\mathbf{V}}_r \tilde{\mathbf{\Sigma}}_r^{-1}\tilde{\mathbf{U}}_h^{\top}\in \mb{R}^{n\times mn}.
\end{equation}

Substituting $\tilde{\mathbf{B}}
=
\begin{bmatrix}
\tilde{\mathbf{B}}_1 & \tilde{\mathbf{B}}_2 & \cdots & \tilde{\mathbf{B}}_m
\end{bmatrix}, ~\tilde{\mf{B}}_i \in \mb{R}^{n \times n}, ~i=1,\dots,m$, yields the estimate of the parametrized discrete-time \texttt{paDMD} operator:
\begin{equation}\label{eq:padmd_K_theta}
\tilde{\mathbf{K}}^{(\boldsymbol{\theta})}
\triangleq
\tilde{\mathbf{A}}+\sum_{i=1}^{m} h_i(\boldsymbol{\theta})\,\tilde{\mathbf{B}}_i \in \mb{R}^{n \times n}.
\end{equation}

Thus, $\tilde{\mathbf{K}}^{(\boldsymbol{\theta})}$ provides a finite-dimensional estimate of the Koopman operator $e^{\mf K^{(\boldsymbol{\theta})}\Delta t}$ over the sampling interval $\Delta t$ and can be evaluated directly at new values of $\boldsymbol{\theta}$.

We further compute a low-rank representation of $\tilde{\mathbf{K}}^{(\boldsymbol{\theta})}$ for prediction. We perform a truncated SVD of $\mathbf{X^+}$ by retaining the first $\hat{r}$ largest singular values, i.e.,
\begin{equation*}
\mathbf{X}^+=\hat{\mathbf{U}}\hat{\mathbf{\Sigma}}\hat{\mathbf{V}}^{\top}
\approx
\hat{\mf U}_{r}\hat{\mf \Sigma}_{r}\hat{\mf V}_{r}^{\top}
\end{equation*}
and obtain the reduced-order parametric operator as follows:
\begin{equation}
\label{eq:padmd_rom}
\hat{\mf K}_{r}^{(\boldsymbol{\theta})}
\triangleq
\hat{\mf U}_{r}^{\top}\tilde{\mathbf{K}}^{(\boldsymbol{\theta})} \hat{\mf U}_{r}
\in\mathbb{R}^{\hat r\times \hat r}.
\end{equation}

We eigendecompose $\hat{\mf K}_{r}^{(\boldsymbol{\theta})}$ as $\hat{\mf K}_{r}^{(\boldsymbol{\theta})}\mathbf{W}^{(\boldsymbol{\theta})}=\mathbf{W}^{(\boldsymbol{\theta})}\mathbf{\Lambda}^{(\boldsymbol{\theta})}$, where $\mathbf{\Lambda}^{(\boldsymbol{\theta})}=\text{diag}(\lambda_1^{(\boldsymbol{\theta})}, ~\lambda_2^{(\boldsymbol{\theta})},\dots,~\lambda_{\hat r}^{(\boldsymbol{\theta})})$ contains the eigenvalues of $\hat{\mathbf{K}}_{r}^{(\boldsymbol{\theta})}$, and $\mathbf{W}^{(\boldsymbol{\theta})}$ contains the corresponding eigenvectors, with $\mathbf{W}^{(\boldsymbol{\theta})}\triangleq [\mf{w}_1^{(\boldsymbol{\theta})}, ~\dots,~\mf{w}_{\hat r}^{(\boldsymbol{\theta})}]$. Thus, the full-state \texttt{paDMD} modes are
\begin{equation*}
\mf{\Phi}^{(\boldsymbol{\theta})}=\hat{\mf U}_{r}\mathbf{W}^{(\boldsymbol{\theta})}.
\end{equation*}

Since the \texttt{paDMD} eigenvalues we obtained are in discrete time, the continuous-time eigenvalues are obtained by taking the matrix logarithm
$\mf{\Omega}^{(\boldsymbol{\theta})}=\frac{1}{\Delta t} \ln({\mf{\Lambda}}^{(\boldsymbol{\theta})})=\text{diag}(\omega_1^{(\boldsymbol{\theta})},~\dots,~\omega_{\hat r}^{(\boldsymbol{\theta})})$. Hence, given an initial condition $\mathbf{x}(t_0)$, the continuous-time, full-state \texttt{paDMD} reconstruction or prediction of the system for any desired parameter $\boldsymbol{\theta}^*$ at time $t$ is
\begin{equation}\label{eq:padmd_reconstruction}
\hat{\mathbf{x}}(t)\approx
\boldsymbol{\Phi}^{(\boldsymbol{\theta}^*)}e^{\mathbf{\Omega}^{(\boldsymbol{\theta}^*)}(t-t_0)}{\boldsymbol{\Phi}^{(\boldsymbol{\theta}^*)}}^\dagger \mathbf{x}(t_0).
\end{equation}

The complete \texttt{paDMD} procedure is summarized in Algorithm~\ref{alg:padmd}, which parallels the standard DMD workflow of Algorithm~\ref{alg:dmd}.

\begin{algorithm}[!h] \small
\caption{{\texttt{paDMD}}}
\label{alg:padmd}
\textbf{Input:} Training data $\{(\boldsymbol{\theta}^{(\ell)},\{\mathbf{x}^{(\ell)}_k\}_{k=0}^{T})\}_{\ell=1}^{L}$, sampling interval $\Delta t$, functions $\mathbf{h}(\boldsymbol{\theta})=[h_1(\boldsymbol{\theta}),~\dots,h_m(\boldsymbol{\theta})]^\top$, model ranks $\tilde r$ and $\hat r$, target parameters $\{\boldsymbol{\theta}^{*(j)}\}_{j=1}^{J}$. \\
\textbf{Output:} Estimates $\tilde{\mf{A}}$, $\{\tilde{\mathbf{B}}_i\}_{i=1}^{m}$, $\tilde{\mathbf{K}}^{(\boldsymbol{\theta}^{*(j)})}$, $\hat{\mf{K}}_{r}^{(\boldsymbol{\theta}^{*(j)})}$, \texttt{paDMD} modes $\mf{\Phi}^{(\boldsymbol{\theta}^{*(j)})}$, and continuous-time eigenvalues $\mf{\Omega}^{(\boldsymbol{\theta}^{*(j)})}$, for all $j=1,\dots,J$.
\begin{algorithmic}[1]
\For{$\ell=1,\dots,L$}
    \State $\mathbf{\Psi}^{(\ell)} \gets [\boldsymbol{\psi}(\mathbf{x}^{(\ell)}_0,\boldsymbol{\theta}^{(\ell)}),\dots,\boldsymbol{\psi}(\mathbf{x}^{(\ell)}_{T-1},\boldsymbol{\theta}^{(\ell)})]$. \Comment{$\boldsymbol{\psi}(\mathbf{x},\boldsymbol{\theta})
\triangleq
\begin{bmatrix}
1\\
\mathbf{h}(\boldsymbol{\theta})
\end{bmatrix}\otimes\mf x$}
    \State ${\mathbf{X}^+}^{(\ell)} \gets [\mathbf{x}^{(\ell)}_1,\dots,\mathbf{x}^{(\ell)}_{T}]$
\EndFor
\State $\mathbf{\Psi}\gets [\mathbf{\Psi}^{(1)}~\cdots~\mathbf{\Psi}^{(L)}]$, ${\mathbf{X}^+} \gets [{\mathbf{X}^+}^{(1)}~\cdots~{\mathbf{X}^+}^{(L)}]$
\Comment{\texttt{paDMD} regression matrices}
\Procedure{\textsc{Compute \texttt{paDMD}}}{$\mathbf{X^+},\mathbf{\Psi},\tilde r,\hat r,\Delta t,\mathbf{h}(\boldsymbol{\theta}),\{\boldsymbol{\theta}^{*(j)}\}_{j=1}^{J}$}
\State $\tilde{\mathbf{U}}_r, \tilde{\mathbf{\Sigma}}_r, \tilde{\mathbf{V}}_r \gets \text{SVD}(\mf{\Psi}, \tilde r)$ \Comment{Truncated rank-$\tilde r$ SVD of $\mf{\Psi}$}
\State $
    \begin{bmatrix} \tilde{\mf U}_{0}\\ \tilde{\mf U}_{h} \end{bmatrix} \gets \tilde{\mf U}_{r}$\Comment{Partitioning $\tilde{\mf U}_{r}$}
\State $\tilde{\mathbf{A}}\gets \mathbf{X}^+\tilde{\mathbf{V}}_r\tilde{\mathbf{\Sigma}}_r^{-1}\tilde{\mathbf{U}}_0^\top$ \Comment{Estimate of $\mf{A}$}
\State $\tilde{\mathbf{B}}\gets \mathbf{X}^+\tilde{\mathbf{V}}_r\tilde{\mathbf{\Sigma}}_r^{-1}\tilde{\mathbf{U}}_h^\top$, $[\tilde{\mathbf{B}}_1,~\cdots,~\tilde{\mathbf{B}}_m] \gets \tilde{\mathbf{B}}$\Comment{Estimates of $\{\mathbf{B}_i\}_{i=1}^{m}$}
\State $\hat{\mf U}_{r}, \hat{\mf \Sigma}_{r},\hat{\mf V}_{r} \gets \text{SVD}(\mf{X}^+, \hat r)$ \Comment{Truncated rank-$\hat r$ SVD of $\mf{X}^+$}
\For{$j=1,\dots,J$}
    \State $\tilde{\mathbf{K}}^{(\boldsymbol{\theta}^{*(j)})}\gets \tilde{\mathbf{A}}+\sum_{i=1}^{m} h_i(\boldsymbol{\theta}^{*(j)})\tilde{\mathbf{B}}_i$ \Comment{Estimate of $\mf{K}^{(\boldsymbol{\theta}^{*(j)})}$}
    \State $\hat{\mf K}_{r}^{(\boldsymbol{\theta}^{*(j)})}\gets \hat{\mf U}_{r}^\top \tilde{\mf K}^{(\boldsymbol{\theta}^{*(j)})}\hat{\mf U}_{r}$ \Comment{Low-rank approximation of $\mf{K}^{(\boldsymbol{\theta}^{*(j)})}$}
    \State $\mathbf{W}^{(\boldsymbol{\theta}^{*(j)})},~\mathbf{\Lambda}^{(\boldsymbol{\theta}^{*(j)})} \gets \text{EIG}(\hat{\mf K}_{r}^{(\boldsymbol{\theta}^{*(j)})})$ \Comment{Eigendecomposition of $\hat{\mf K}_{r}^{(\boldsymbol{\theta}^{*(j)})}$}
    \State $\boldsymbol{\Phi}^{(\boldsymbol{\theta}^{*(j)})}\gets \hat{\mathbf{U}}_{r}\mathbf{W}^{(\boldsymbol{\theta}^{*(j)})}$ \Comment{Full-state \texttt{paDMD} modes}
    \State $\mathbf{\Omega}^{(\boldsymbol{\theta}^{*(j)})}\gets \frac{1}{\Delta t}\ln(\mathbf{\Lambda}^{(\boldsymbol{\theta}^{*(j)})})$ \Comment{Continuous-time eigenvalues}
\EndFor
\EndProcedure
\end{algorithmic}
\end{algorithm}

To capture more strongly nonlinear dynamics, as in DMD, we can also use a Hankel variant of \texttt{paDMD} by replacing $\mathbf{x}_k$ with the delay-embedded state $\mf z_k=\begin{bmatrix}
\mathbf{x}_k^\top,~
\mathbf{x}_{k-1}^\top,~
\cdots,~
\mathbf{x}_{k-(d-1)}^\top
\end{bmatrix}^\top$ in $\boldsymbol{\psi}(\cdot,\boldsymbol{\theta})$, so that $\boldsymbol{\psi}(\mf z_k,\boldsymbol{\theta})=[\mf z_k^\top,~h_1(\boldsymbol{\theta})\mf z_k^\top,~\dots,~h_m(\boldsymbol{\theta}) \mf z_k^\top]^\top$. The remainder of the \texttt{paDMD} workflow is unchanged, and the predicted state $\hat{\mathbf{x}}(t)$ is recovered as the first $n$ entries of the predicted delay-embedded state. Tikhonov regularization can also be applied to \texttt{paDMD} exactly as in Remark~\ref{rem:tikhonov}, with $\mf\Psi$ in place of $\mf X$, giving $\begin{bmatrix}\tilde{\mf A} & \tilde{\mf B}\end{bmatrix}=\mf X^{+}\mf\Psi^{\top}\big(\mf\Psi\mf\Psi^{\top}+\lambda\mf I\big)^{-1}$, which can improve numerical stability. The parameter $\lambda$ is selected by minimizing a multi-step reconstruction error over the training parameter samples, in which each training trajectory is predicted forward over the whole training window from its own initial condition and compared with the truth, so the criterion is not a one-step residual. We have used Tikhonov regularization in all our experiments in Section~\ref{sec:results},~\ref{app:sensitivity} and \ref{app:timing}.

\begin{remark}\label{rem:padmd_rank}
The model ranks $\tilde r$ and $\hat r$ in \texttt{paDMD} determine how many dominant singular values of $\mf\Psi$ and $\mf X^+$ are retained. The energy criterion \eqref{eq:energy_criterion} used for DMD can also be applied here. As in DMD, however, no single energy threshold suits every problem, or every choice of training parameter samples, so $\tilde r$ and $\hat r$ are chosen heuristically for each problem by inspecting the decay of the singular values and the resulting prediction accuracy. Of the two, $\hat r$ requires much less tuning. Since the learned operator $\tilde{\mf K}^{(\boldsymbol\theta)}$ in \eqref{eq:padmd_K_theta} has rank at most $\min\{\tilde r,n\}$, a common choice is $\hat r=\min\{\tilde r,n\}$ (with $n$ replaced by $dn$ for the Hankel variant of \texttt{paDMD}). This choice gave the most accurate results in our experiments.
\end{remark}

\begin{remark}\label{rem:padmd_error_sources}
The approximation error in \eqref{eq:padmd_euler}--\eqref{eq:padmd_reconstruction} has two sources: truncation error from the model rank truncations $\tilde r$ and $\hat r$, and discretization error from approximating the flow map \eqref{eq:padmd_flow} by a discrete-time operator that is affine in the parameters. Expanding the exponential in \eqref{eq:padmd_flow} shows that the neglected $\mathcal{O}(\Delta t^2)$ terms contain products of the $h_i(\boldsymbol{\theta})$ and $\mathbf{B}_i$, which are not affine in $\{h_i(\boldsymbol{\theta})\}_{i=1}^m$. Consequently, even when the continuous-time dynamics \eqref{eq:param_affine_dyn} are exactly affine, the discrete-time operator identified by \texttt{paDMD} is only approximately affine, with an error that grows with $\Delta t$ and with the range of the parameter space spanned by the training data. Because this approximation is present in every result reported here, those results also show that \texttt{paDMD} tolerates moderate departures from exact affine dependence.
\end{remark}

\begin{remark}\label{rem:padmd_local_affine}
In general, any smooth parametric dependence is approximately affine in a neighborhood of a reference parameter. Hence, even for systems that are not exactly parameter-affine, a \texttt{paDMD} model remains accurate over a finite parameter range. The results with the fluidic pinball system illustrate this, since here we observe only the cylinder lift and drag coefficients, whose dynamics are not exactly affine in the chosen parameters. The electron injection rate $\mc{N}_e$ in the plasma systems is another example: the deposited current density scales linearly with $\mc{N}_e$, so the source term driving the field evolution is affine in it, but the space-charge field also scales with $\mc{N}_e$ and feeds back on the particle trajectories, so the self-consistent solution is only approximately affine in it. Yet \texttt{paDMD} still predicts accurately across the sampled range of parameters in both cases. Deriving rigorous error bounds that quantify this local accuracy is left for future work.
\end{remark}

\section{Results}\label{sec:results}

In this section, we evaluate the proposed \texttt{paDMD} method on several complex, high-dimensional nonlinear systems: incompressible fluid flow past a cylinder, the fluidic pinball system, an oscillating electron beam in a transverse sinusoidal magnetic field, and virtual cathode oscillations. In each case, we benchmark \texttt{paDMD} against state-of-the-art parametric ROM methods: \emph{dimension reduction and interpolation in parameter space (DRIPS)} \cite{lu2023drips}, \emph{parametric operator inference (pOpInf)} \cite{mcquarrie2023nonintrusive}, used here in its linear-state form as described below, \emph{stacked parametric DMD} \cite{sayadi2015parametrized} and \emph{reduced Koopman operator interpolation (rKOI)} \cite{huhn2023parametric}. Of the four benchmarks, three interpolate reduced quantities across the parameter space. In the one-parameter settings, stacked parametric DMD and rKOI use piecewise-linear interpolation between the training samples, following \cite{huhn2023parametric}, while for DRIPS the choice is left to the user in \cite{lu2023drips}, and we use a cubic spline. In the multi-parameter settings all three use barycentric weights obtained from a Delaunay triangulation of the training samples. A sample lying outside the triangulation is assigned the barycentric coordinates of the nearest simplex evaluated at the test sample itself, so that such a sample is extrapolated rather than projected onto the simplex. Here, both \cite{huhn2023parametric} and \cite{lu2023drips} demonstrate their methods on tensor-product grids of training parameter samples, and rely on substantially more training parameter samples than we use here. We instead use scattered samples at a much smaller budget, which is the regime of interest in this work, and the barycentric weights extend their interpolation rules to that setting. For rKOI we follow \cite{huhn2023parametric} in interpolating the reduced operator and the basis, but we fit the initial modal amplitudes to the initial condition at the target parameter rather than interpolating the training amplitudes. This supplies that method with the same initial state available to every other method here, as we believe it makes for a more fair comparison. The pOpInf method requires no interpolation, since the parametric structure enters its regression directly. We use it in the form that matches the models compared here, with the reduced operators expanded affinely in the parametrization functions together with a constant term and without higher-order polynomial terms, so that every method is fitted at the same model order. The required time derivatives of the reduced states are estimated from the snapshots by a first-order forward difference. As our error metric, we use the $2$-norm residual (relative) error $\delta^{(j)}(t)=\frac{\lVert\mf{x}^{(j)} (t)-\hat{\mf{x}}^{(j)} (t)\rVert_2}{\lVert\mf{x}^{(j)}(t)\rVert_2}$, where $\mf{x}^{(j)}(t)$ is the true state and $\hat{\mf{x}}^{(j)}(t)$ is the predicted state for a test parameter sample $\boldsymbol{\theta}^{*(j)}$ at time $t$. For every system, we report the following results:
\begin{itemize}
    \item \emph{Sensitivity to model rank.} A table of the relative error of every method, averaged over the prediction horizon and over the test parameter samples, for a finite range of model ranks. The lowest error for each method is marked in bold and `--' indicates a diverging prediction, so that the sensitivity of each method to the choice of model rank can be inferred. Entries marked `N/A' denote model ranks unavailable to a given method.
    \item \emph{Time-averaged relative error.} A box plot of the error $\delta^{(j)}(t)$ averaged over the prediction horizon, with the distribution taken over all the test parameter samples and each method evaluated at its own best model rank. This compares the methods by the median, spread and worst case of their accuracy over the range of test parameter samples, rather than at any individual sample. For the two-parameter electron beam case, we additionally report the time-averaged error at each individual test sample across the parameter space.
    \item \emph{Parameter-averaged error evolution.} The error $\delta^{(j)}(t)$ averaged over the test parameter samples at each time step of the prediction horizon, with the shaded region indicating standard deviation across those samples. This shows how the error accumulates over the horizon.
    \item \emph{Ground truth versus prediction.} Snapshots of the true and predicted states at the final prediction time step for representative test parameter samples, together with the corresponding pointwise error. For the fluidic pinball, where the observed quantities are the lift and drag coefficients rather than a field, we instead compare the predicted and true coefficients over the entire prediction horizon. For the plasma systems we additionally show the predicted and true evolution of the electric field at a single degree of freedom, i.e., at a single mesh edge. In all three of these systems we also report the relative error over the prediction horizon at the individual test parameter samples shown.
    \item \emph{Sensitivity analysis.} \ref{app:noise} reports the time-averaged relative error for each method and each system under additive measurement noise. \ref{app:ntrain} reports the same errors for clean data, but as the number of training parameter samples $L$ is varied on a fixed test set.
    \item \emph{Computational cost.} \ref{app:timing} reports the measured computation time of each method, together with the machine specification.
\end{itemize}
In all systems, \texttt{paDMD} outperforms all the parametric ROM methods, and shows significantly higher robustness to the choice of model ranks.

\subsection{Flow past cylinder}\label{subsec:results_cylinder}

We consider the two-dimensional incompressible flow past a stationary circular cylinder, a canonical benchmark in nonlinear dynamics that exhibits rich nonlinear behavior. The flow is governed by the incompressible Navier--Stokes (NS) equations
\begin{equation*}\label{eq:NS}
\frac{\partial \mathbf{u}}{\partial t}+(\mathbf{u}\cdot\nabla)\mathbf{u}
=
-\nabla p+\frac{1}{\text{Re}} \nabla^{2}\mathbf{u},
\quad
\nabla\cdot \mathbf{u}=0,
\end{equation*}
where $\mathbf{u}(\mf{x},t)$ is the velocity field and $p(\mf{x},t)$ is the pressure field. The Reynolds number $\text{Re}$ is given by $\text{Re}=\frac{U_\infty D}{\nu}$, where $U_\infty$ is the inlet velocity, $D$ is the cylinder diameter, and $\nu$ is the kinematic viscosity. Throughout this analysis, all quantities are assumed to be non-dimensionalized, and we consider the horizontal velocity field as our state to be modeled.

We generate training and testing data by numerically solving the NS equations using the immersed boundary projection method (IBPM) \cite{taira2007}, which employs a vorticity-based formulation on a staggered Cartesian grid. The computational domain is $\Omega=[-1,8]\times[-2,2]$, discretized with a uniform $300\times 133$ grid. The cylinder is centered at the origin. The simulation is advanced with a fixed time step $\Delta t=0.02$ and is carried out in three phases: (i) a spin-up phase to establish the wake, (ii) a brief perturbation phase in which the cylinder undergoes a small plunging motion to excite the dynamics, and (iii) a data-collection phase in which velocity snapshots are recorded at each time step. We enforce zero Dirichlet conditions on the domain boundaries and no-slip condition on the cylinder surface.

For parametric modeling, we first consider the viscosity as the parameter of interest, i.e., $\boldsymbol{\theta} =\mf h(\boldsymbol{\theta})=\nu$. Keeping the (non-dimensional) $U_\infty=1$ and (non-dimensional) $D=1$, we vary $\nu$ from 0.01 to 0.02 in increments of $0.001$, which corresponds to Reynolds numbers $\mathrm{Re}\in[50,100]$. To reduce computational cost, we spatially downsample the simulated velocity field by a factor of $3$ in each spatial direction and retain only the streamwise horizontal velocity component $u$. This yields a state dimension of $n=4400$. For each $\nu$, we discard the first $1500$ snapshots to remove the transient portion of the dynamics.

We train the \texttt{paDMD} model using $T=250$ snapshots from three parameter values $\nu= \{0.01,\,0.015,\,0.02\}$, and evaluate the trained model on the remaining values (i.e., we have $L=3$ training and 8 test parameter samples). The $\nu$ values are normalized to be between 0 and 0.1. For each test $\nu$, we predict over a horizon of $1000$ time steps (after transient removal) to evaluate the model's performance. We benchmark this method against the other parametric ROM methods with the same training-data budget. Table~\ref{tab:fpc_nu_ranks} reports the mean prediction error for every method across several model ranks, with the lowest error for each method marked in bold. Figure~\ref{fig:fpc_nu_err_boxplot} plots the time-averaged relative error for each method at the rank with the lowest error in that range, and Figure~\ref{fig:fpc_nu_err_time_series} shows the time evolution of the relative prediction error over the prediction horizon, averaged across the test $\nu$ values. The ground truth vs prediction comparisons at the final prediction time step for \texttt{paDMD} at various test values of $\nu$ are shown in Figure~\ref{fig:fpc_nu_padmd_ground_snaps}.

\begin{table}[pos=!h]
\centering

\footnotesize
\setlength{\tabcolsep}{3pt}
\caption{Mean relative prediction errors for the flow past cylinder system over the test $\nu$ values across various model ranks. The best result for each method is marked in bold, and `--' indicates error blow up.}
\label{tab:fpc_nu_ranks}
\begin{tabular}{lccccccccccc}
\toprule
\multicolumn{1}{c}{\multirow{2}{*}{Method}} & \multicolumn{11}{c}{Model rank ($r$)} \\
\cmidrule(lr){2-12}
 & 10 & 15 & 20 & 25 & 30 & 35 & 40 & 45 & 50 & 55 & 60 \\
\midrule
\makecell[l]{\texttt{paDMD}\\[-1pt]{\tiny $(\tilde r = r,~\hat r = \min\{r,n\})$}} & 0.0317 & 0.0138 & 0.0125 & 0.0111 & 0.0099 & 0.0099 & 0.0093 & 0.0091 & 0.0095 & 0.0092 & \textbf{0.0089} \\
\midrule
DRIPS & \textbf{0.0181} & 0.0493 & 0.0264 & 0.3026 & 0.0404 & -- & -- & -- & -- & -- & -- \\
\midrule
pOpInf & 0.0485 & 0.0291 & 0.0302 & -- & -- & \textbf{0.0279} & -- & 0.0739 & 0.0351 & -- & -- \\
\midrule
Stacked DMD & 0.0687 & 0.0809 & 0.0707 & \textbf{0.0602} & 0.0667 & 0.0702 & 0.0804 & 0.0639 & 0.1103 & 0.1628 & 0.2846 \\
\midrule
rKOI & 0.0487 & 0.0488 & \textbf{0.0486} & 0.0490 & 0.0487 & 0.0488 & 0.0517 & 0.1481 & 0.1700 & 0.5537 & -- \\
\bottomrule
\end{tabular}

\end{table}

\begin{figure}[pos=t]
    \centering
    \includegraphics[width=0.5\textwidth]{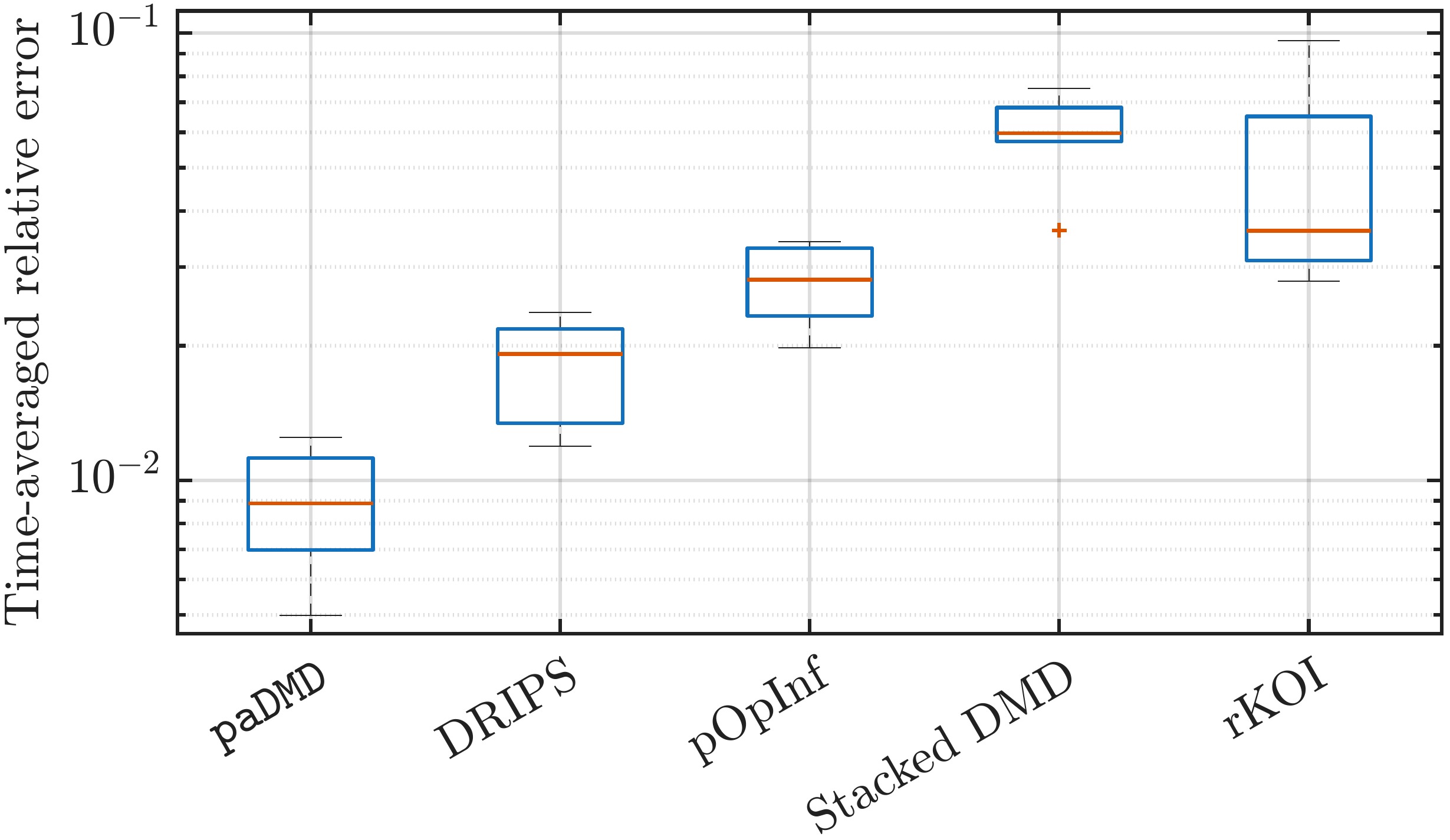}
    \caption{Distribution of the time-averaged relative error over the test $\nu$ values for the flow past cylinder system. The error is averaged over the prediction horizon at each test $\nu$ value, and each method is shown at the model rank with its lowest error.}
    \label{fig:fpc_nu_err_boxplot}
\end{figure}

\begin{figure}[pos=!htbp]
    \centering
    \includegraphics[width=0.6\textwidth]{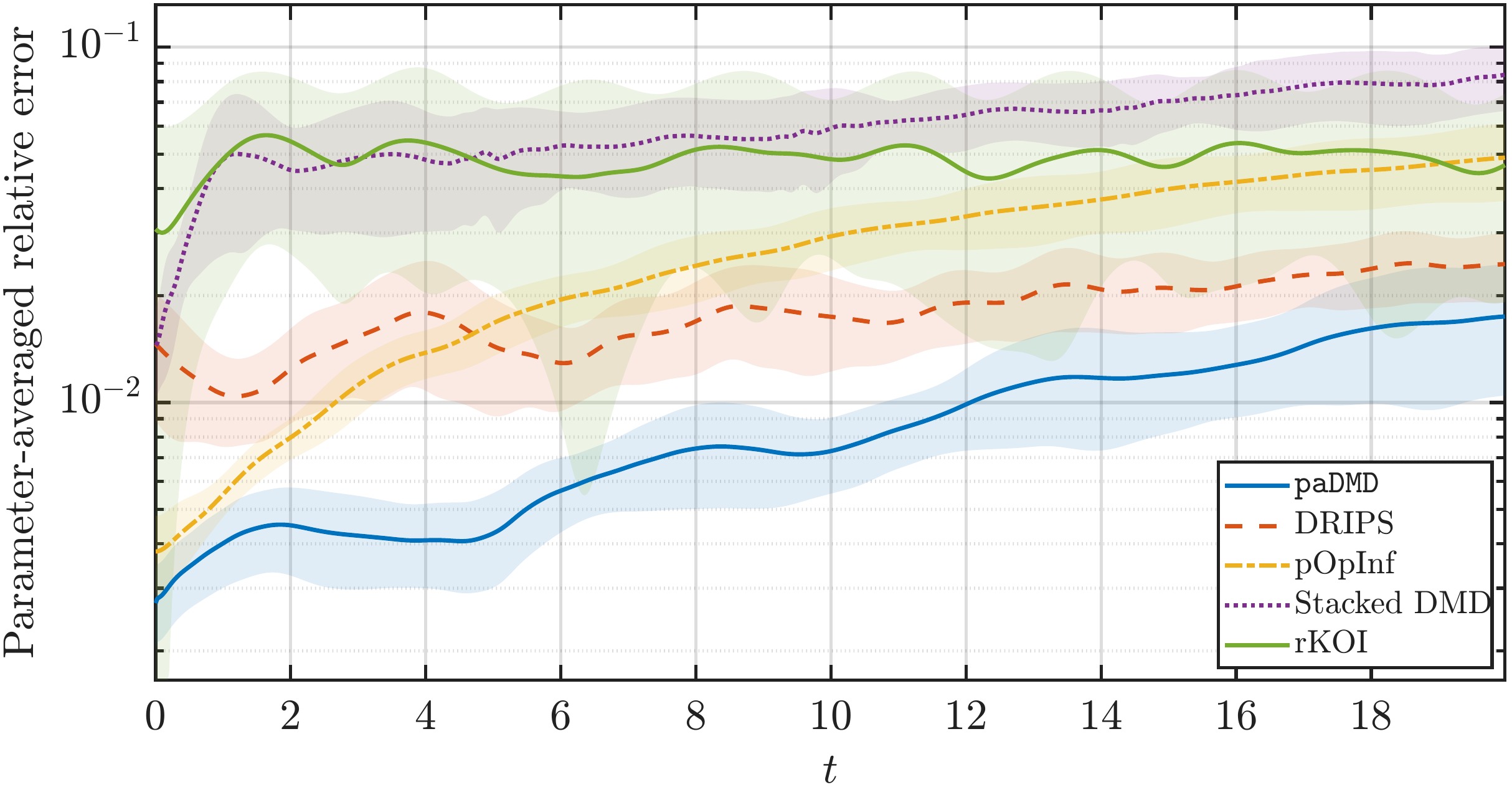}
    \caption{Relative error for the flow past cylinder system at the model ranks with lowest error, averaged over the test $\nu$ values at each time step in the prediction horizon. The shaded regions indicate the standard deviation across the test $\nu$ values.}
    \label{fig:fpc_nu_err_time_series}
\end{figure}

\begin{figure}[pos=!h]
\centering
\subfloat[$\boldsymbol{\theta}^*=\nu=0.013$ (Re $\approx 77$)]{%
  \includegraphics[width=0.31\textwidth]{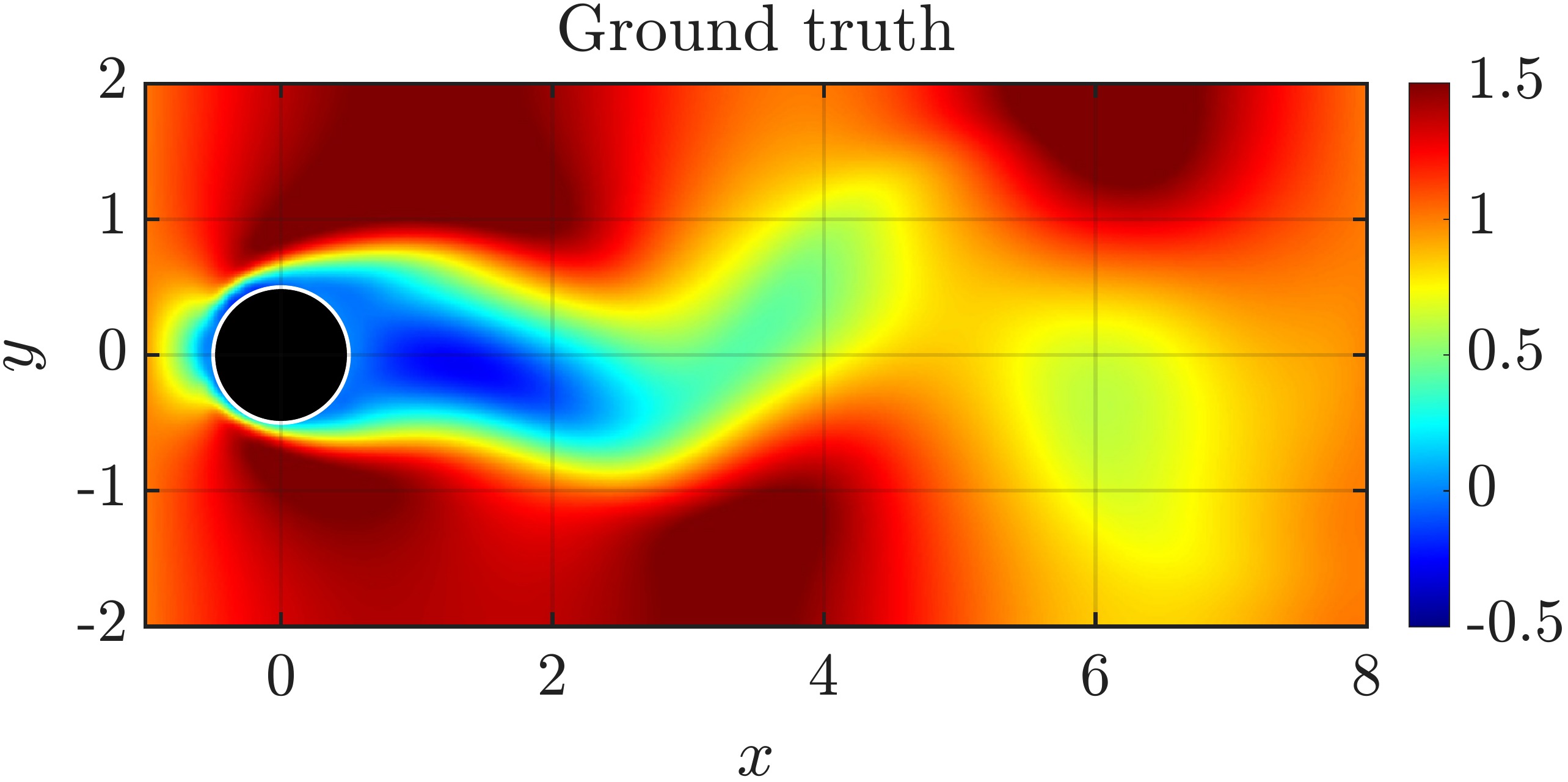}\hspace{0.02\textwidth}%
  \includegraphics[width=0.31\textwidth]{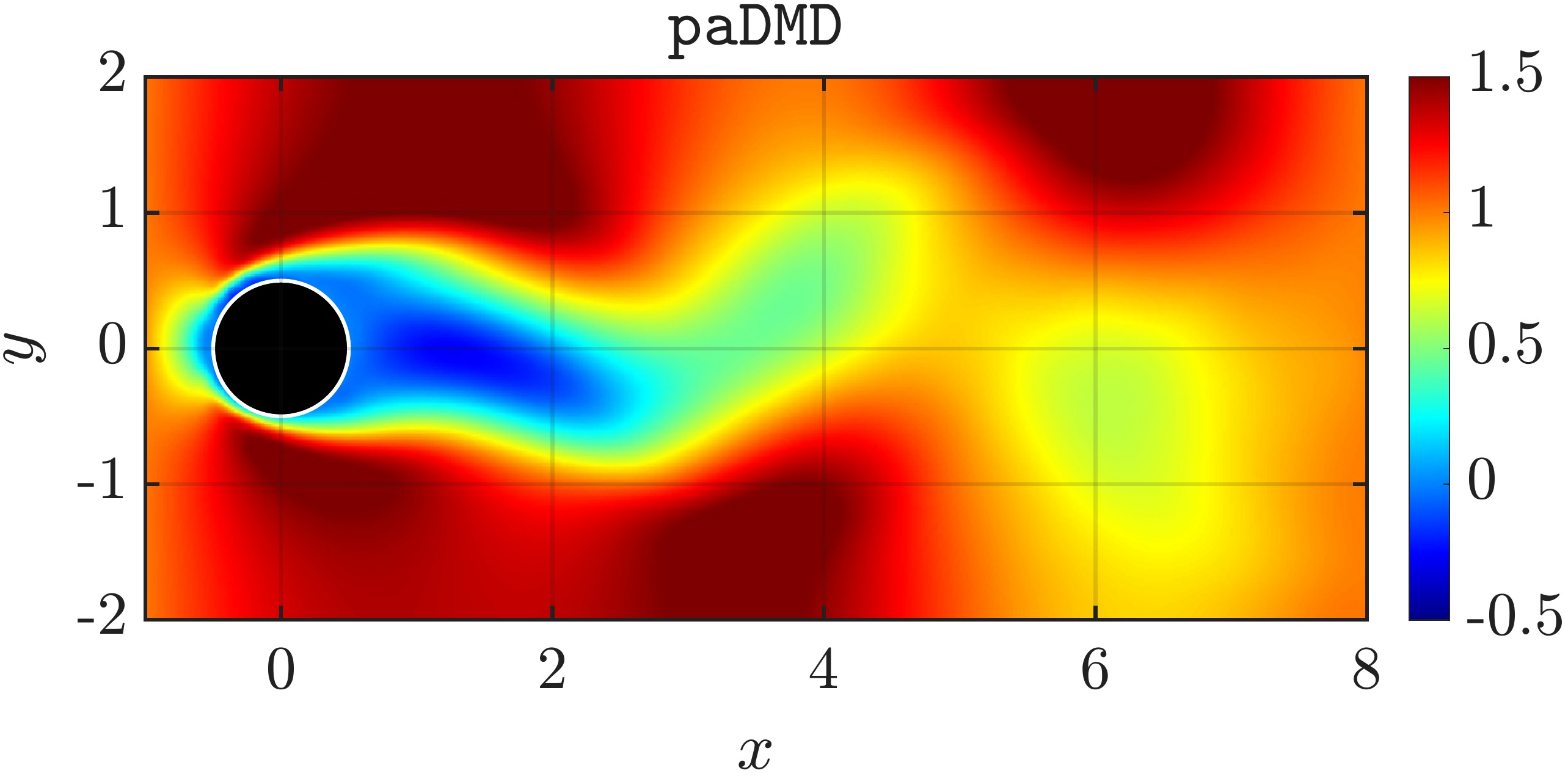}\hspace{0.02\textwidth}%
  \includegraphics[width=0.31\textwidth]{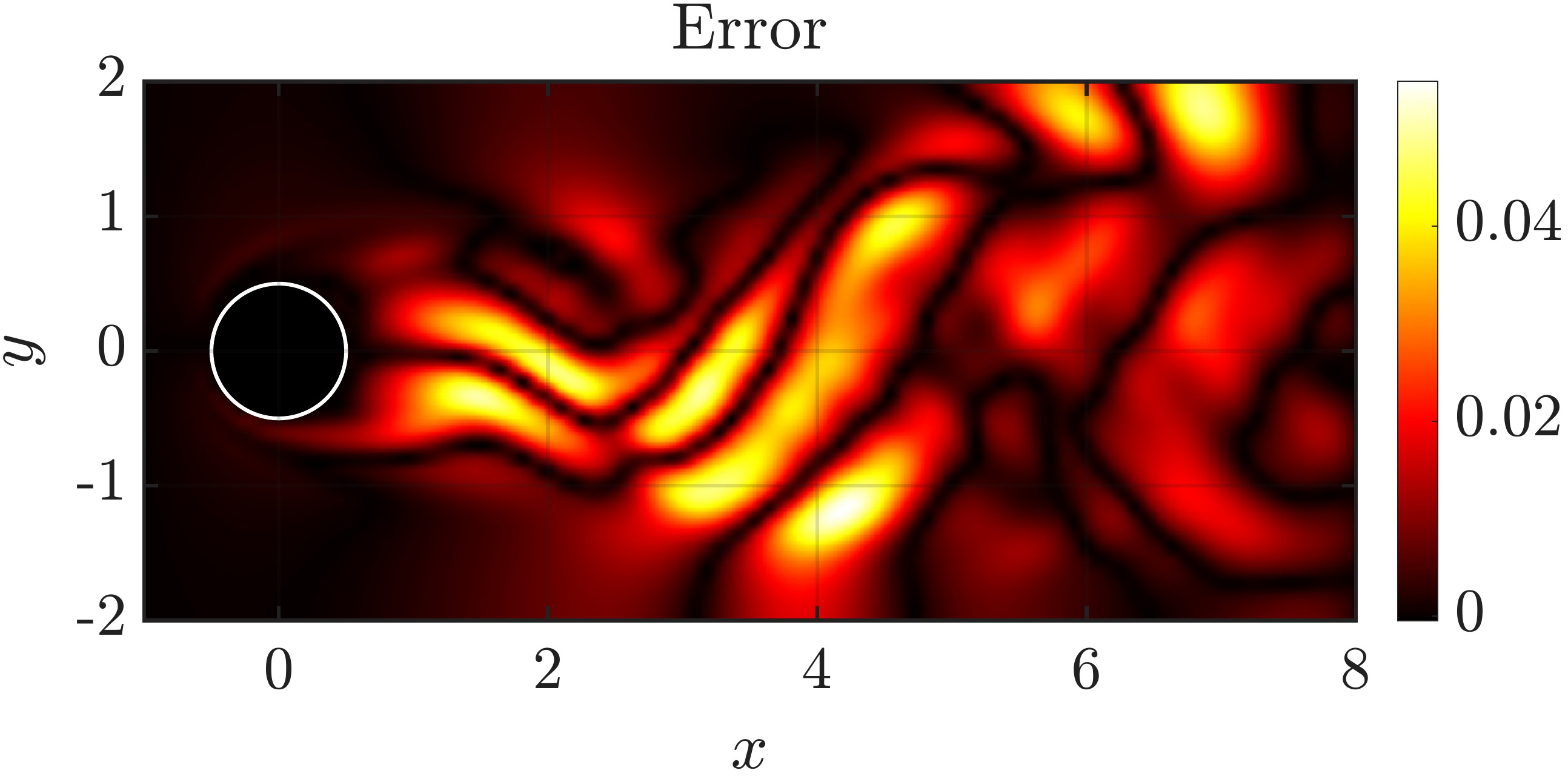}%
}

\vspace{0.8em}
\subfloat[$\boldsymbol{\theta}^*=\nu=0.017$ (Re $\approx 59$)]{%
  \includegraphics[width=0.31\textwidth]{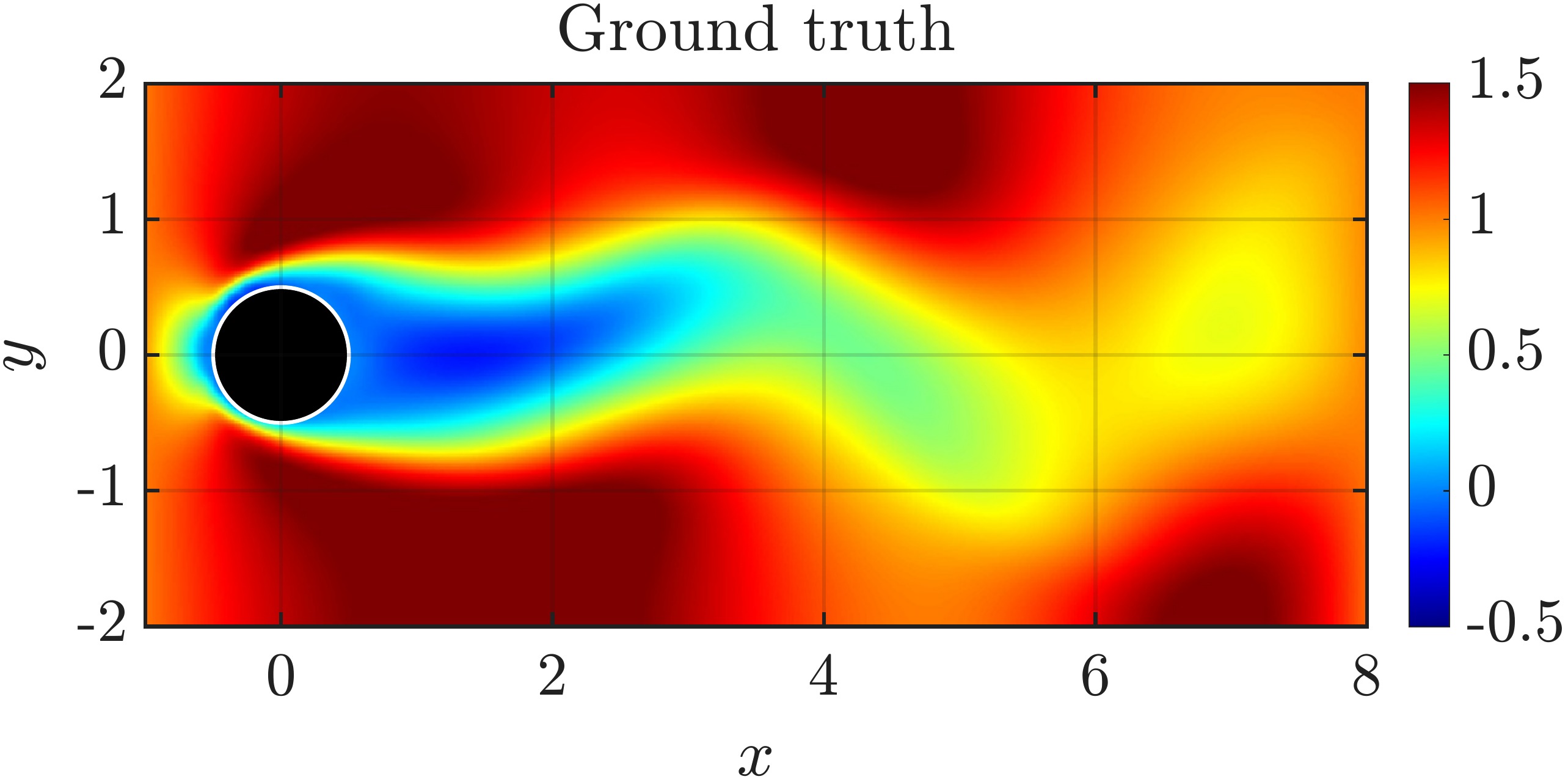}\hspace{0.02\textwidth}%
  \includegraphics[width=0.31\textwidth]{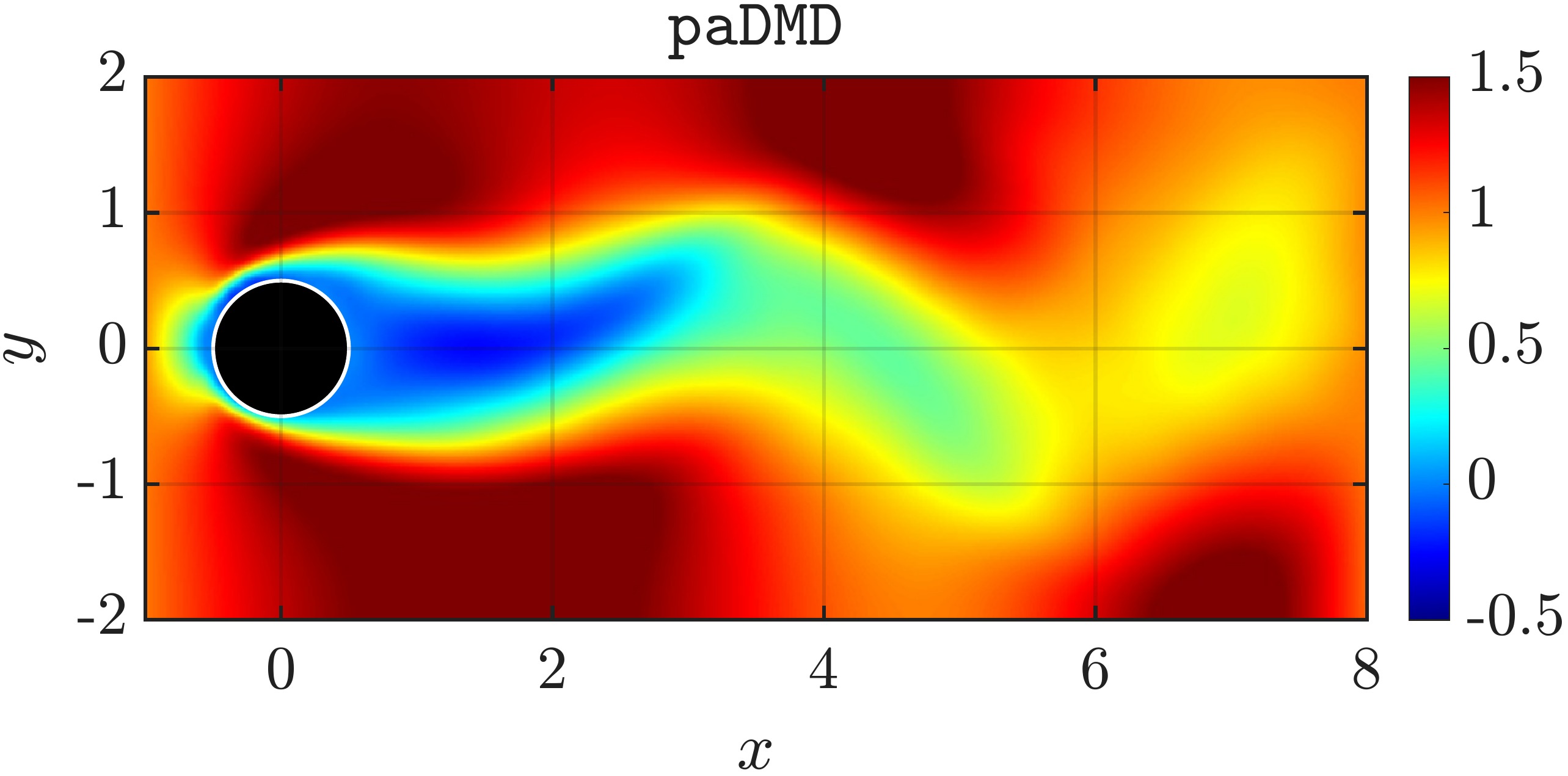}\hspace{0.02\textwidth}%
  \includegraphics[width=0.31\textwidth]{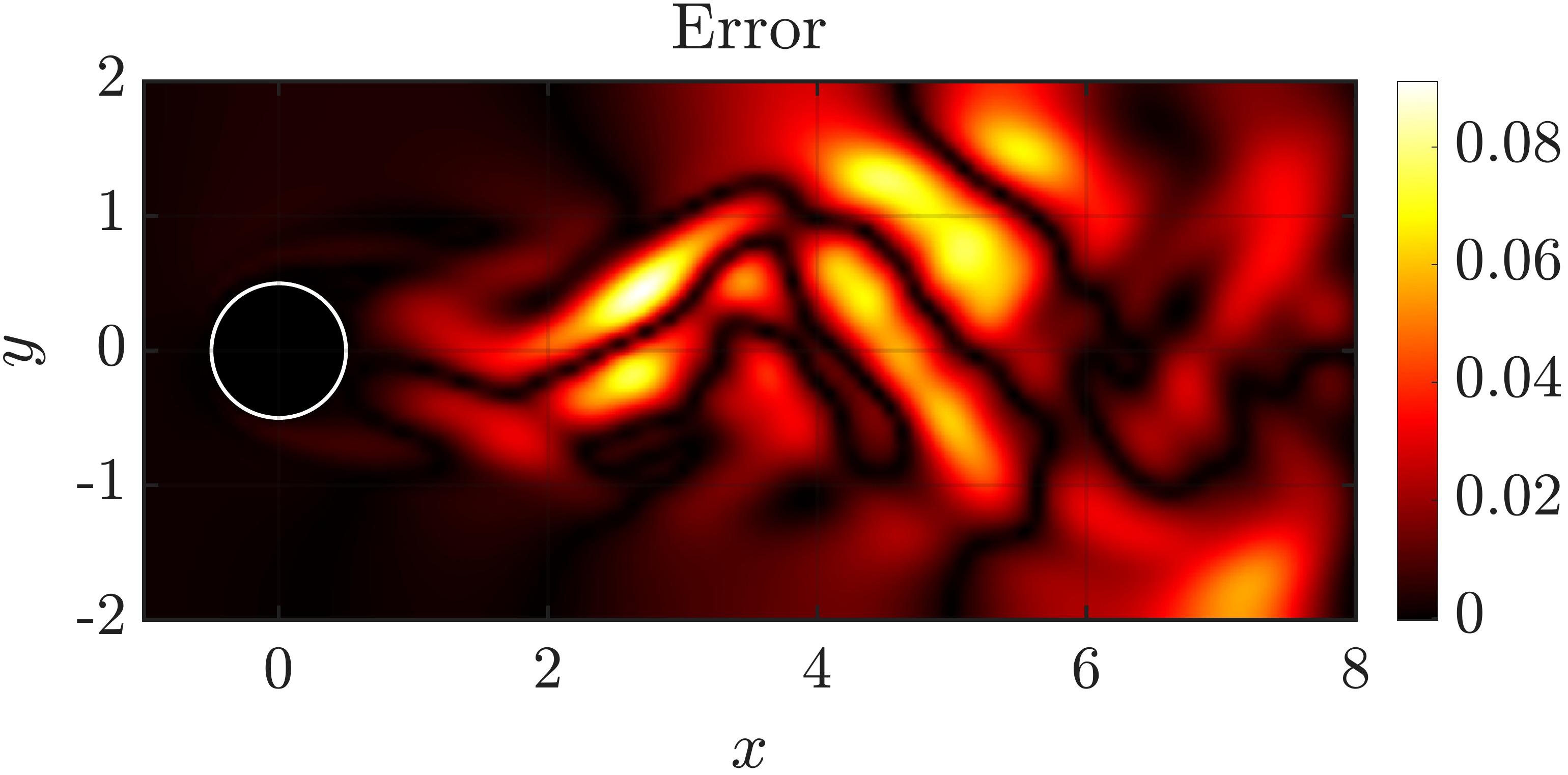}%
}

\caption{Horizontal velocity field snapshots for the flow past cylinder system at the final prediction time step for two different test values of $\nu$. Each row shows the ground truth, \texttt{paDMD} prediction ($\tilde r=\hat r=60$), and pointwise absolute error.}
\label{fig:fpc_nu_padmd_ground_snaps}
\end{figure}

Next, we consider a more challenging case involving \emph{geometric variation}, by taking the cylinder diameter $D$ as the parameter of interest, with $\boldsymbol{\theta}=D$ and $\mf h(\boldsymbol{\theta})=1/D$. Keeping $U_\infty=1$ and $\nu=0.01$ fixed, we vary $D$ from $0.80$ to $1.30$ in increments of $0.05$, which corresponds to Reynolds numbers $\mathrm{Re}\in[80,130]$. The remaining simulation setup is kept the same as earlier. For each $D$, we discard the first $1500$ snapshots to remove the transient portion of the dynamics.

We train the \texttt{paDMD} model using $T=250$ snapshots from three parameter values $D= \{0.80,\,1.00,\,1.30\}$, and evaluate the trained model on the remaining values (i.e., we have $L=3$ training parameter samples and 8 test parameter samples). The values of $\mf h(\boldsymbol{\theta})=1/D$ are normalized to be between 0 and 0.1. We predict over a horizon of $1000$ time steps (after transient removal) for each test $D$. Table~\ref{tab:fpc_D_ranks} reports the mean prediction error for every method across several model ranks, with the lowest error for each method marked in bold. Figure~\ref{fig:fpc_d_err_boxplot} plots the time-averaged relative error for each method at the rank with the lowest error in that range, and Figure~\ref{fig:fpc_d_err_time_series} shows the time evolution of the relative prediction error over the prediction horizon, averaged across the test $D$ values. The ground truth vs prediction comparisons at the final prediction time step for \texttt{paDMD} at various test values of $D$ are shown in Figure~\ref{fig:fpc_d_padmd_ground_snaps}. %
\begin{table}[pos=h]
\centering

\footnotesize
\setlength{\tabcolsep}{3pt}
\caption{Mean relative prediction errors over the test $D$ values across various model ranks. The best result for each method is marked in bold, and `--' indicates error blow up.}
\label{tab:fpc_D_ranks}
\begin{tabular}{lccccccccccc}
\toprule
\multicolumn{1}{c}{\multirow{2}{*}{Method}} & \multicolumn{11}{c}{Model rank ($r$)} \\
\cmidrule(lr){2-12}
 & 10 & 15 & 20 & 25 & 30 & 35 & 40 & 45 & 50 & 55 & 60 \\
\midrule
\makecell[l]{\texttt{paDMD}\\[-1pt]{\tiny $(\tilde r = r,~\hat r = \min\{r,n\})$}} & 0.0531 & 0.0420 & 0.0321 & 0.0286 & 0.0278 & 0.0276 & 0.0268 & 0.0266 & 0.0267 & 0.0268 & \textbf{0.0266} \\
\midrule
DRIPS & \textbf{0.0347} & 0.0619 & 0.0547 & 0.0859 & 0.0480 & 0.0630 & 0.0623 & 0.0637 & 0.1511 & -- & -- \\
\midrule
pOpInf & 0.1554 & 0.1161 & 0.3061 & -- & 0.1060 & -- & \textbf{0.0656} & 0.0676 & 0.1772 & -- & -- \\
\midrule
Stacked DMD & 0.1475 & 0.1256 & 0.1321 & 0.1254 & \textbf{0.1202} & 0.1318 & 0.1213 & 0.1255 & 0.5007 & 0.2446 & -- \\
\midrule
rKOI & \textbf{0.1421} & 0.1424 & 0.1426 & 0.1436 & 0.1437 & 0.1439 & 0.1443 & 0.1443 & 0.1445 & 0.1491 & 0.1460 \\
\bottomrule
\end{tabular}

\end{table}

\begin{figure}[pos=h]
    \centering
    \includegraphics[width=0.5\textwidth]{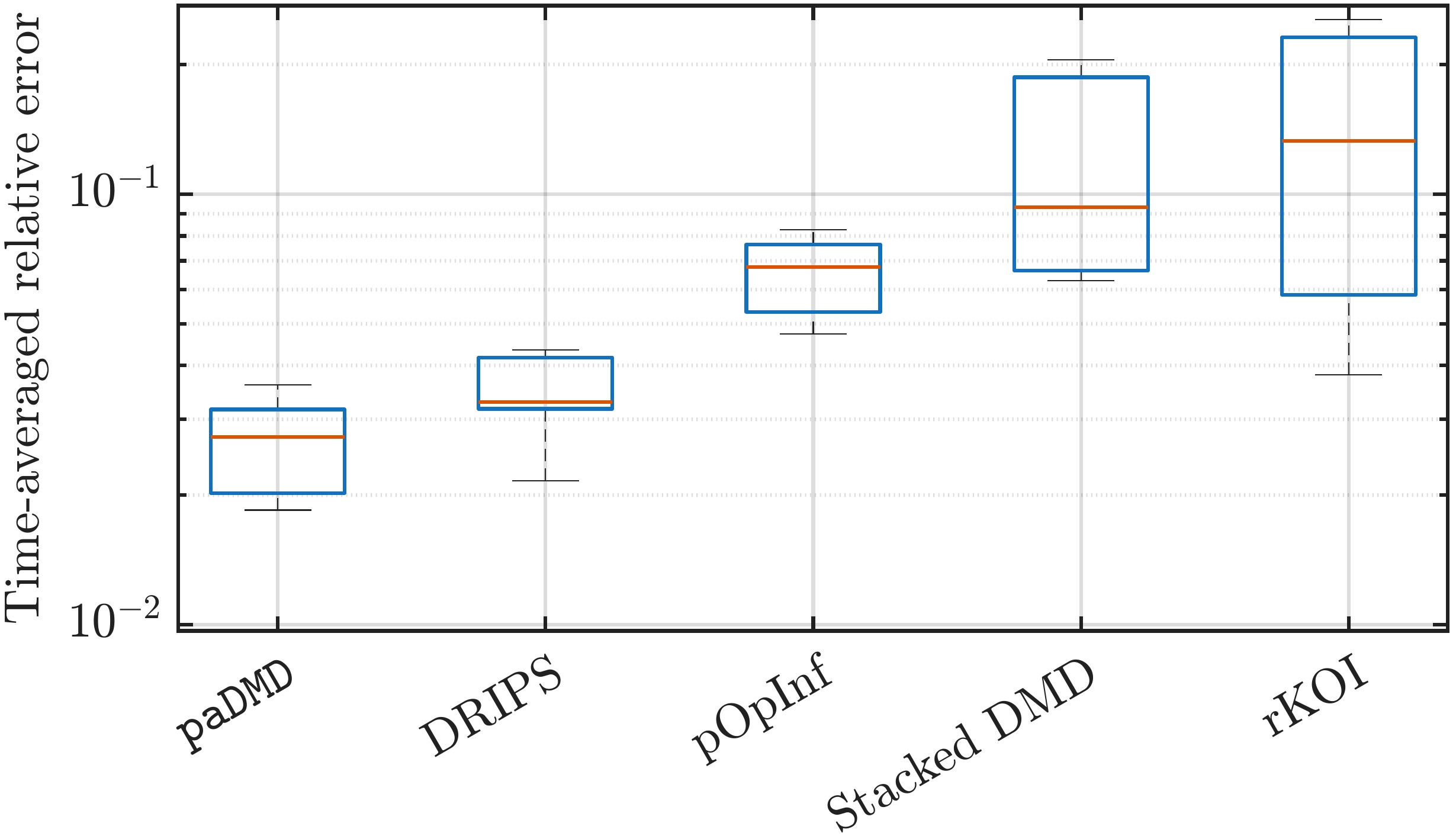}
    \caption{Distribution of the time-averaged relative error over the test $D$ values for the flow past cylinder system. The error is averaged over the prediction horizon at each test $D$ value, and each method is shown at the model rank with its lowest error.}
    \label{fig:fpc_d_err_boxplot}
\end{figure}

\begin{figure}[pos=!htbp]
    \centering
    \includegraphics[width=0.6\textwidth]{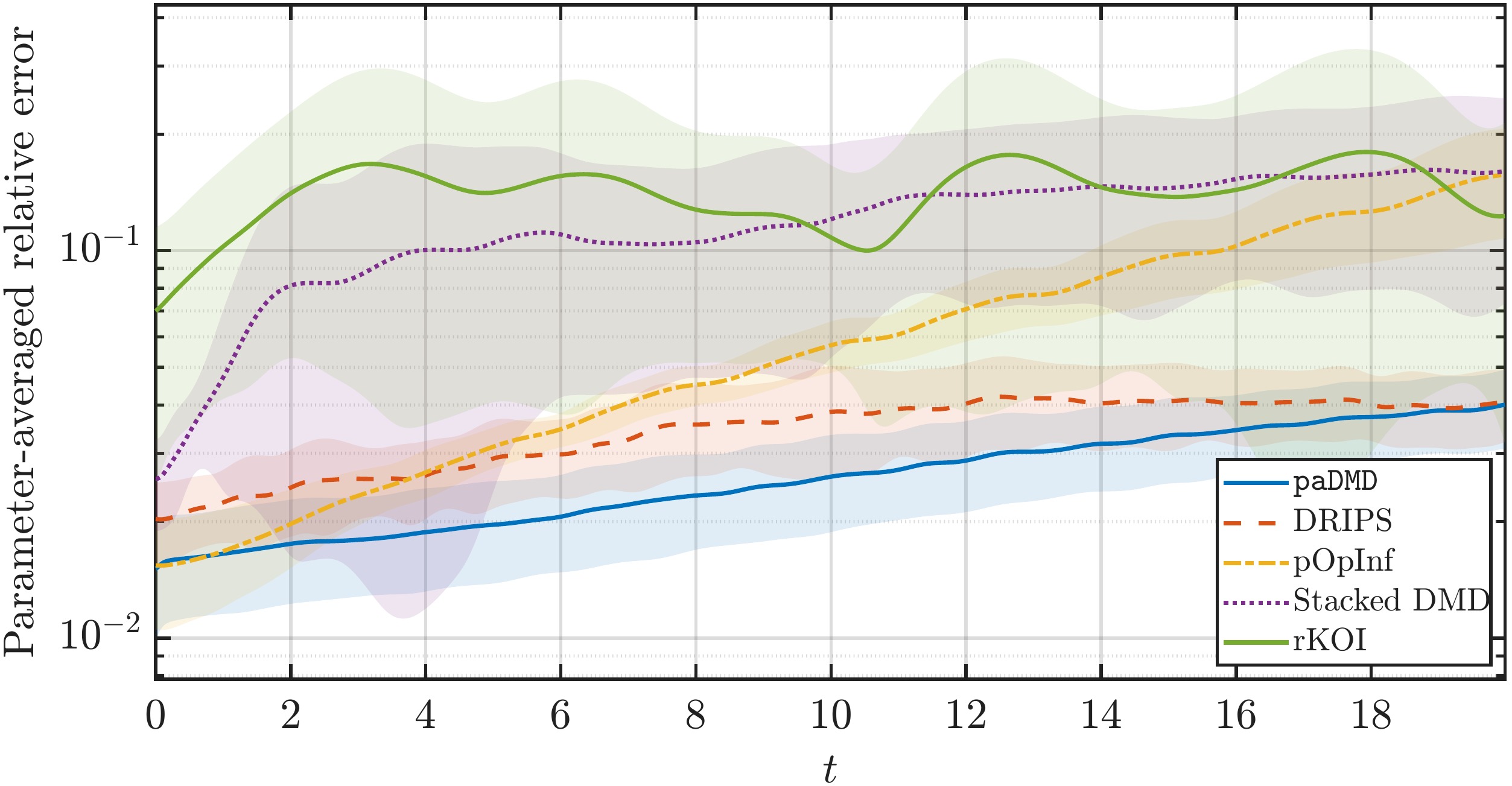}
    \caption{Relative error for the flow past cylinder system at the model ranks with lowest error, averaged over the test $D$ values at each time step in the prediction horizon. The shaded regions indicate the standard deviation across the test $D$ values.}
    \label{fig:fpc_d_err_time_series}
\end{figure}

\begin{figure}[pos=!h]
\centering
\subfloat[$\boldsymbol{\theta}^*=D=0.95$ ($\rm{Re} = 95$)]{%
  \includegraphics[width=0.31\textwidth]{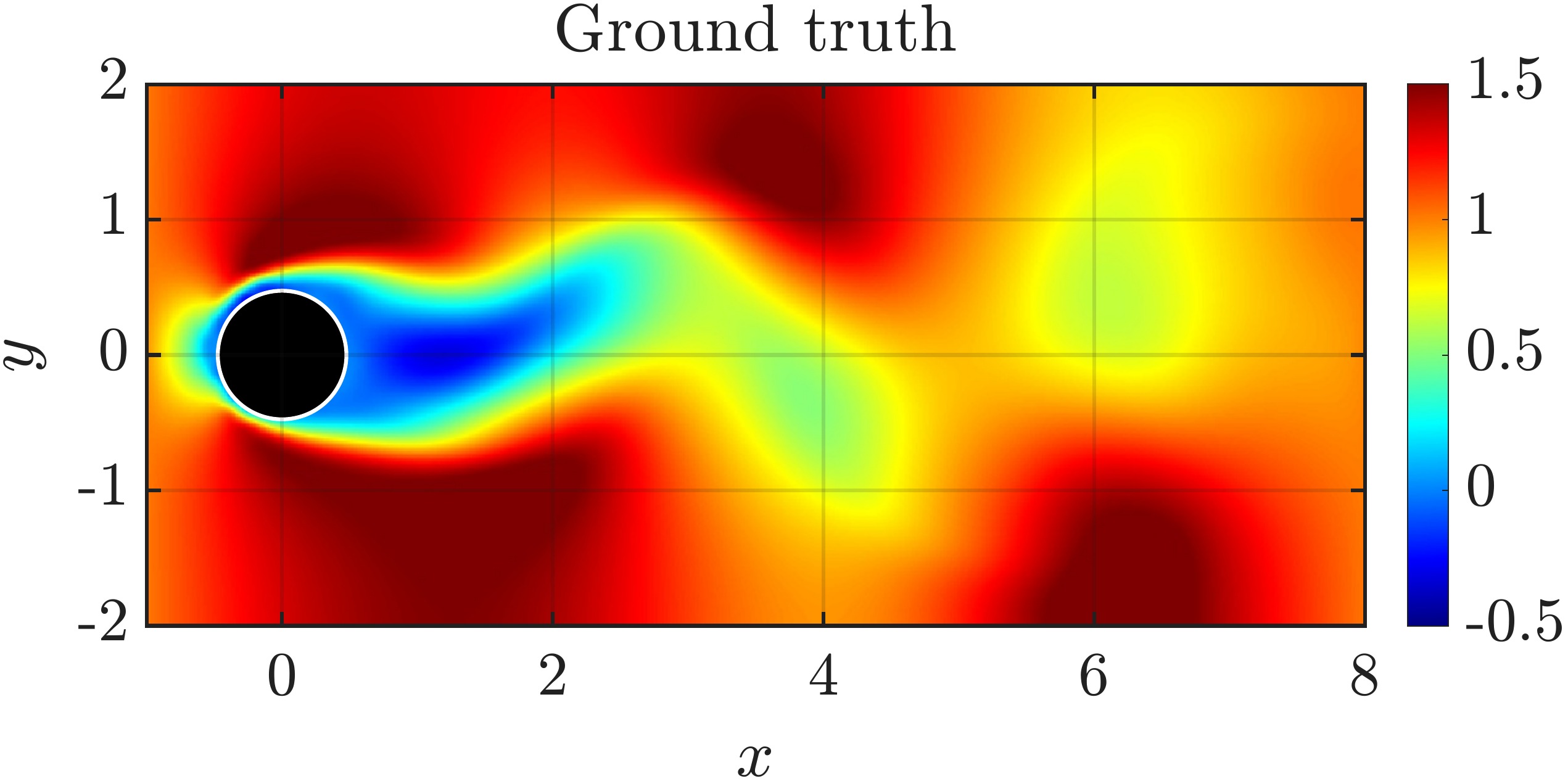}\hspace{0.02\textwidth}%
  \includegraphics[width=0.31\textwidth]{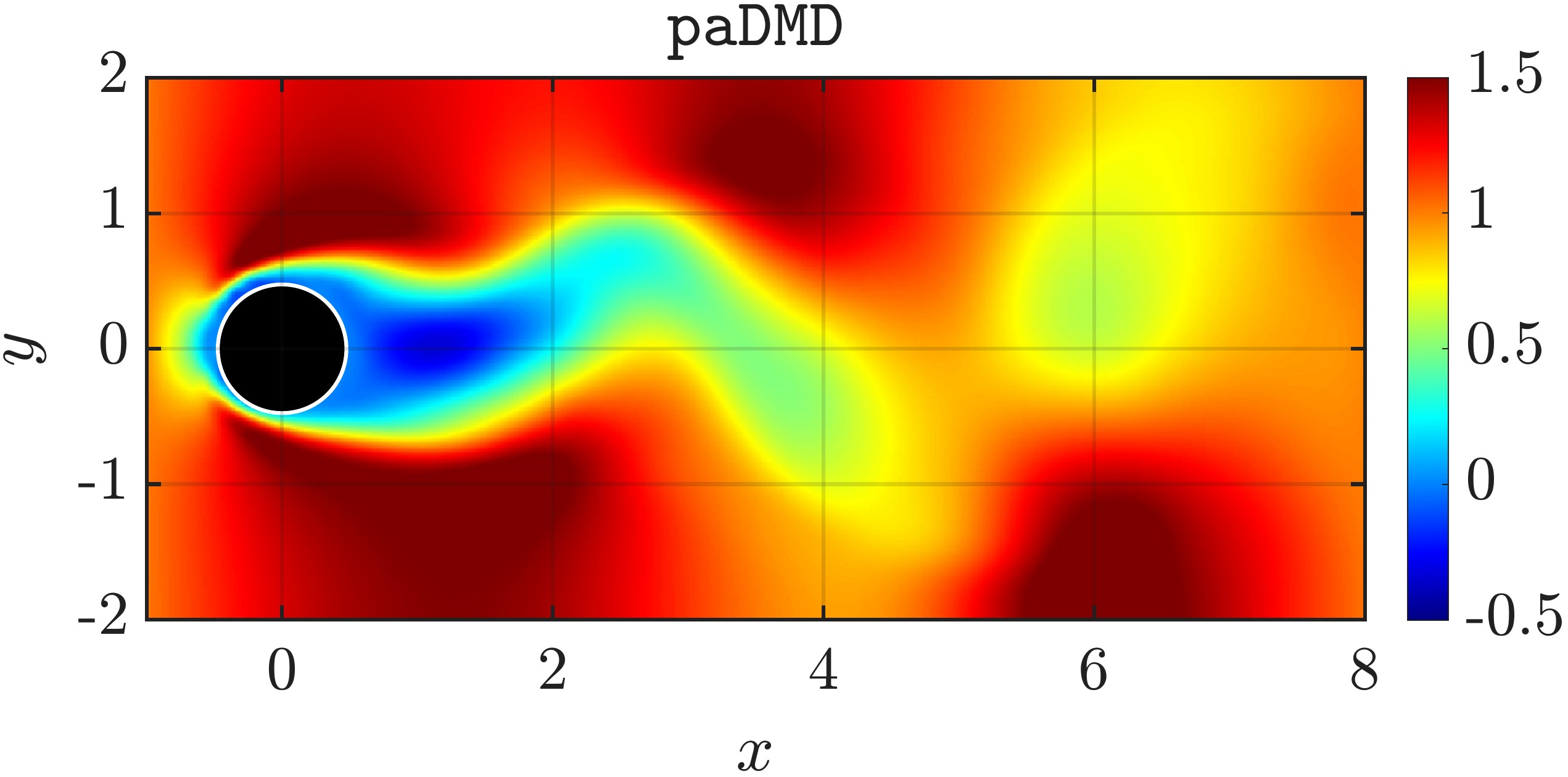}\hspace{0.02\textwidth}%
  \includegraphics[width=0.31\textwidth]{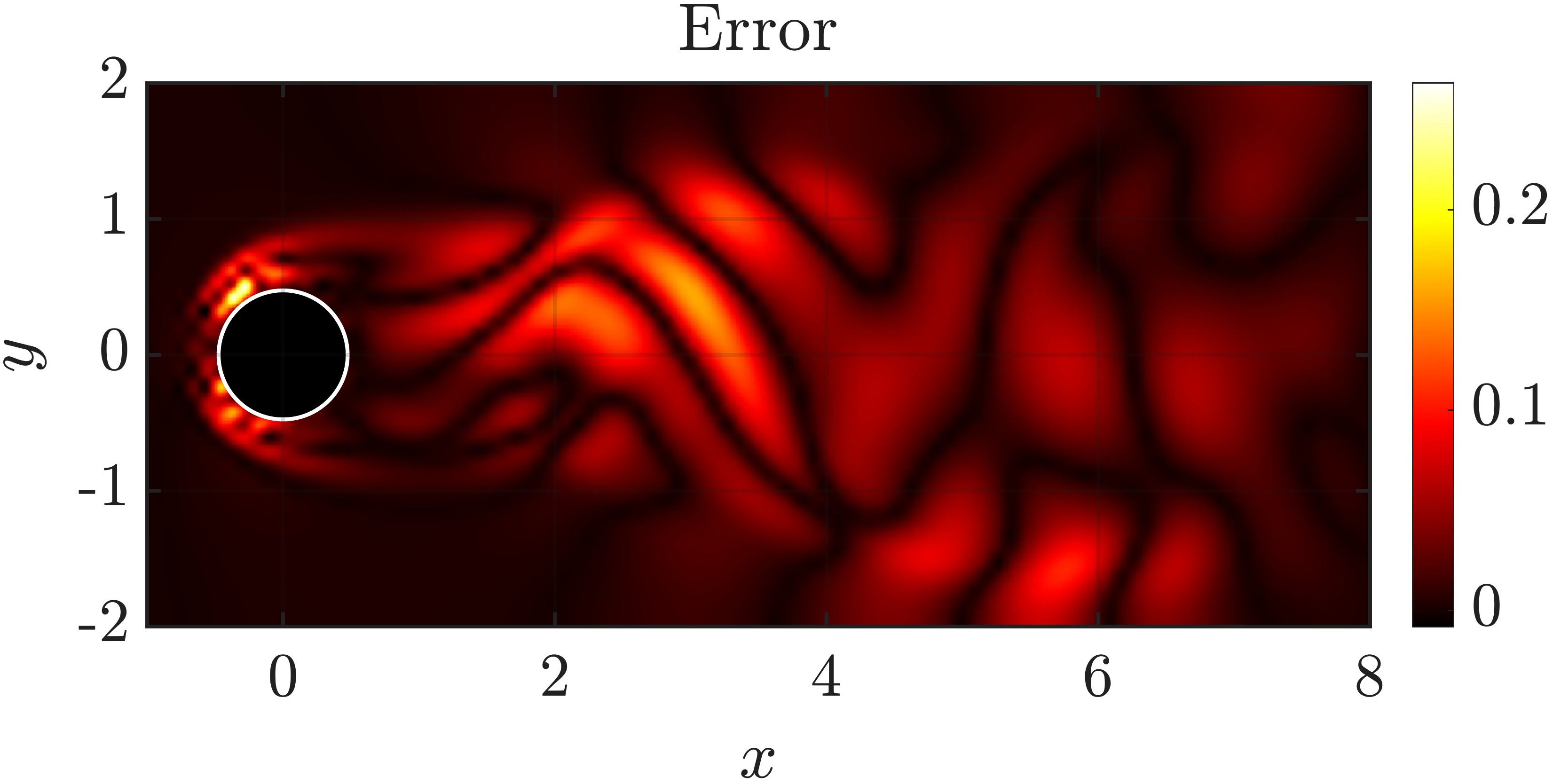}%
}

\vspace{0.8em}
\subfloat[$\boldsymbol{\theta}^*=D=1.20$ ($\rm{Re} = 120$)]{%
  \includegraphics[width=0.31\textwidth]{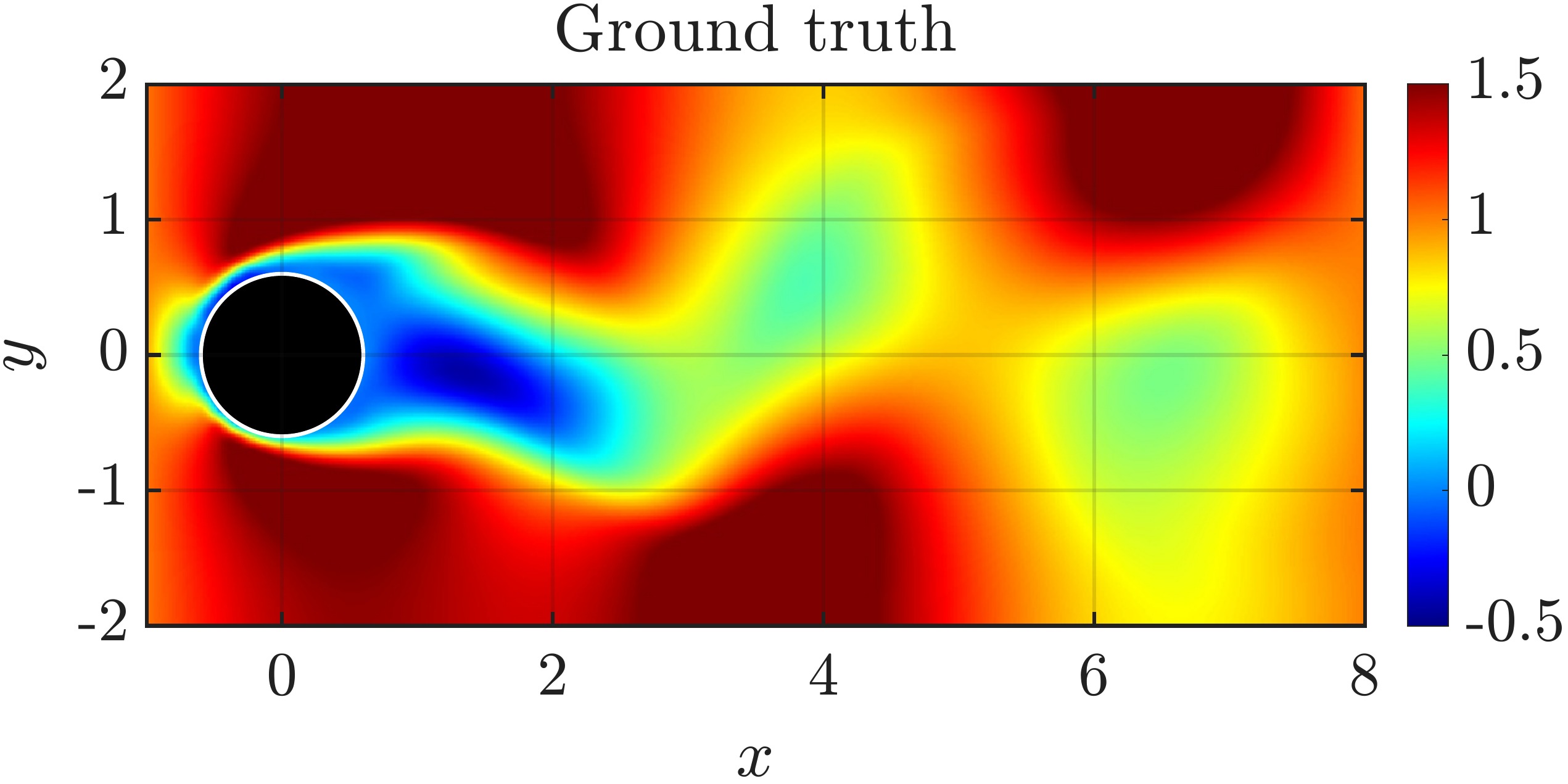}\hspace{0.02\textwidth}%
  \includegraphics[width=0.31\textwidth]{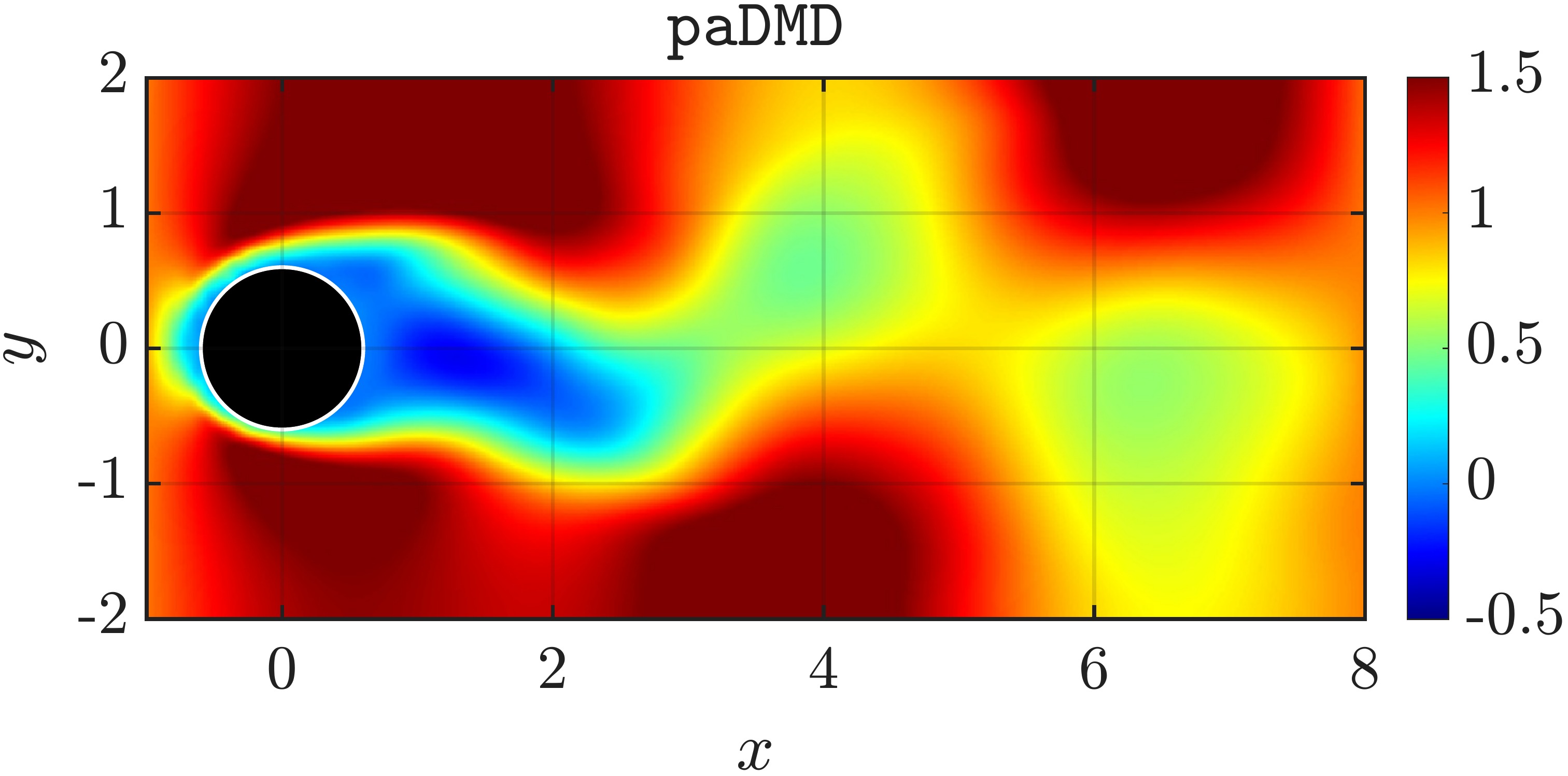}\hspace{0.02\textwidth}%
  \includegraphics[width=0.31\textwidth]{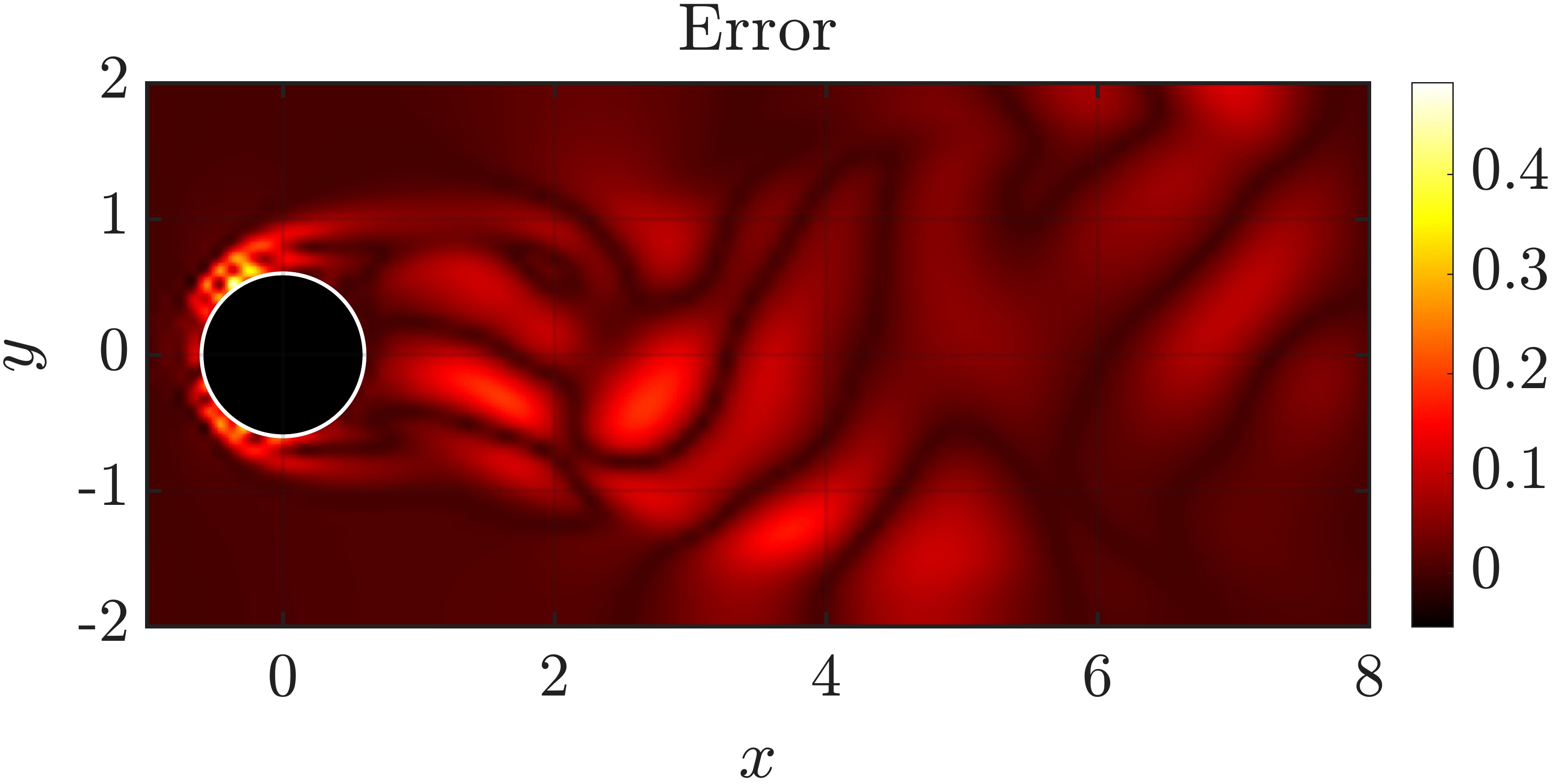}%
}

\caption{Horizontal velocity field snapshots for the flow past cylinder system at the final prediction time step for two different test values of $D$. Each row shows the ground truth, \texttt{paDMD} prediction ($\tilde r = \hat r=60$), and pointwise absolute error.}
\label{fig:fpc_d_padmd_ground_snaps}
\end{figure}

\subsection{Fluidic pinball}\label{subsec:results_pinball}

We now consider the fluidic pinball, a more complex system in fluid dynamics. This is a two-dimensional fluid flow past three equal cylinders of diameter $D$, with the cylinder centers placed in an equilateral triangle formation of side length $1.5D$. Each cylinder rotates about its own axis with a constant angular velocity. We otherwise utilize the same IBPM setup as in Section~\ref{subsec:results_cylinder}. As before, all quantities are assumed to be non-dimensionalized. Figure~\ref{fig:pinball_schematic} shows the schematic of the fluidic pinball system configuration.

\begin{figure}[pos=h]
    \centering
    \includegraphics[width=0.5\textwidth]{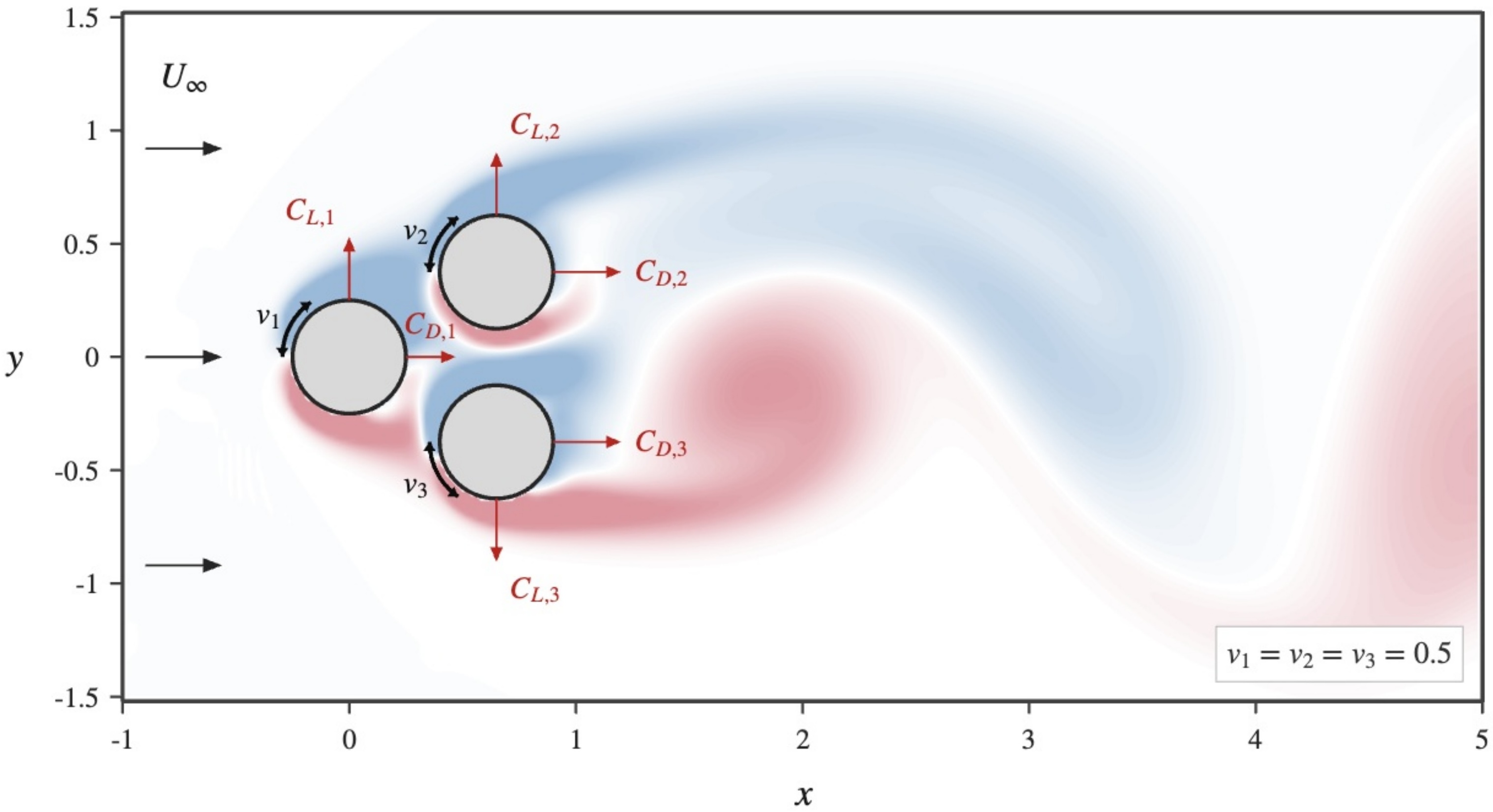}
    \caption{Schematic illustrating the fluidic pinball system for a particular set of cylinder angular velocities.}
    \label{fig:pinball_schematic}
\end{figure}

Here, instead of the velocity field, we only observe the drag ($C_D$) and lift ($C_L$) force coefficients acting on each cylinder, i.e., $\mathbf{x} = [C_{D,1}~C_{D,2}~C_{D,3}~C_{L,1}~C_{L,2}~C_{L,3}]^\top$ and $n = 6$. We consider the three cylinder rotation velocities as our parameters of interest, with $\boldsymbol{\theta} = \mf h(\boldsymbol{\theta})=[v_1~v_2~v_3]^\top$, thereby having a three-dimensional parameter space. In this setting the parameters influence the dynamics through the boundary conditions rather than through the governing equations \cite{Otto2020}, and hence do not exactly satisfy the parameter-affine assumption. The flow is simulated at a fixed $\rm{Re} = 160$. Because this model exhibits stronger nonlinearity and has a low state dimension, we use the Hankel variant of \texttt{paDMD} with $d = 2$ delays to learn an accurate system model. This results in an embedded state of dimension $dn = 12$. Snapshots are sampled with an interval of $\Delta t = 0.05$.

We consider $60$ parameter samples of $v_1,~v_2~\text{and}~v_3$ drawn at random from the range $[0.1, 0.6]^3$. To maintain persistency of excitation across $v_1,~v_2~\text{and}~v_3$, out of those 60 parameter samples we select a small set of $L=15$ values of $\boldsymbol{\theta}$ spanning the range of each of $v_1$, $v_2$ and $v_3$ as our training samples, and leave the remaining $45$ samples for testing. We normalize each of the $v_1$, $v_2$, and $v_3$ parameters to be between 0 and 0.005. We choose $T = 300$ time snapshots from each training sample to train our model, and predict over a horizon of $1000$ snapshots. Table~\ref{tab:pinball_ranks} reports the prediction results for each method when swept over various model ranks. For brevity, the table shows a representative subset of the ranks swept, and the optimal rank for every method is included and marked in bold. Note that in this setting each method was evaluated using the delay-embedded state for a fair comparison. Here, entries marked `N/A' denote model ranks beyond the maximum available for that particular method. This is because the model ranks for DRIPS, rKOI and pOpInf are limited up to the delay-embedded state dimension $dn = 12$, whereas \texttt{paDMD} has the augmented regressor matrix $\mf \Psi$ and stacked DMD utilizes a stacked matrix of time snapshots, and therefore both admit larger ranks. For each method, we select the rank yielding the lowest error for evaluation. The time-averaged relative error results for each method are shown in Figure~\ref{fig:pinball_err_boxplot}, and Figure~\ref{fig:pinball_err_time_series} shows the time evolution of the relative prediction errors over the prediction horizon, averaged across the test parameter samples. The ground truth vs \texttt{paDMD} results over the prediction horizon for two different test parameter values are shown in Figure~\ref{fig:pinball_predictions}.

\begin{table}[pos=!h]
\centering

\footnotesize
\setlength{\tabcolsep}{3pt}
\caption{Mean relative prediction errors over the test parameter samples of the fluidic pinball, across various model ranks. Each method uses the same delay embedding $d=2$. The best result for each method is marked in bold, `--' indicates error blow up, and `N/A' denotes ranks beyond the maximum available for that method.}
\label{tab:pinball_ranks}
\begin{tabular}{lcccccccccc}
\toprule
\multicolumn{1}{c}{\multirow{2}{*}{Method}} & \multicolumn{10}{c}{Model rank ($r$)} \\
\cmidrule(lr){2-11}
 & 2 & 3 & 5 & 7 & 9 & 11 & 12 & 24 & 32 & 48 \\
\midrule
\makecell[l]{\texttt{paDMD}\\[-1pt]{\tiny $(\tilde r = r,~\hat r = \min\{r,dn\})$}} & 0.1781 & 0.2357 & 0.1842 & 0.0617 & 0.0950 & 0.0860 & 0.0945 & 0.1170 & 0.0273 & \textbf{0.0271} \\
\midrule
DRIPS & 0.1196 & \textbf{0.0784} & -- & -- & -- & -- & -- & N/A & N/A & N/A \\
\midrule
pOpInf & 0.1733 & 0.1303 & 0.0819 & 0.0691 & 0.3648 & \textbf{0.0568} & 0.0568 & N/A & N/A & N/A \\
\midrule
Stacked DMD & 0.1113 & \textbf{0.0587} & 0.1011 & -- & 0.4588 & 0.7564 & -- & 0.3580 & 0.3358 & -- \\
\midrule
rKOI & \textbf{0.2099} & -- & -- & -- & -- & -- & -- & N/A & N/A & N/A \\
\bottomrule
\end{tabular}

\end{table}

\begin{figure}[pos=h]
    \centering
    \includegraphics[width=0.5\textwidth]{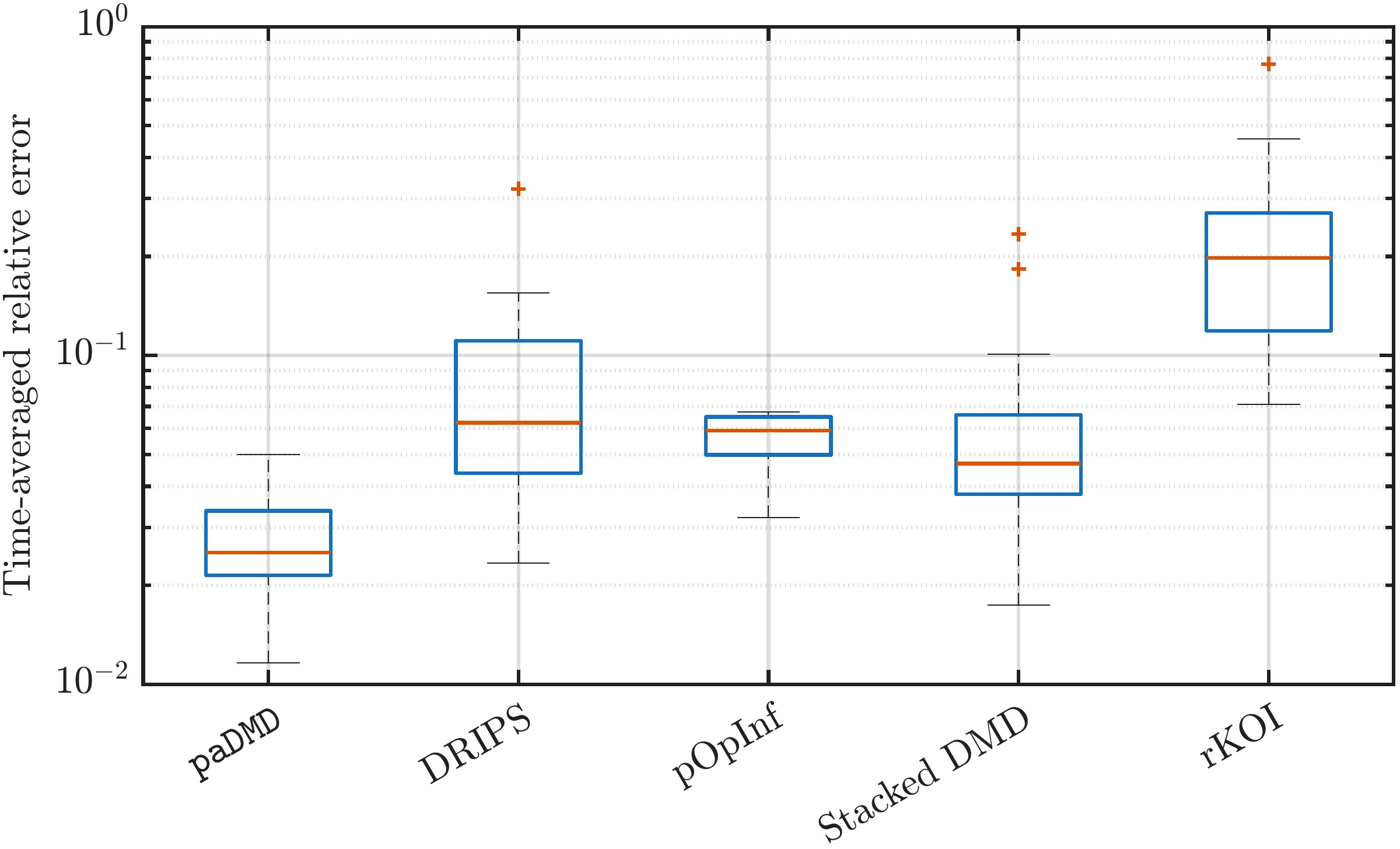}
    \caption{Distribution of the time-averaged relative error over the test $[v_1,~v_2,~v_3]^\top$ values for the fluidic pinball system. The error is averaged over the prediction horizon at each test $[v_1,~v_2,~v_3]^\top$ value, and each method is shown at the model rank with its lowest error.}
    \label{fig:pinball_err_boxplot}
\end{figure}

\begin{figure}[pos=!htbp]
    \centering
    \includegraphics[width=0.6\textwidth]{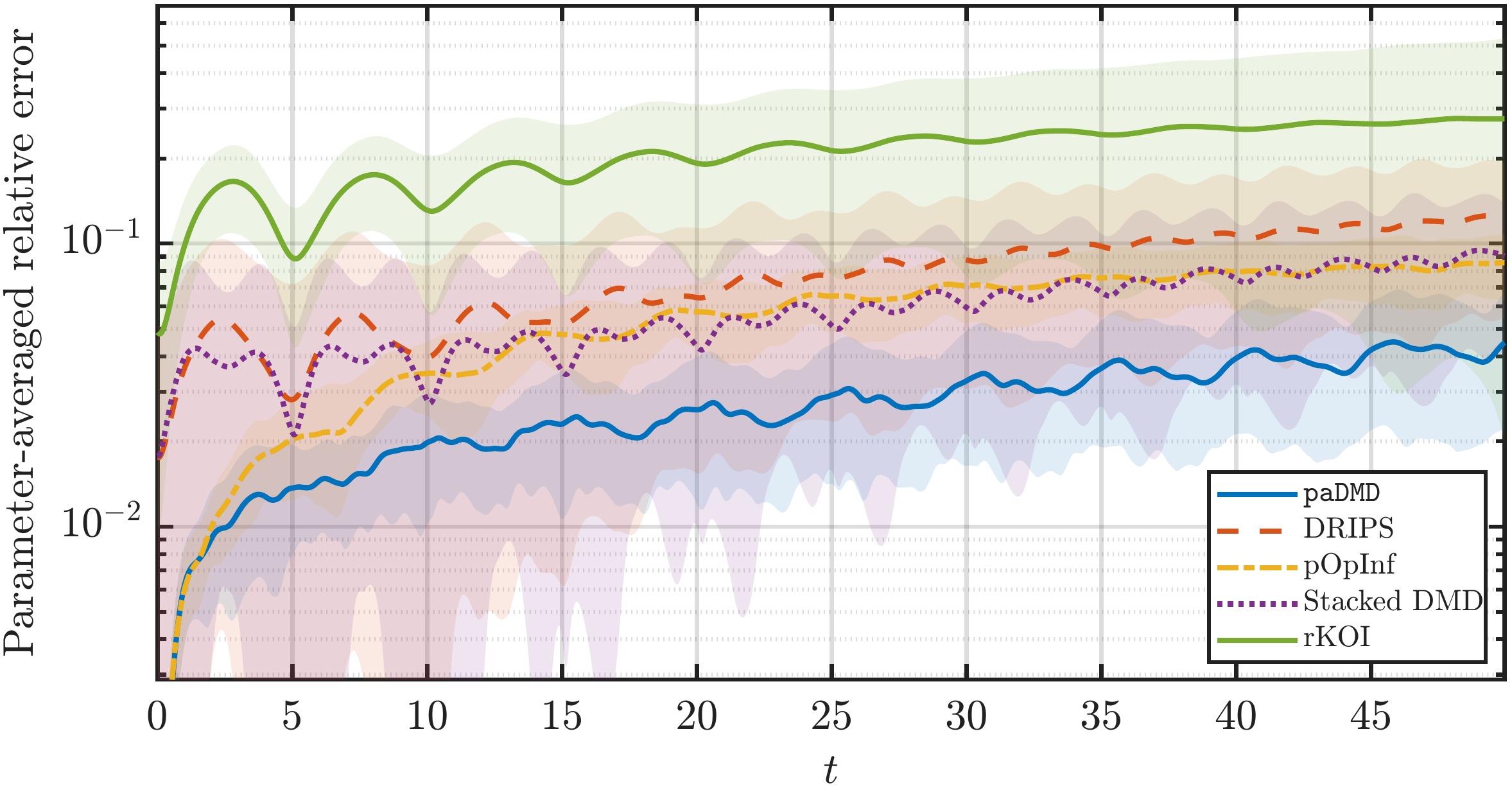}
    \caption{Relative error for the fluidic pinball system at the model ranks with lowest error, averaged over the test $[v_1,~v_2,~v_3]$ parameter values at each time step in the prediction horizon. The shaded regions indicate the standard deviation across the test parameter values.}
    \label{fig:pinball_err_time_series}
\end{figure}

\begin{figure}[pos=!h]
\centering
\subfloat[$\boldsymbol{\theta}^* = (v_1~ v_2~ v_3)^\top = (0.320~0.204~ 0.211)^\top$]{%
  \includegraphics[width=0.8\textwidth]{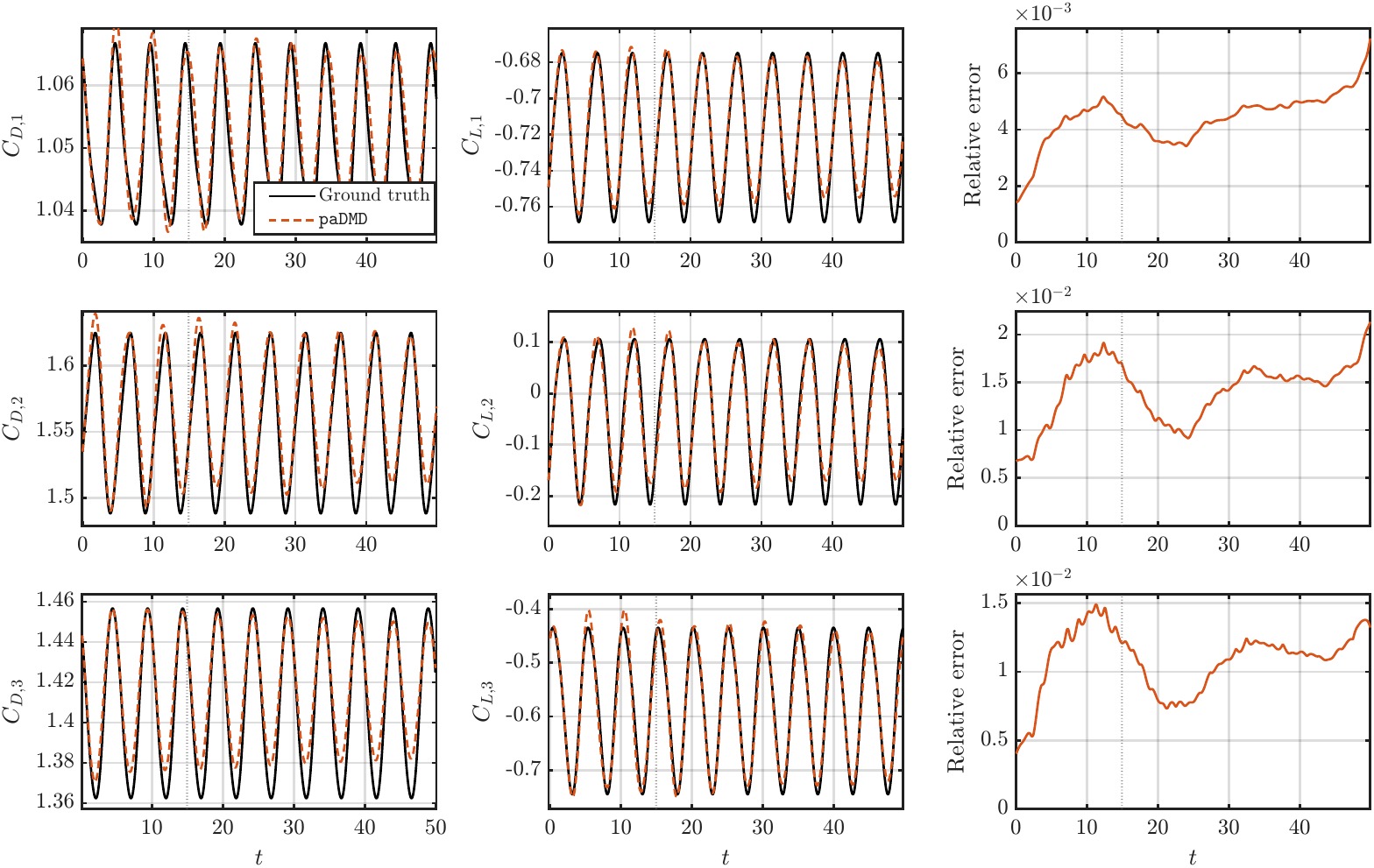}%
}

\vspace{0.8em}
\subfloat[$\boldsymbol{\theta}^* = (v_1~ v_2~ v_3)^\top = (0.588~ 0.386~ 0.346)^\top$]{%
  \includegraphics[width=0.8\textwidth]{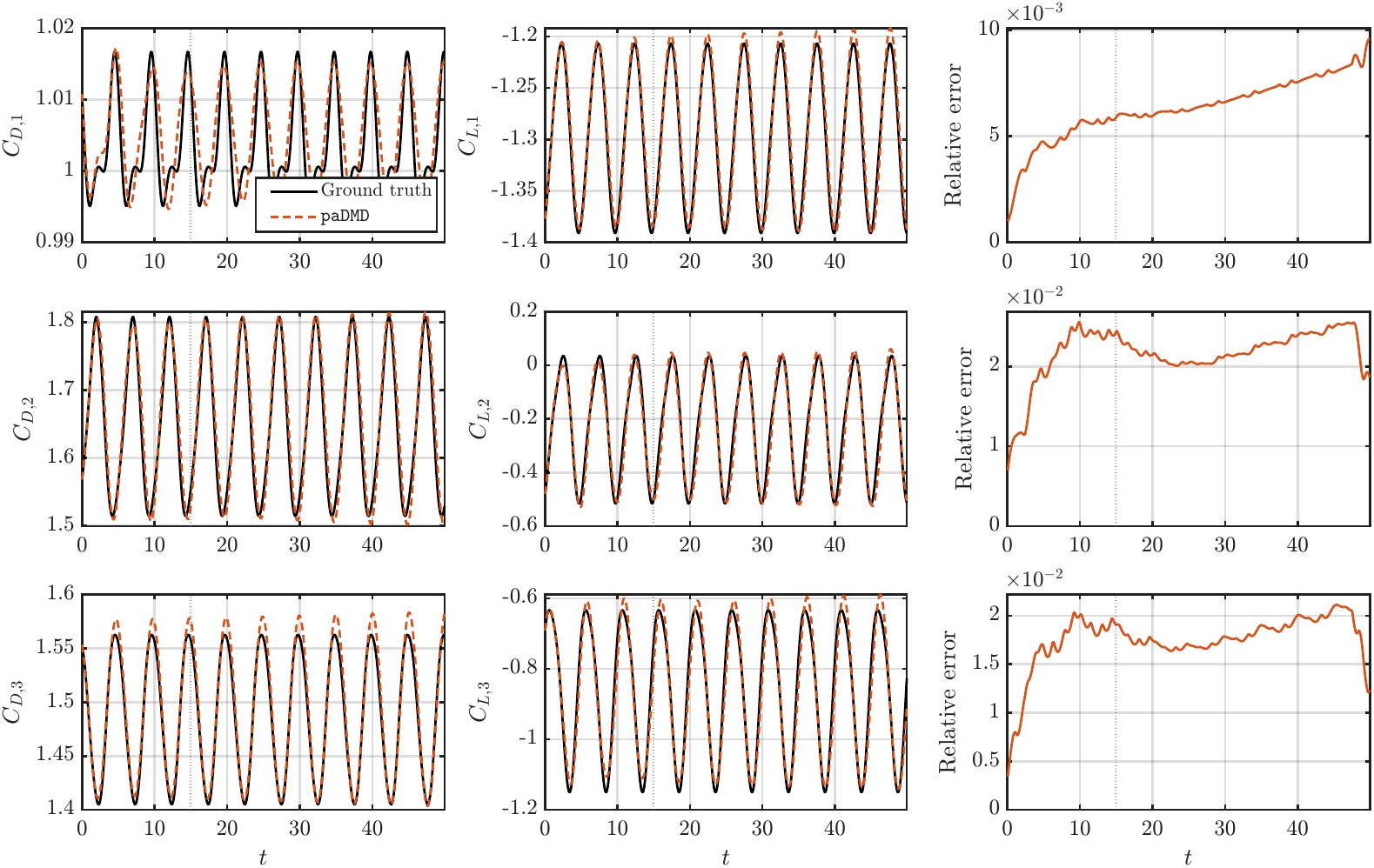}%
}
\caption{\texttt{paDMD} predictions ($\tilde r = 48,~\hat r=12$) of the drag ($C_D$) and lift ($C_L$) force coefficients acting on each cylinder for two test parameter samples of the fluidic pinball system. The first two columns compare the predicted $C_D$ and $C_L$ against the ground truth for each cylinder, and the third column shows the corresponding $2$-norm relative error over the prediction horizon.}
\label{fig:pinball_predictions}
\end{figure}

\subsection{Oscillating electron beam}\label{subsec:results_electron_beam}

\begin{figure}[pos=t]
    \centering
    \includegraphics[width=0.5\textwidth]{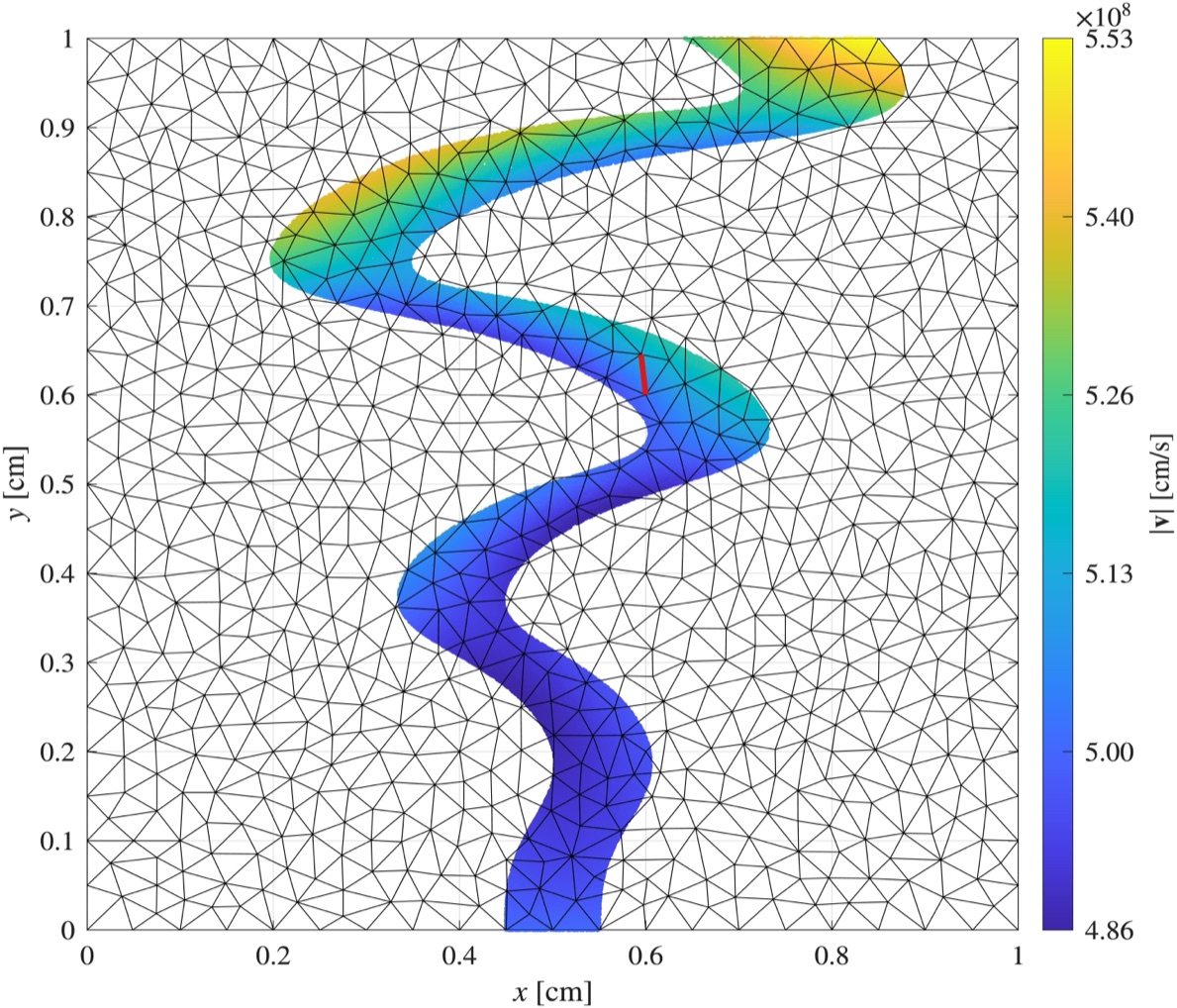}
    \caption{Electron beam snapshot at $20.8 ~\text{ns}$ with $B_0=25~\text{mT}$ and $r_{\mathrm{sp}}=2 \times 10^4$, with superparticle injection rate of 12 per time step. Each superparticle is colored based on its instantaneous speed $| \mf{v} |$. A sample mesh edge is highlighted in red (location: center-right).}
    \label{fig:wavy_particle_distrib}
\end{figure}

\begin{table}[pos=!h]
\centering

\footnotesize
\setlength{\tabcolsep}{3pt}
\caption{Mean relative error comparison for the oscillating electron beam system, averaged over the prediction horizon for each test $B_0$ value.}
\label{tab:oeb_ranks}
\begin{tabular}{lccccccccccc}
\toprule
\multicolumn{1}{c}{\multirow{2}{*}{Method}} & \multicolumn{11}{c}{Model rank ($r$)} \\
\cmidrule(lr){2-12}
 & 10 & 15 & 20 & 25 & 30 & 35 & 40 & 45 & 50 & 55 & 60 \\
\midrule
\makecell[l]{\texttt{paDMD}\\[-1pt]{\tiny $(\tilde r = r,~\hat r = \min\{r,n\})$}} & 0.0843 & 0.0515 & 0.0516 & 0.0449 & 0.0435 & 0.0424 & 0.0418 & 0.0414 & 0.0414 & 0.0417 & \textbf{0.0413} \\
\midrule
DRIPS & \textbf{0.0543} & 0.5313 & -- & -- & -- & -- & -- & -- & -- & -- & -- \\
\midrule
pOpInf & 0.1581 & -- & 0.1543 & 0.1525 & 0.1543 & 0.1537 & 0.1524 & 0.1525 & 0.1524 & -- & \textbf{0.1513} \\
\midrule
Stacked DMD & 0.0545 & 0.0537 & \textbf{0.0537} & 0.0537 & 0.0537 & 0.0538 & 0.0539 & 0.0539 & 0.0539 & 0.0540 & 0.0540 \\
\midrule
rKOI & 0.1880 & \textbf{0.1879} & 0.1881 & 0.1881 & 0.1882 & 0.1882 & 0.1882 & 0.1882 & 0.1882 & 0.1883 & 0.1883 \\
\bottomrule
\end{tabular}

\end{table}

\begin{figure}[pos=!h]
    \centering
    \includegraphics[width=0.5\textwidth]{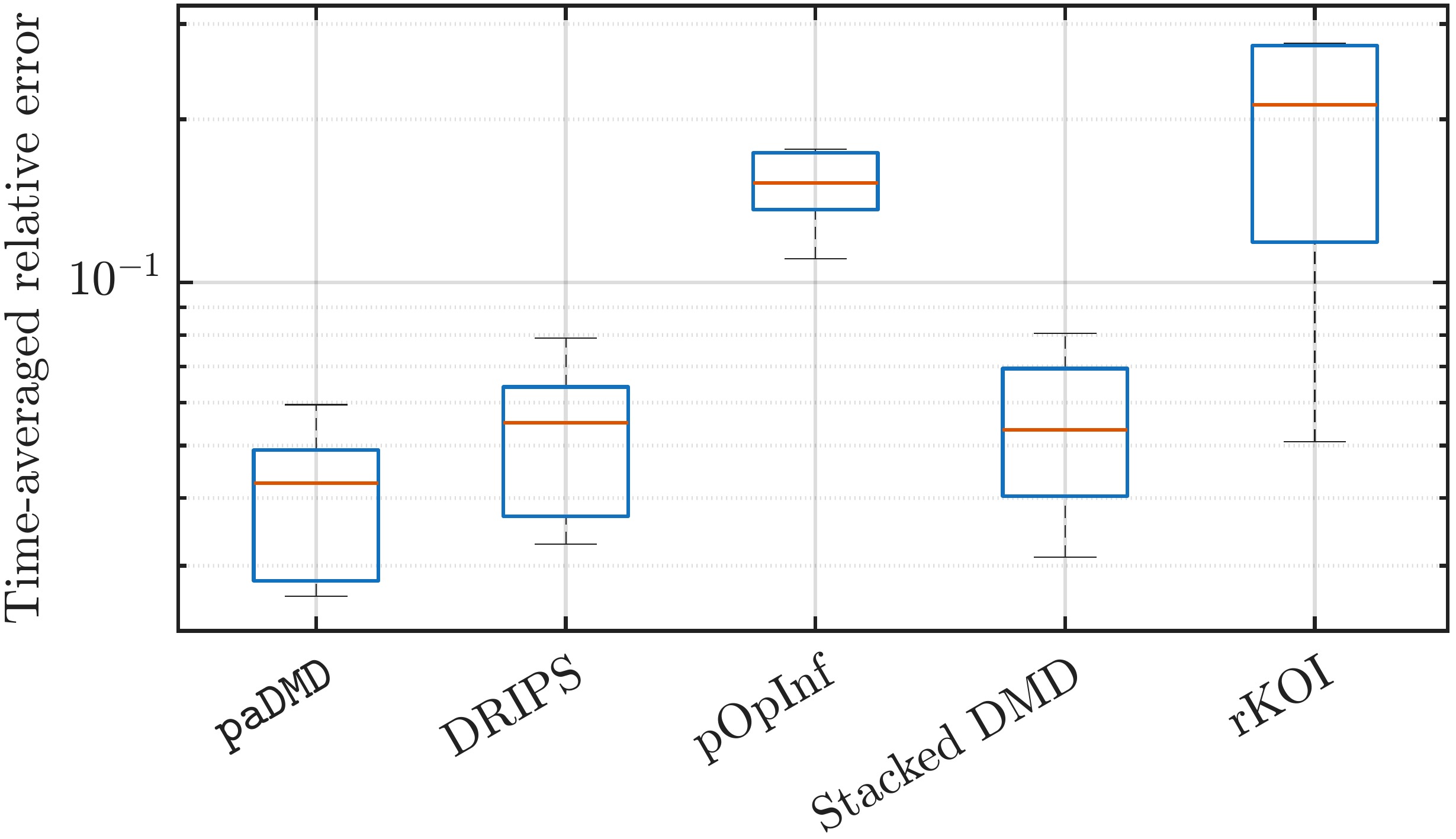}
    \caption{Distribution of the time-averaged relative error over the test $B_0$ values for the oscillating electron beam. The error is averaged over the prediction horizon at each test $B_0$ value, and each method is shown at the model rank with its lowest error.}
    \label{fig:wavy_err_boxplot}
\end{figure}

\begin{figure}[pos=!htbp]
    \centering
    \includegraphics[width=0.6\textwidth]{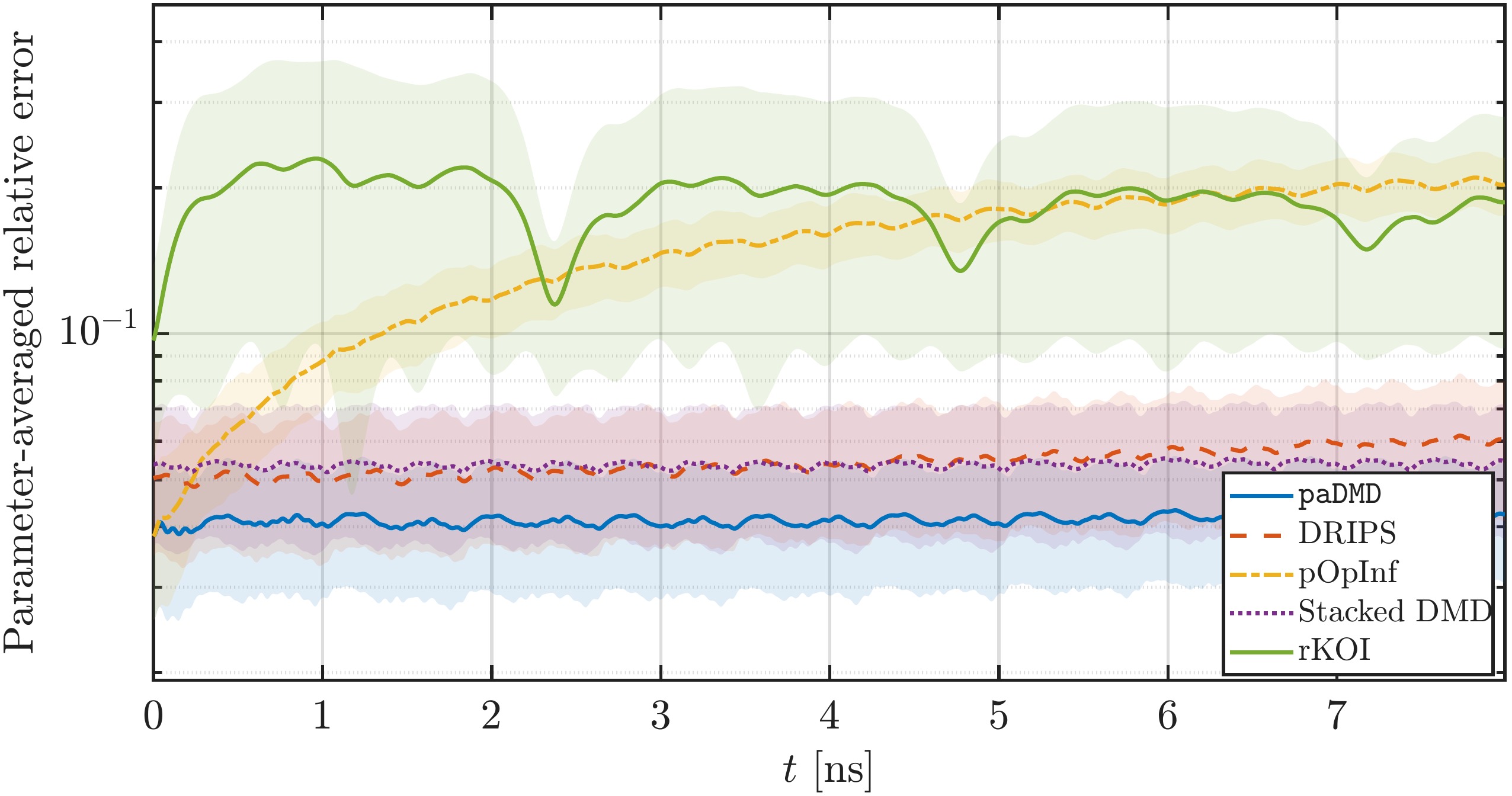}
    \caption{Relative error for the oscillating electron beam at the model ranks with lowest error, averaged over the test $B_0$ values at each time step in the prediction horizon. The shaded regions indicate the standard deviation across the test $B_0$ values.}
    \label{fig:wavy_err_time_series}
\end{figure}

\begin{figure*}[pos=!t]
\centering
\subfloat[$\boldsymbol{\theta}^*=B_0=15$ mT]{%
  \includegraphics[width=0.31\textwidth]{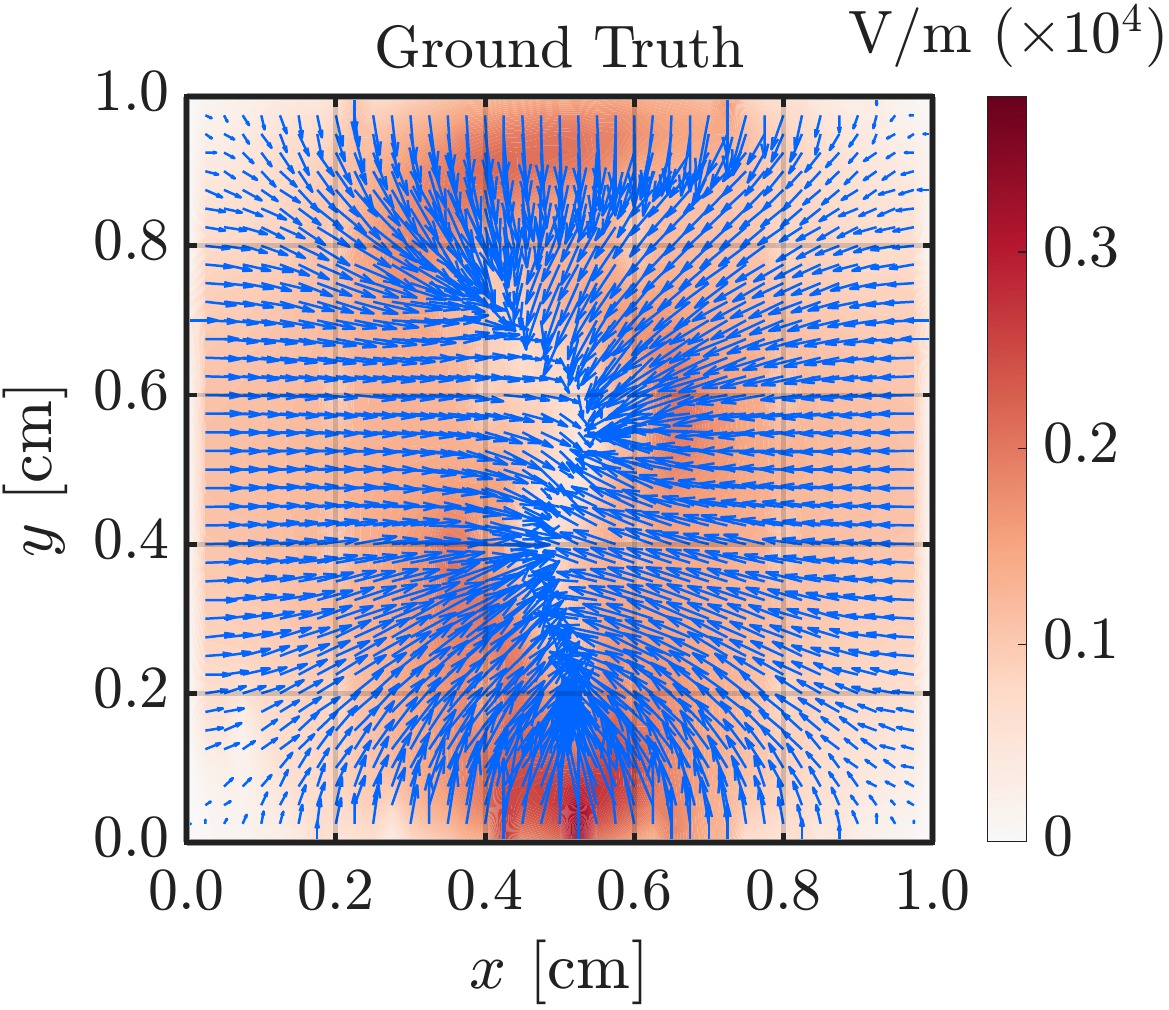}\hspace{0.02\textwidth}%
  \includegraphics[width=0.31\textwidth]{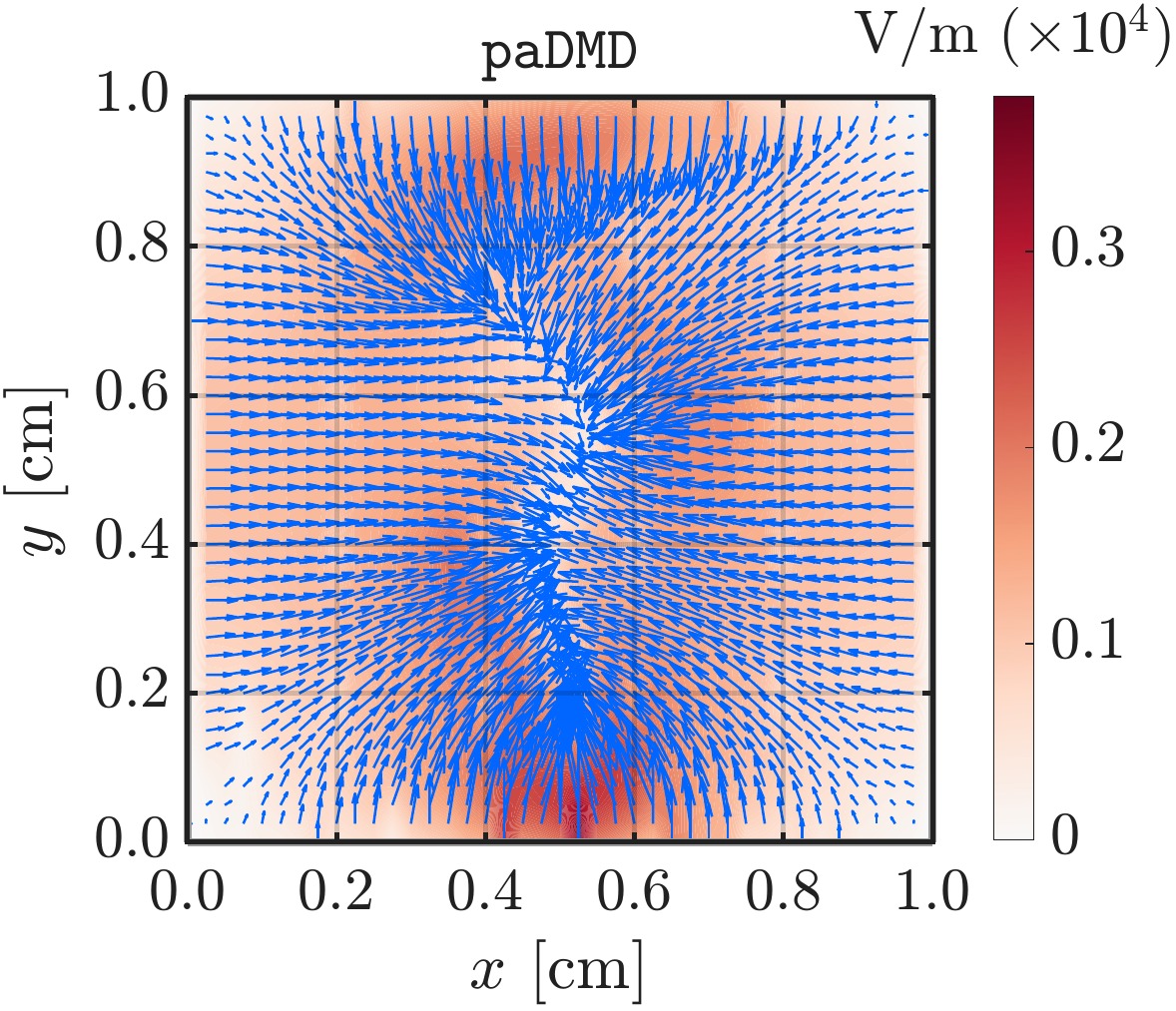}\hspace{0.02\textwidth}%
  \includegraphics[width=0.31\textwidth]{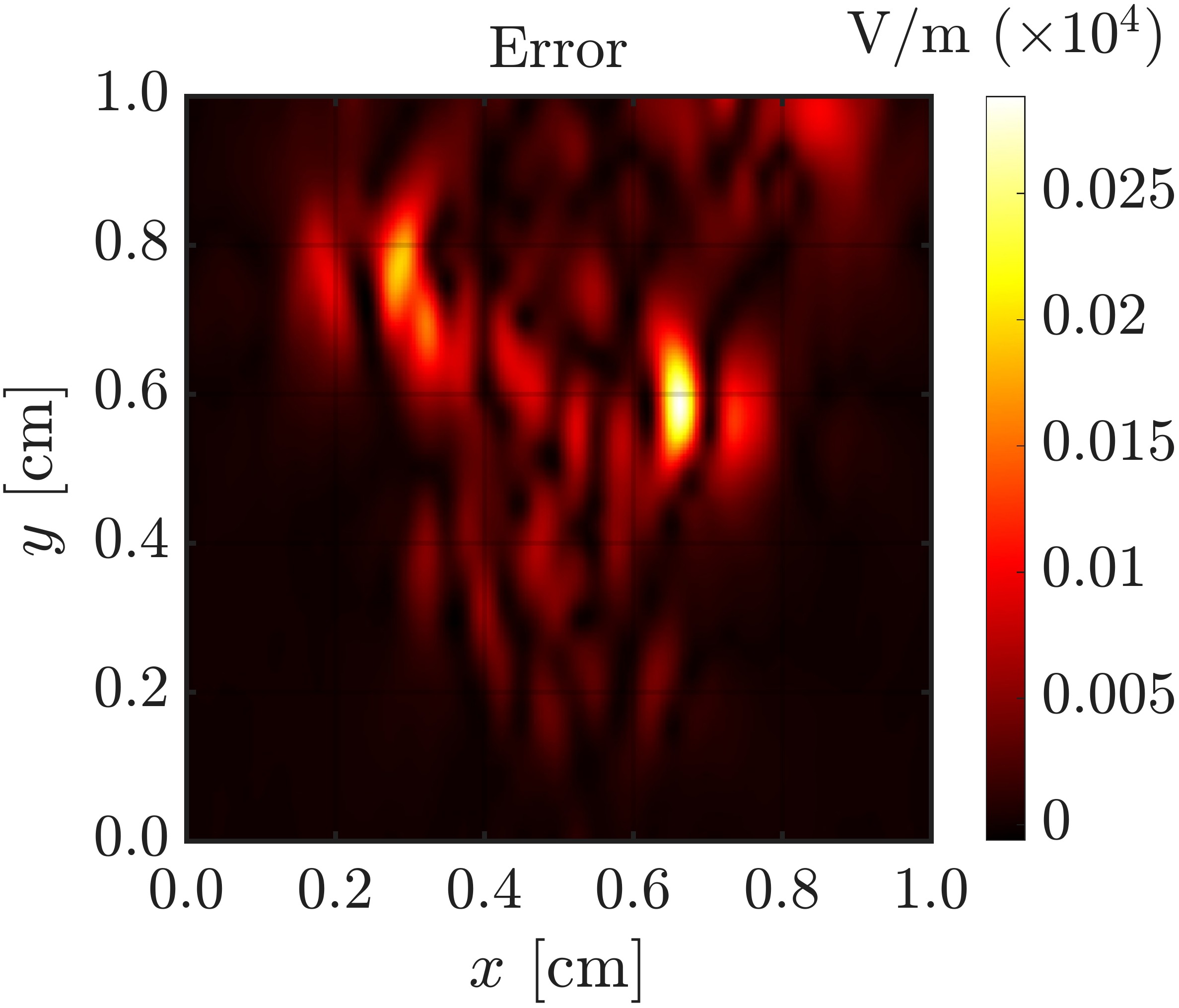}%
}

\vspace{0.8em}
\subfloat[$\boldsymbol{\theta}^*=B_0=27.5$ mT]{%
  \includegraphics[width=0.31\textwidth]{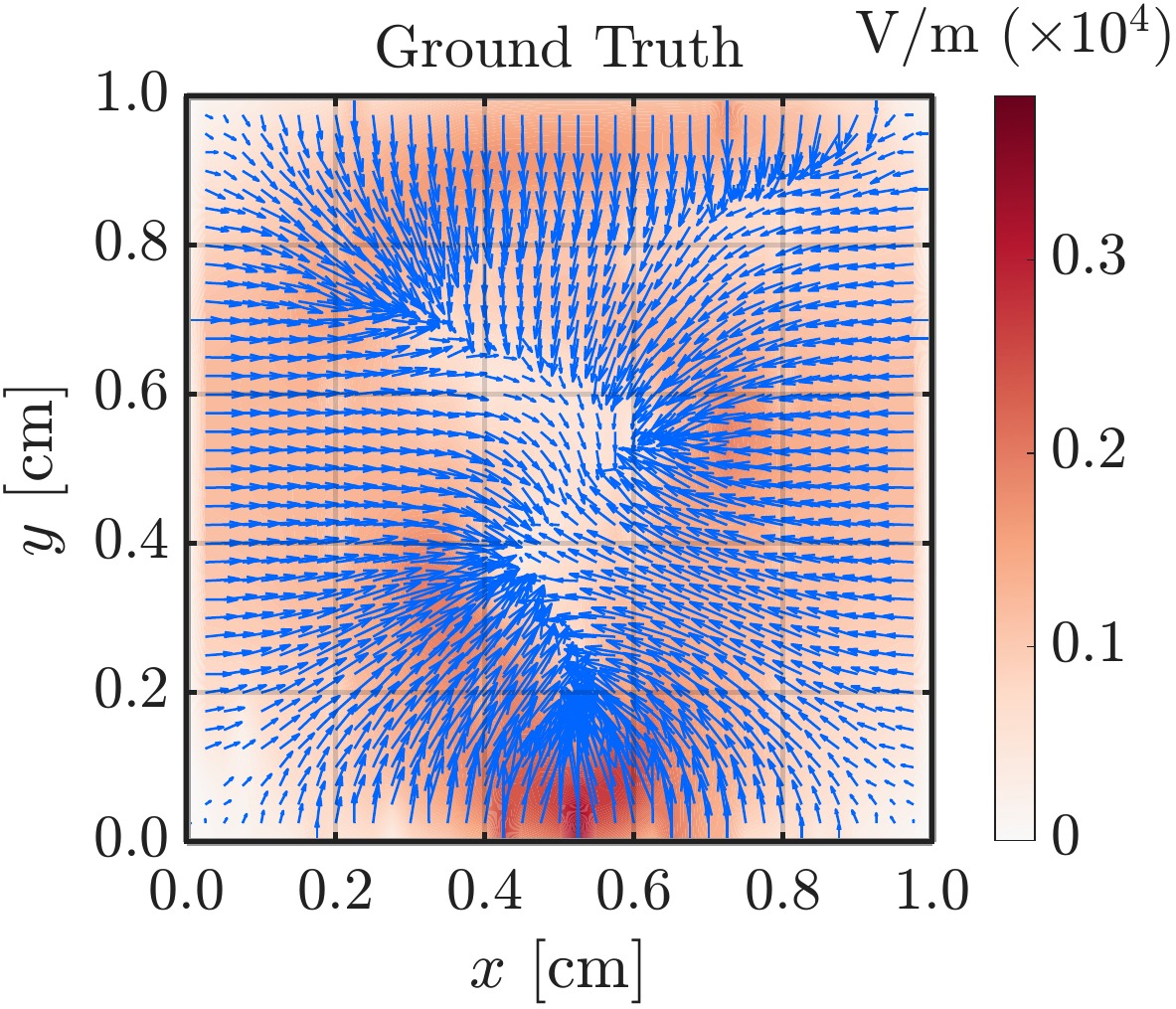}\hspace{0.02\textwidth}%
  \includegraphics[width=0.31\textwidth]{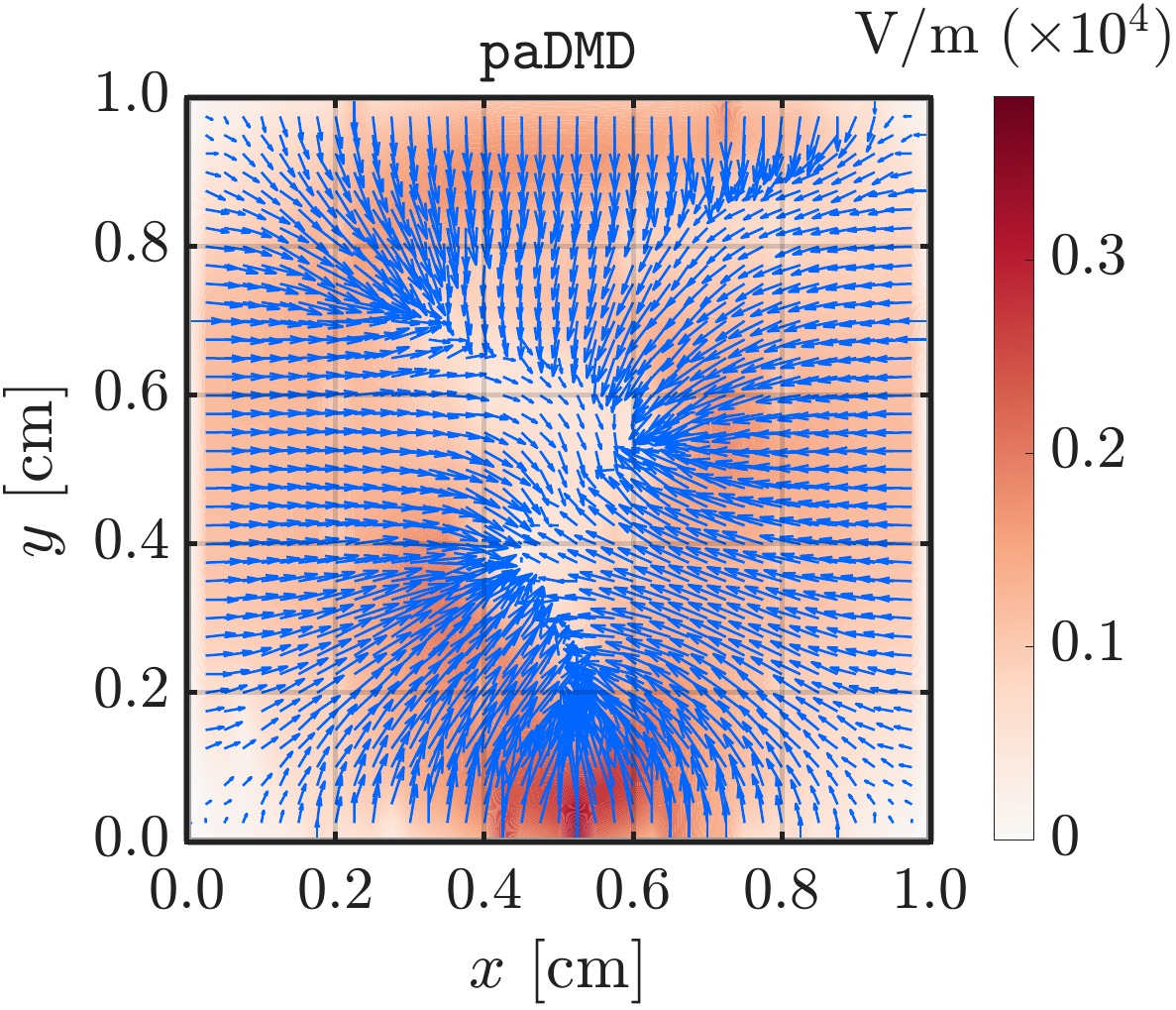}\hspace{0.02\textwidth}%
  \includegraphics[width=0.31\textwidth]{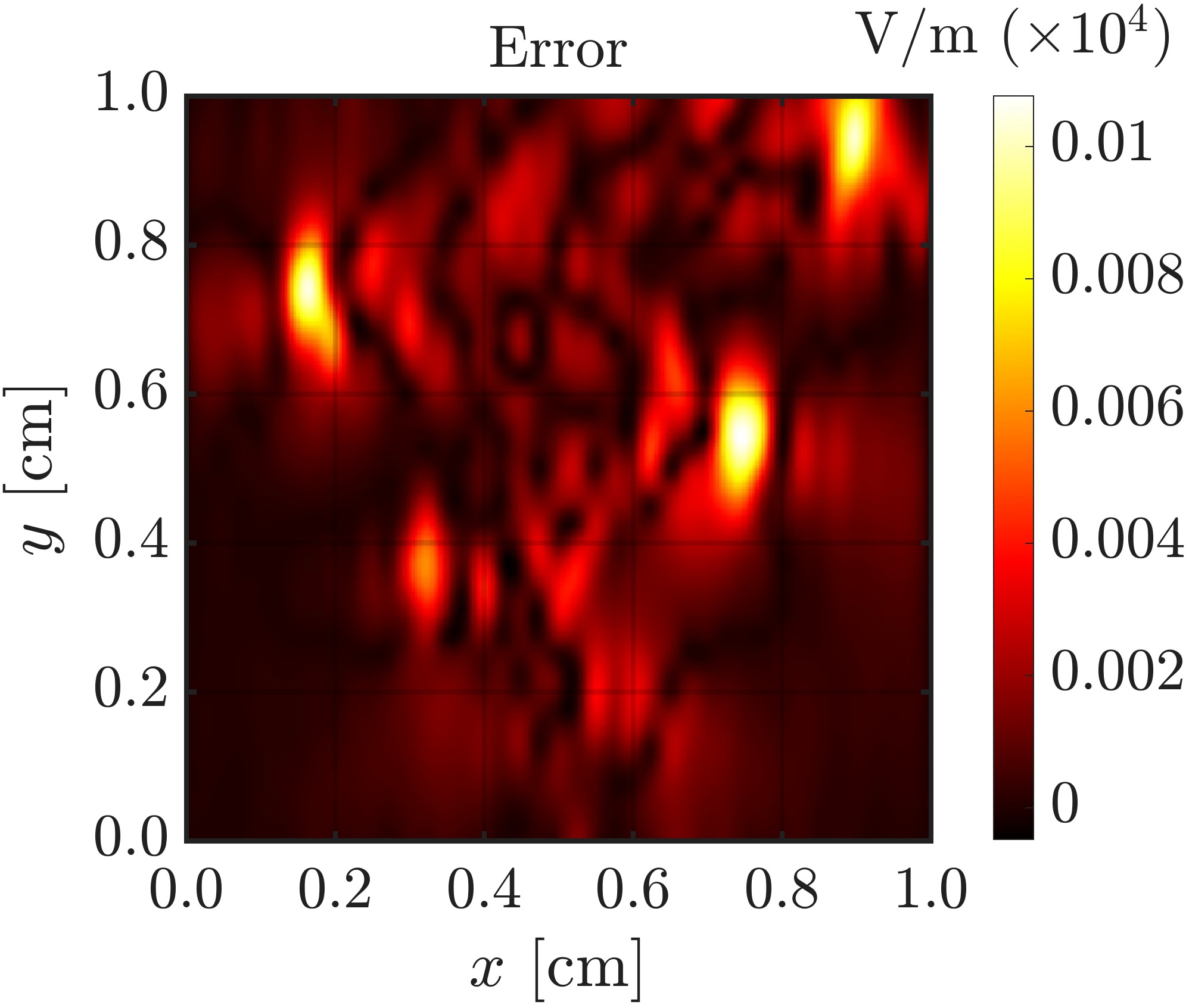}%
}

\vspace{0.8em}
\subfloat[$\boldsymbol{\theta}^*=B_0=32.5$ mT]{%
  \includegraphics[width=0.31\textwidth]{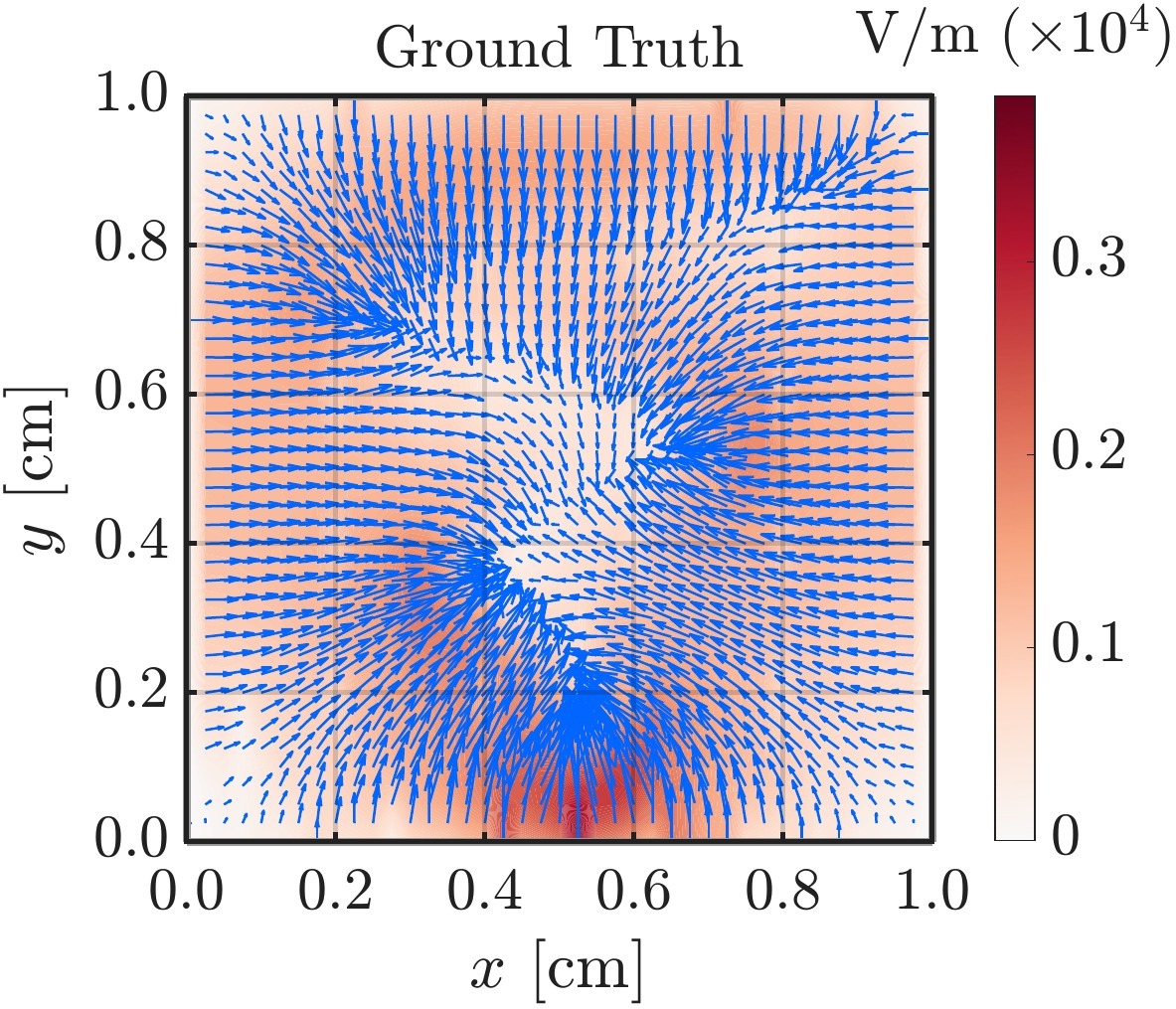}\hspace{0.02\textwidth}%
  \includegraphics[width=0.31\textwidth]{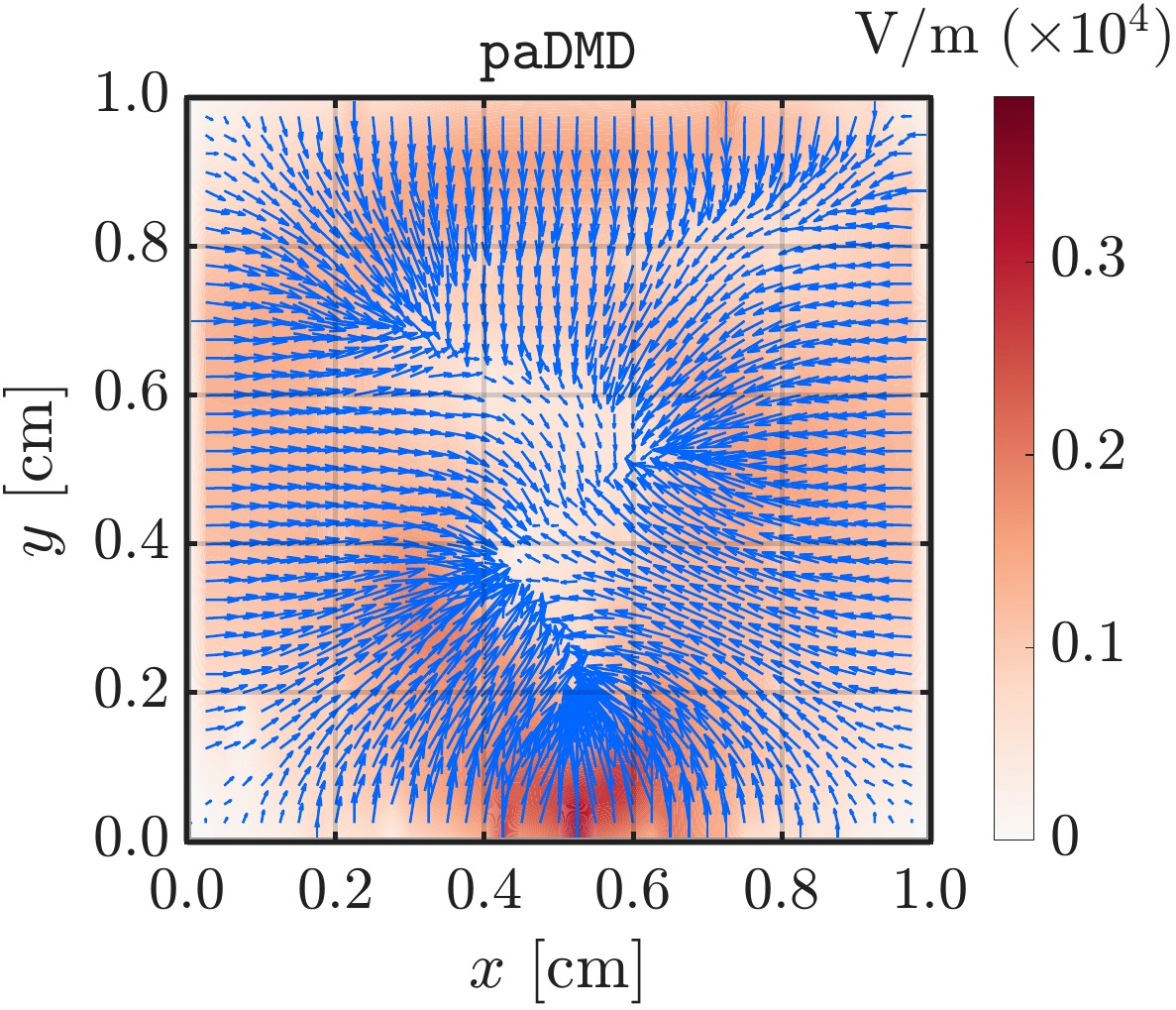}\hspace{0.02\textwidth}%
  \includegraphics[width=0.31\textwidth]{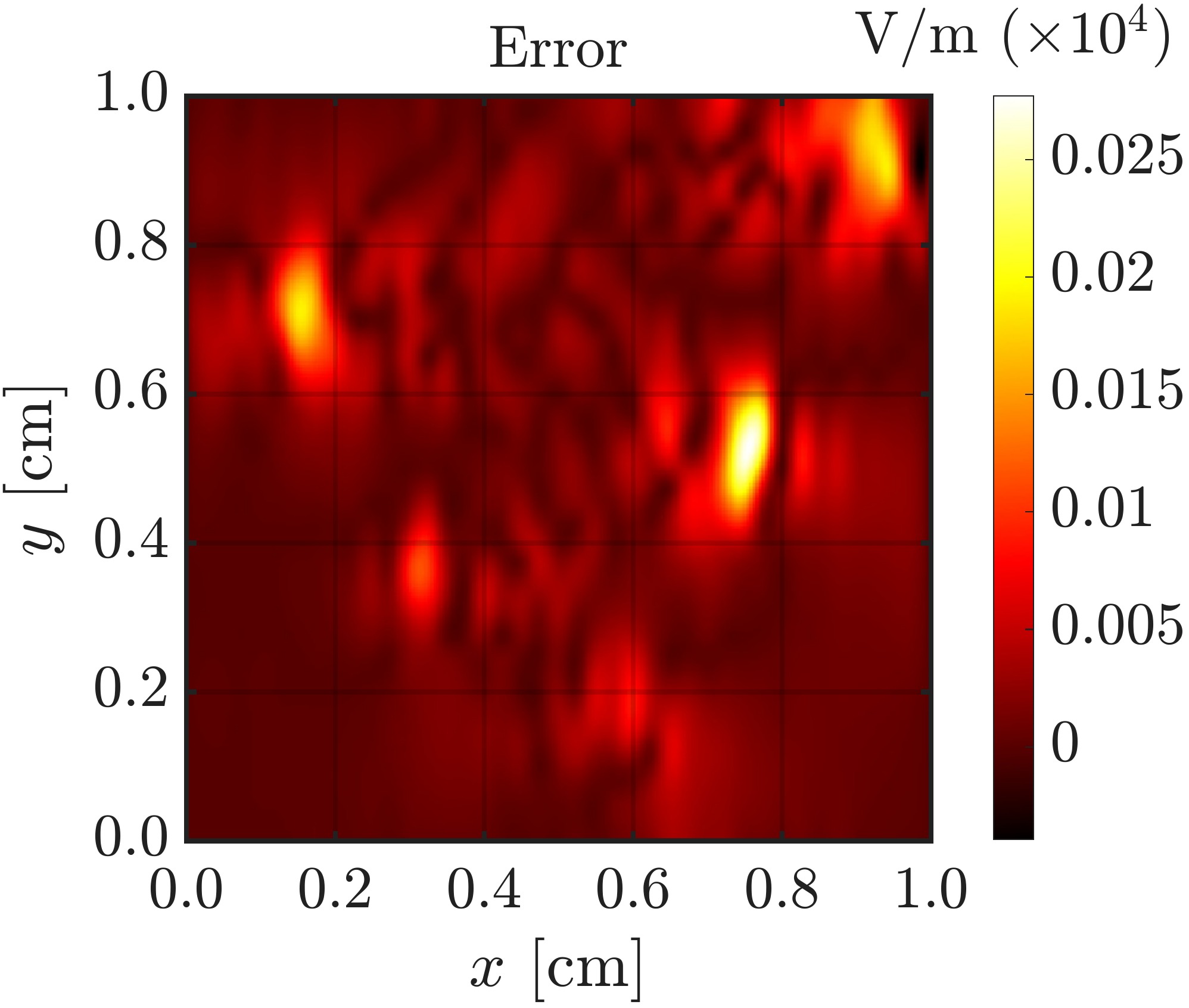}%
}

\caption{Electric field snapshots for the oscillating electron beam system at the final prediction time step for three different test values of $B_0$. Each row shows the ground truth, \texttt{paDMD} prediction ($\tilde r = \hat r=60$), and pointwise absolute error.}
\label{fig:wavy_padmd_ground_snaps}
\end{figure*}

Plasma systems are excellent candidates for data-driven modeling due to their high-dimensionality and inherent nonlinearity arising from the complex wave-particle interactions. We consider a two-dimensional (2-D) electron beam simulated inside a square cavity of dimension $1~\mathrm{cm}\times 1~\mathrm{cm}$ using a charge-conserving electromagnetic particle-in-cell (EMPIC) algorithm. More details about that algorithm, including results on grid convergence, can be found in \cite{na2016local, nayak2021koopman}. The electron beam propagates along the $+y$ direction under the influence of a time-harmonic transverse magnetic flux along the $z$ direction. The solution domain is discretized using an irregular triangular mesh with $N_0=844$ nodes, $N_1=2447$ edges, and $N_2=1604$ elements (triangles). In the EMPIC formulation the electric field is expanded in Whitney edge basis functions, so we take the state to be the electric field on the $N_1=2447$ mesh edges, giving $n=2447$ for the plasma system under consideration. Plasma self-interactions are included by applying the EMPIC algorithm coupled to a finite element field solver to the governing Maxwell-Vlasov system of equations~\cite{na2016local}. In the PIC algorithm, the phase-space distribution of electrons is represented by superparticles (point charges) with charge $q_{\mathrm{sp}}=r_{\mathrm{sp}}q_e$ and mass $m_{\mathrm{sp}}=r_{\mathrm{sp}}m_e$, where $q_e$ and $m_e$ denote the electron charge and mass, respectively. Superparticles are injected from the bottom of the cavity with a uniform spatial distribution over $[0.45~\mathrm{cm},\,0.55~\mathrm{cm}]$ at an injection rate of $12$ superparticles per time-step. The injection velocity is set to $v_y=5\times 10^6~\mathrm{m/s}$. The external oscillating magnetic flux is given by $\mathbf{B}_{\mathrm{ext}}(t)=B_0 \sin\!\left({2\pi t}/T_{\mathrm{osc}}\right)\hat{\mathbf{z}}$, with $T_{\mathrm{osc}}=0.8~\mathrm{ns}$. The time-step for the EMPIC simulation is $0.2~\mathrm{ps}$, and we save the time snapshots every 40 time steps, resulting in the DMD sampling interval of $\Delta t=8~\mathrm{ps}$. Figure~\ref{fig:wavy_particle_distrib} shows the superparticle distribution of the system at 20.8 ns, with $B_0=25~\text{mT}$, $r_{\mathrm{sp}}=2 \times 10^4$, and a superparticle injection rate of 12 per time step, with each superparticle colored according to its instantaneous speed.

Here, we first consider the transverse sinusoidal magnetic-field amplitude $B_0$ as the parameter of interest, with $\boldsymbol{\theta}=\mf h(\boldsymbol{\theta})=B_0$. Keeping the superparticle ratio $r_{\mathrm{sp}}$ fixed at $2\times 10^4$, we vary $B_0$ from $10~\mathrm{mT}$ to $40~\mathrm{mT}$ in increments of $2.5~\mathrm{mT}$, while keeping all other simulation settings fixed. We normalize the $B_0$ values to be between 0 and 0.01. After each simulation, we discard the first 500 time-steps due to the transient behavior. We choose $B_0=\{10,~25, ~40\}$ mT as our training parameter samples, and leave the rest for testing (i.e., we have $L=3$ training parameter samples and 10 test parameter samples). We choose $T=200$ snapshots from the post-transient range for each $B_0$ in the training set as our training data, and evaluate over 1000 time-steps for each test $B_0$. Table~\ref{tab:oeb_ranks} reports the mean prediction error for every method across several model ranks, with the lowest error for each method marked in bold. Figure~\ref{fig:wavy_err_boxplot} plots the time-averaged relative error for each method at the rank with the lowest error in that range, and Figure~\ref{fig:wavy_err_time_series} shows the time evolution of the relative prediction errors over the prediction horizon, averaged across the test $B_0$ values. The ground truth vs prediction comparisons for \texttt{paDMD} at various test values of $B_0$ are shown in Figures \ref{fig:wavy_padmd_ground_snaps} and \ref{fig:wavy_dof_err}.

\begin{figure}[pos=!t]
\centering
\includegraphics[width=0.31\textwidth]{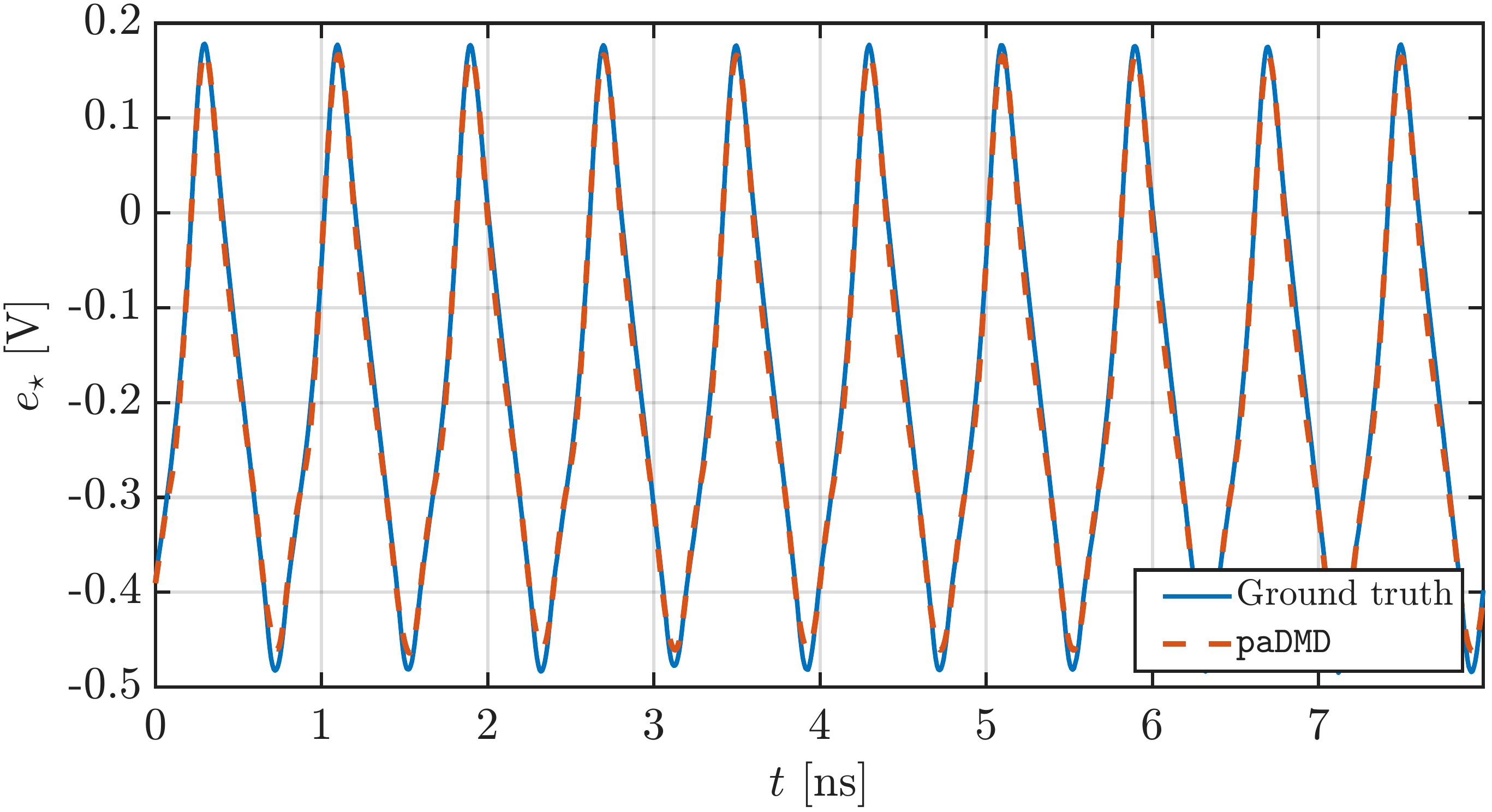}\hspace{0.02\textwidth}%
\includegraphics[width=0.31\textwidth]{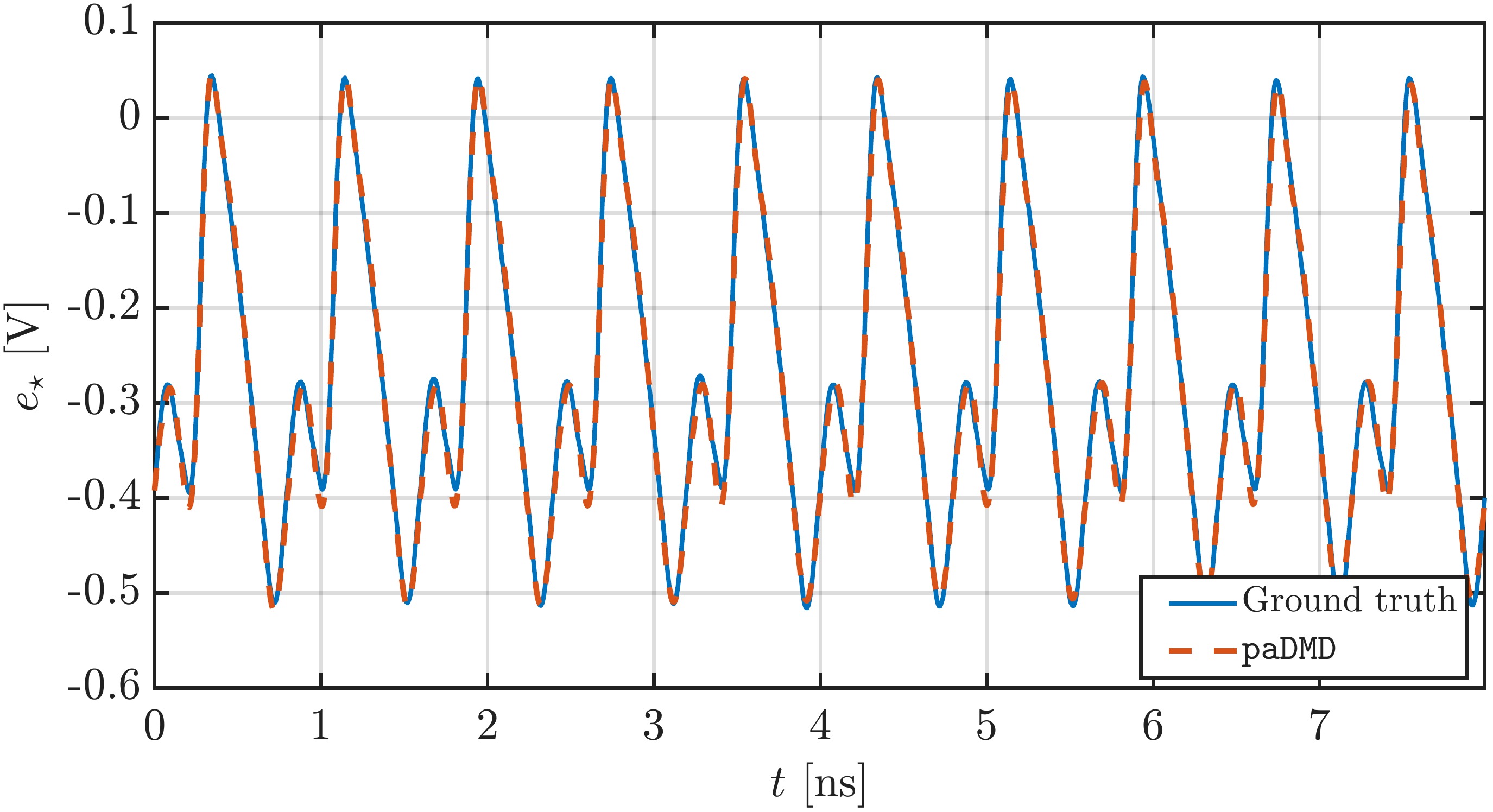}\hspace{0.02\textwidth}%
\includegraphics[width=0.31\textwidth]{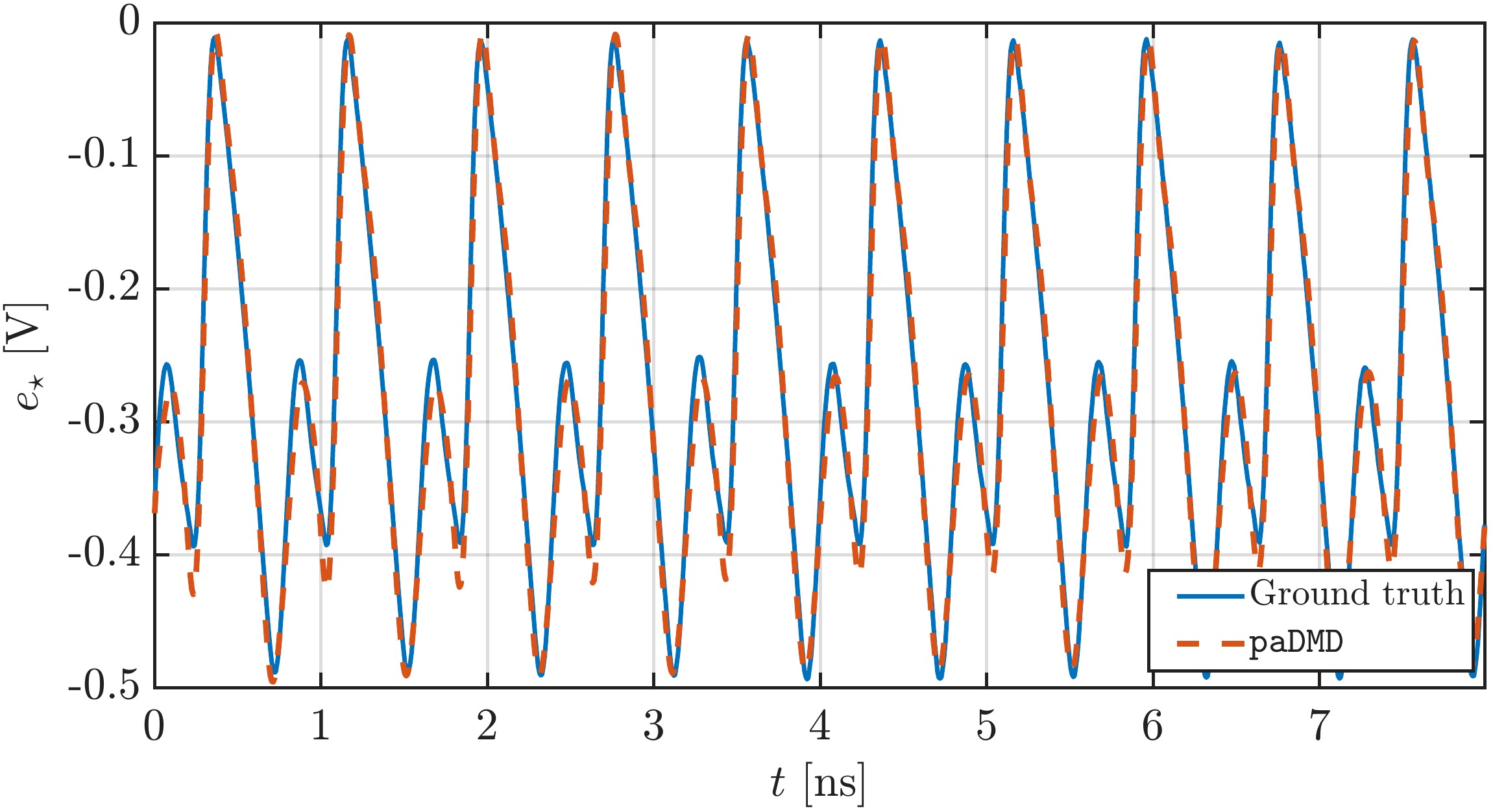}

\vspace{0.6em}
\subfloat[$\boldsymbol{\theta}^*=B_0 = 15$ mT]{%
  \includegraphics[width=0.31\textwidth]{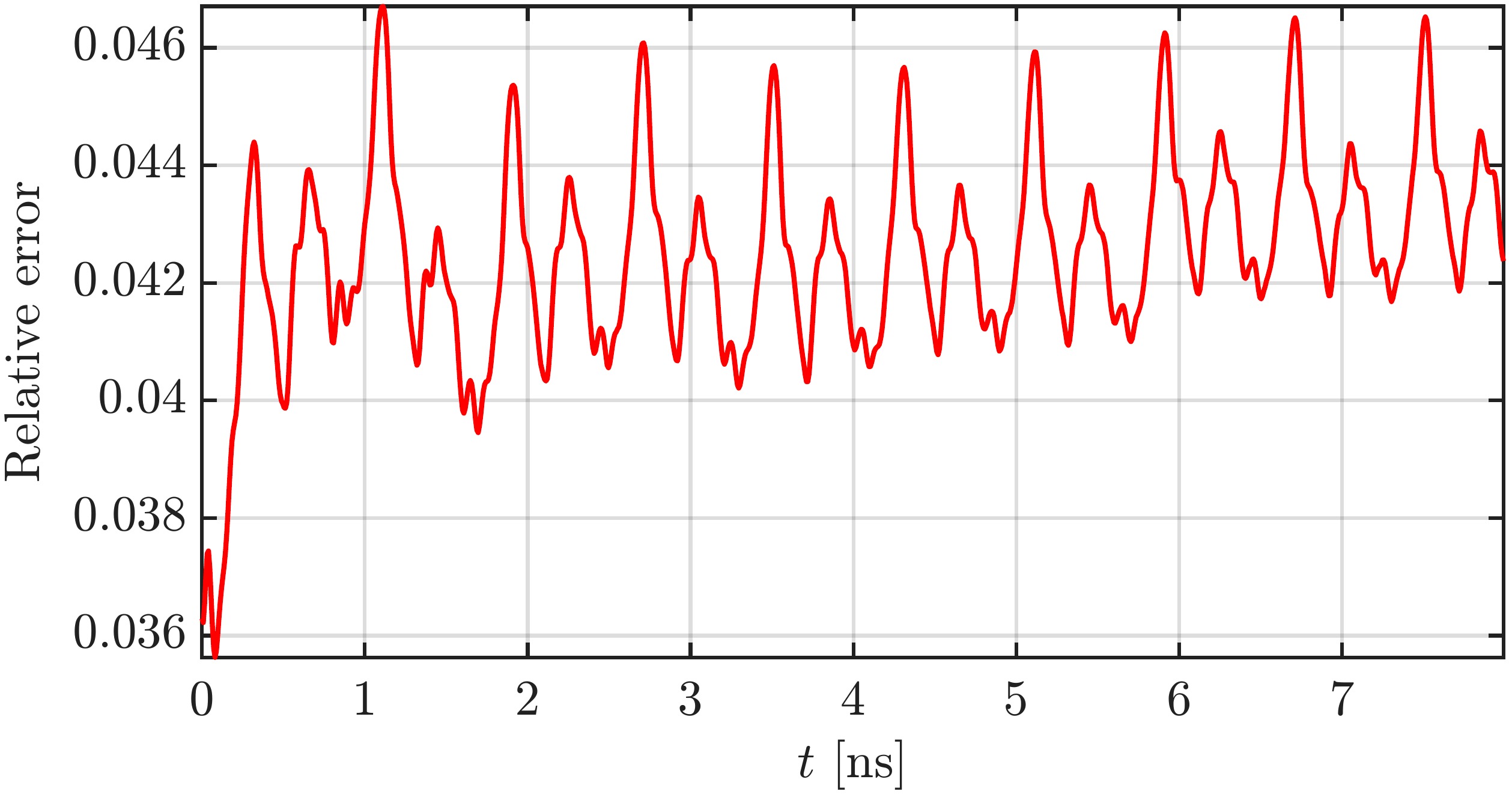}%
}\hspace{0.02\textwidth}
\subfloat[$\boldsymbol{\theta}^*=B_0 = 27.5$ mT]{%
  \includegraphics[width=0.31\textwidth]{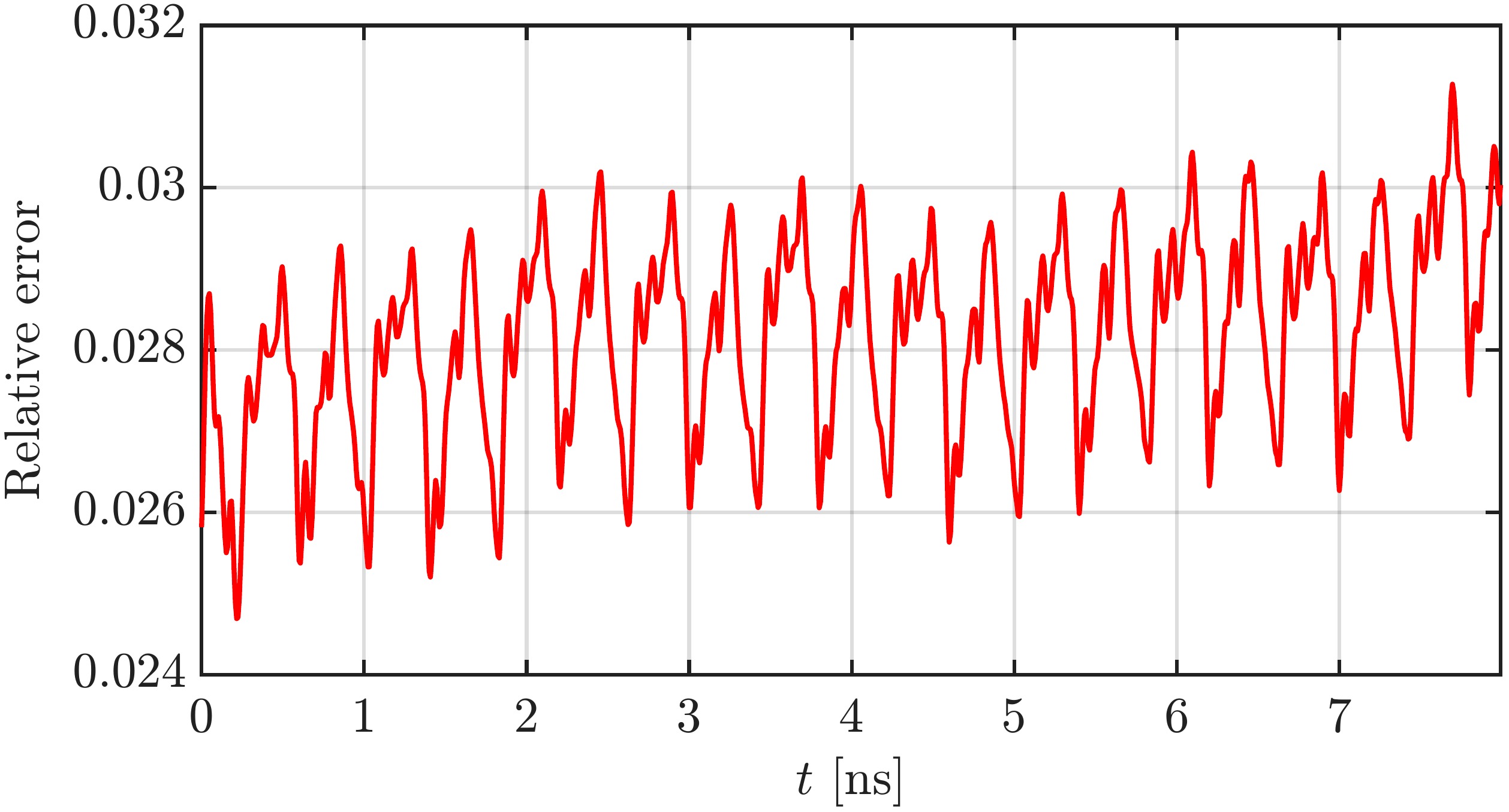}%
}\hspace{0.02\textwidth}
\subfloat[$\boldsymbol{\theta}^*=B_0 = 32.5$ mT]{%
  \includegraphics[width=0.31\textwidth]{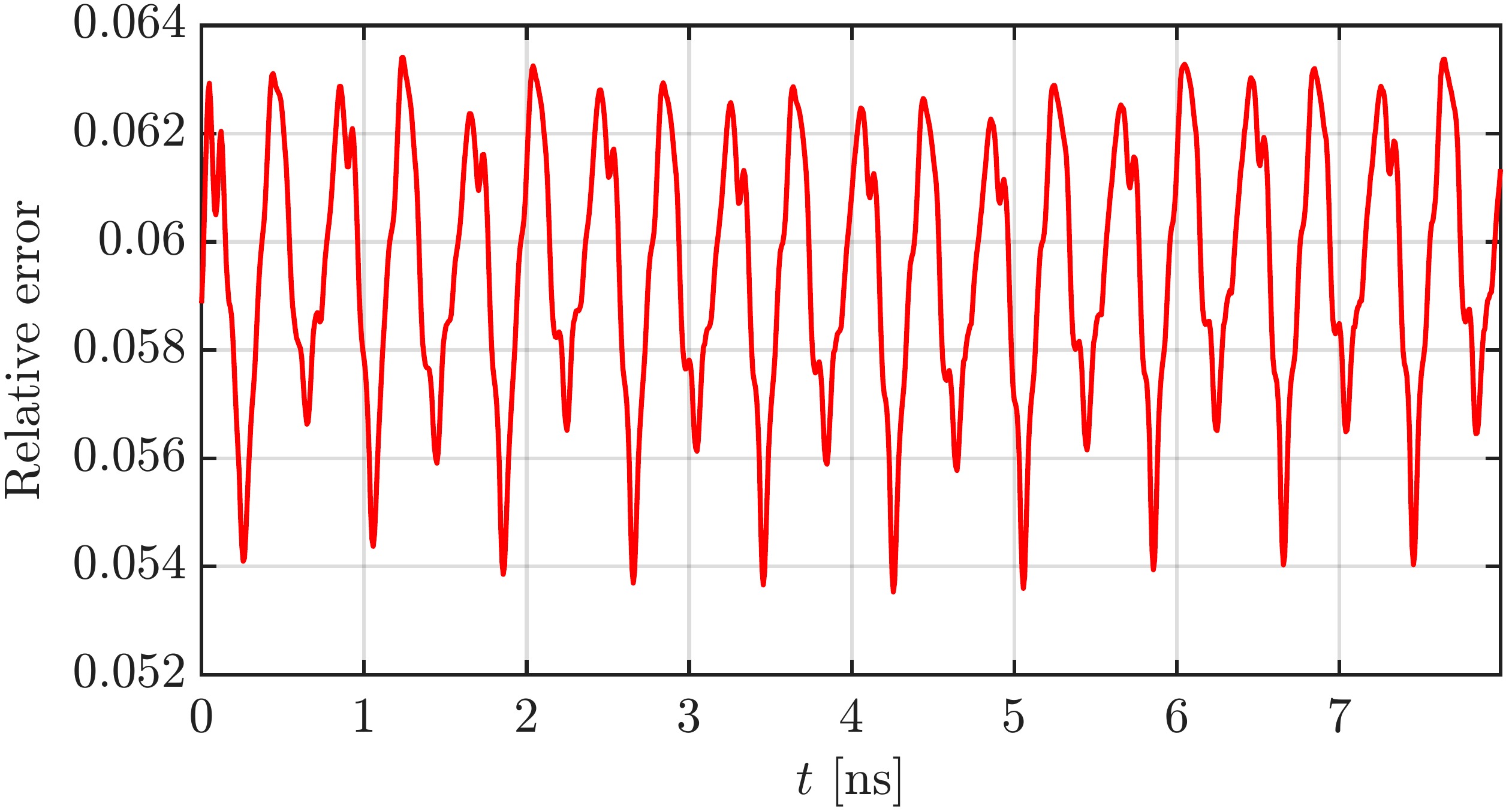}%
}
\caption{\emph{Top row}: \texttt{paDMD} prediction results ($\tilde r = \hat r=60$) of the electric field for the oscillating electron beam at three test values of $B_0$, sampled at the location (mesh edge) highlighted in red in Figure~\ref{fig:wavy_particle_distrib}. \emph{Bottom row}: Corresponding $2$-norm relative errors of the \texttt{paDMD} electric field prediction over the prediction horizon.}
\label{fig:wavy_dof_err}
\end{figure}

Next, we consider a more challenging two-dimensional parameter setting (i.e., we have two parameters of interest) where $\boldsymbol{\theta}=\mf h(\boldsymbol{\theta})=[B_0 \quad \mc{N}_e]^\top$, with $\mc{N}_e$ being the electron injection rate. We select $B_0$ values from the set $B_0 \in \{15,~20,~25,~30,~35\}$ mT. To vary $\mc{N}_e$, we vary the superparticle ratio $r_{\mathrm{sp}}$ by choosing values from the set $r_{\mathrm{sp}} \in \{1\times10^4,~1.5\times10^4,~2\times10^4,~2.5\times10^4,~3\times10^4\}$. Since we inject 12 superparticles per time step, this gives us varying $\mc{N}_e \in \{1.2\times10^5,~1.8\times10^5,~2.4\times10^5,~3.0\times10^5,~3.6\times10^5\}$ per time step. Therefore, we have 25 total parameter samples of interest. Each of the $B_0$ and $\mc{N}_e$ parameters is normalized to be between 0 and 0.01. To maintain persistency of excitation across both $B_0$ and $\mc{N}_e$, out of those 25 parameter samples we select a small set of $L=4$ values of $\boldsymbol{\theta}$ spanning the range of both $B_0$ and $\mc{N}_e$ as our training samples, and leave the rest for testing (i.e., we have 4 training and 21 test parameter samples). As before, we discard the first 500 time steps due to transients, choose the next 200 time steps from each training sample (i.e., $T=200$), and predict over a 1000 time step horizon for each test parameter. Table~\ref{tab:oeb2p_ranks} reports the mean prediction error for every method across several model ranks, with the lowest error for each method marked in bold. Interestingly, we observe that rKOI fails to converge for several test parameter samples for every model rank in the chosen range, and the time-averaged relative error blows up. This is likely because the interpolation of the DMD operators becomes unreliable when the training samples are few and scattered in more than one parameter dimension, particularly for test samples lying outside the region the training samples cover. Figure~\ref{fig:wavy_2p_err_boxplot} plots the time-averaged relative error for the other methods at the rank with the lowest error in that range, and Figure~\ref{fig:wavy_2p_err_time_series} shows the time evolution of the relative prediction errors over the prediction horizon, averaged across the test $[B_0~~\mc{N}_e]^\top$ values. Figure~\ref{fig:wavy_checkerboards_2param} shows the distribution of time-averaged relative errors for the test parameter samples across the chosen range of the two-dimensional parameter space. The four training samples do not span the full extent of this range, so a number of the test samples require the parametric models to be evaluated outside the region those samples cover. We retain all of them in the reported errors. \texttt{paDMD} keeps a uniformly low error across the whole parameter space, including these samples, whereas the interpolation-based methods degrade markedly at the test samples furthest from the training set. This is consistent with the benefit of evaluating a single learned parameter-affine operator at the requested parameter, rather than forming a weighted combination of models built at the training samples, over the range tested here.

\begin{table}[pos=!h]
\centering

\footnotesize
\setlength{\tabcolsep}{3pt}
\caption{Mean relative prediction errors over the test $[B_0~~\mc{N}_e]^\top$ parameter samples of the oscillating electron beam, across various model ranks. The best result for each method is marked in bold, and `--' indicates error blow up.}
\label{tab:oeb2p_ranks}
\begin{tabular}{lccccccccccc}
\toprule
\multicolumn{1}{c}{\multirow{2}{*}{Method}} & \multicolumn{11}{c}{Model rank ($r$)} \\
\cmidrule(lr){2-12}
 & 10 & 15 & 20 & 25 & 30 & 35 & 40 & 45 & 50 & 55 & 60 \\
\midrule
\makecell[l]{\texttt{paDMD}\\[-1pt]{\tiny $(\tilde r = r,~\hat r = \min\{r,n\})$}} & 0.0836 & 0.0627 & 0.0591 & 0.0573 & 0.0567 & 0.0545 & 0.0526 & 0.0514 & 0.0513 & 0.0510 & \textbf{0.0503} \\
\midrule
DRIPS & \textbf{0.1472} & -- & -- & -- & -- & -- & -- & -- & -- & -- & -- \\
\midrule
pOpInf & 0.1699 & -- & 0.1664 & -- & 0.1737 & 0.1620 & -- & 0.1644 & 0.1648 & \textbf{0.1566} & 0.1614 \\
\midrule
Stacked DMD & 0.0945 & 0.0941 & \textbf{0.0940} & 0.0943 & 0.0950 & 0.0952 & 0.0946 & 0.0945 & 0.0946 & 0.0946 & 0.0945 \\
\midrule
rKOI & -- & -- & -- & -- & -- & -- & -- & -- & -- & -- & -- \\
\bottomrule
\end{tabular}
\end{table}

\begin{figure}[pos=!h]
    \centering
    \includegraphics[width=0.5\textwidth]{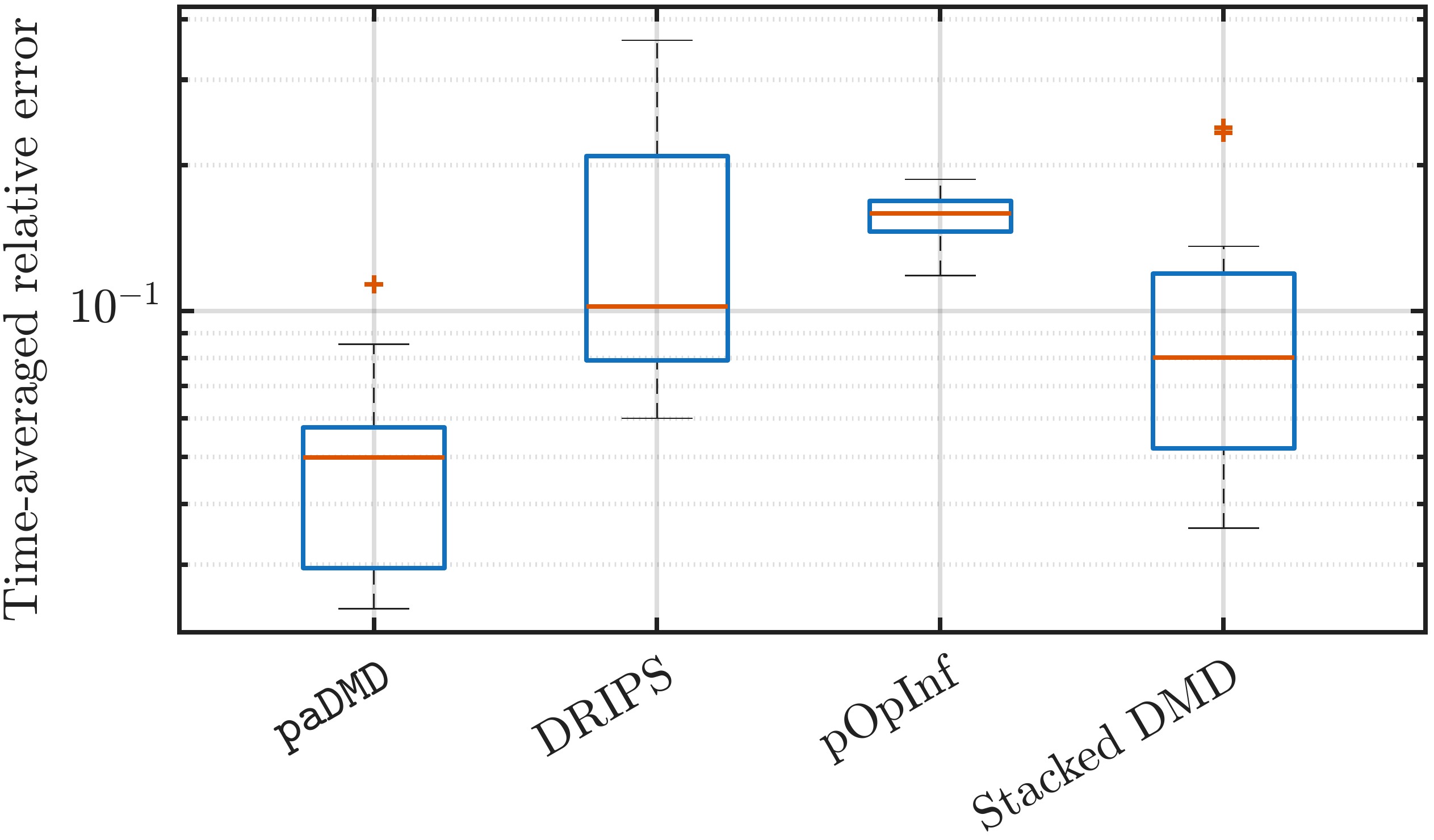}
    \caption{Distribution of the time-averaged relative error over the test $[B_0~~\mc{N}_e]^\top$ values for the oscillating electron beam. The error is averaged over the prediction horizon at each test $[B_0~~\mc{N}_e]^\top$ value, and each converging method is shown at the model rank with its lowest error.}
    \label{fig:wavy_2p_err_boxplot}
\end{figure}

\begin{figure}[pos=!h]
    \centering
    \includegraphics[width=0.6\textwidth]{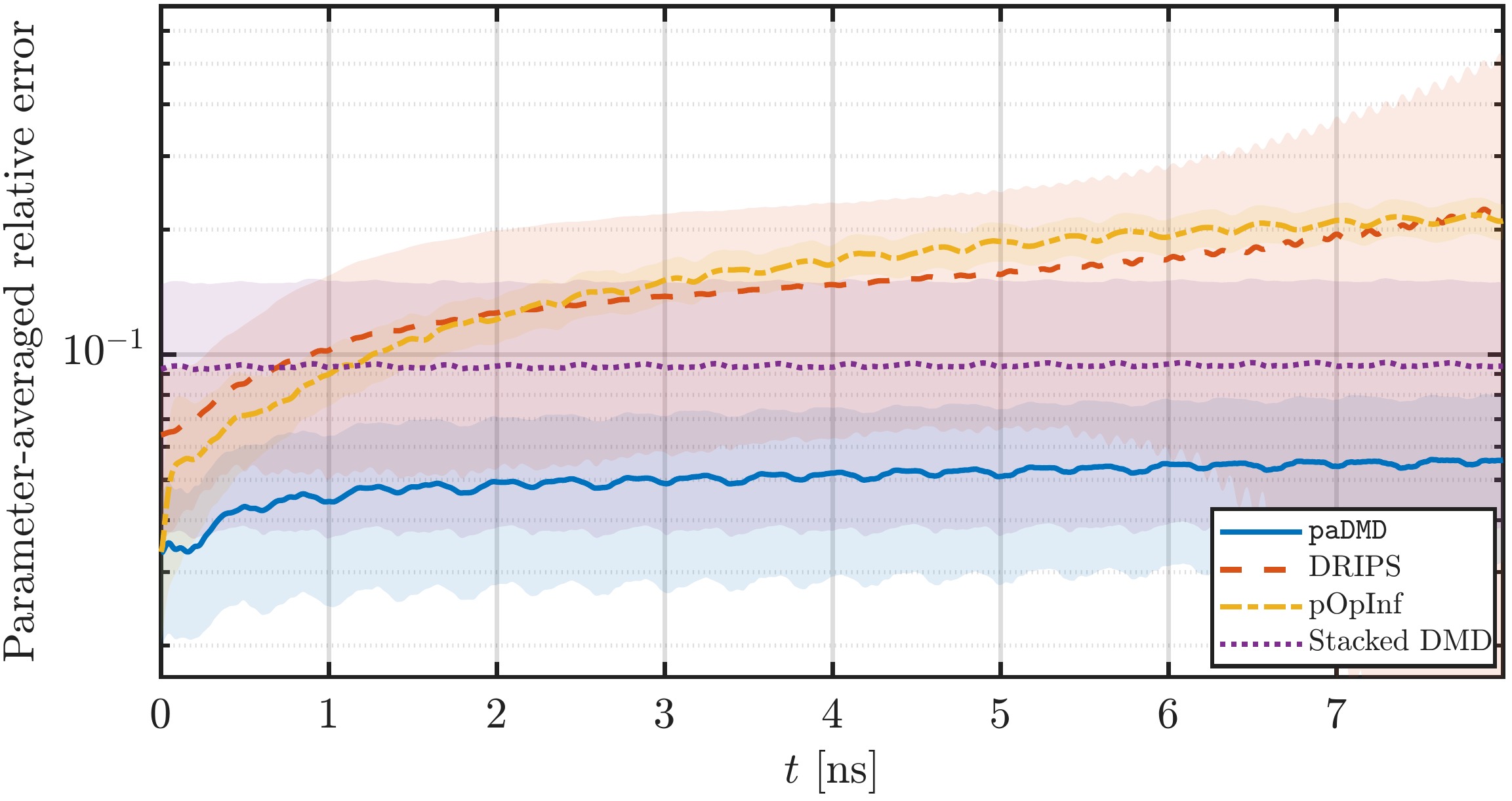}
    \caption{Relative error for the oscillating electron beam at the ranks with lowest error for each converging method, averaged over the test $[B_0~~\mc{N}_e]^\top$ values at each time step in the prediction horizon. The shaded regions indicate the standard deviation across the test parameter values.}
    \label{fig:wavy_2p_err_time_series}
\end{figure}

\begin{figure}[pos=!t]
    \centering
    \subfloat[]{%
        \includegraphics[width=0.35\textwidth]{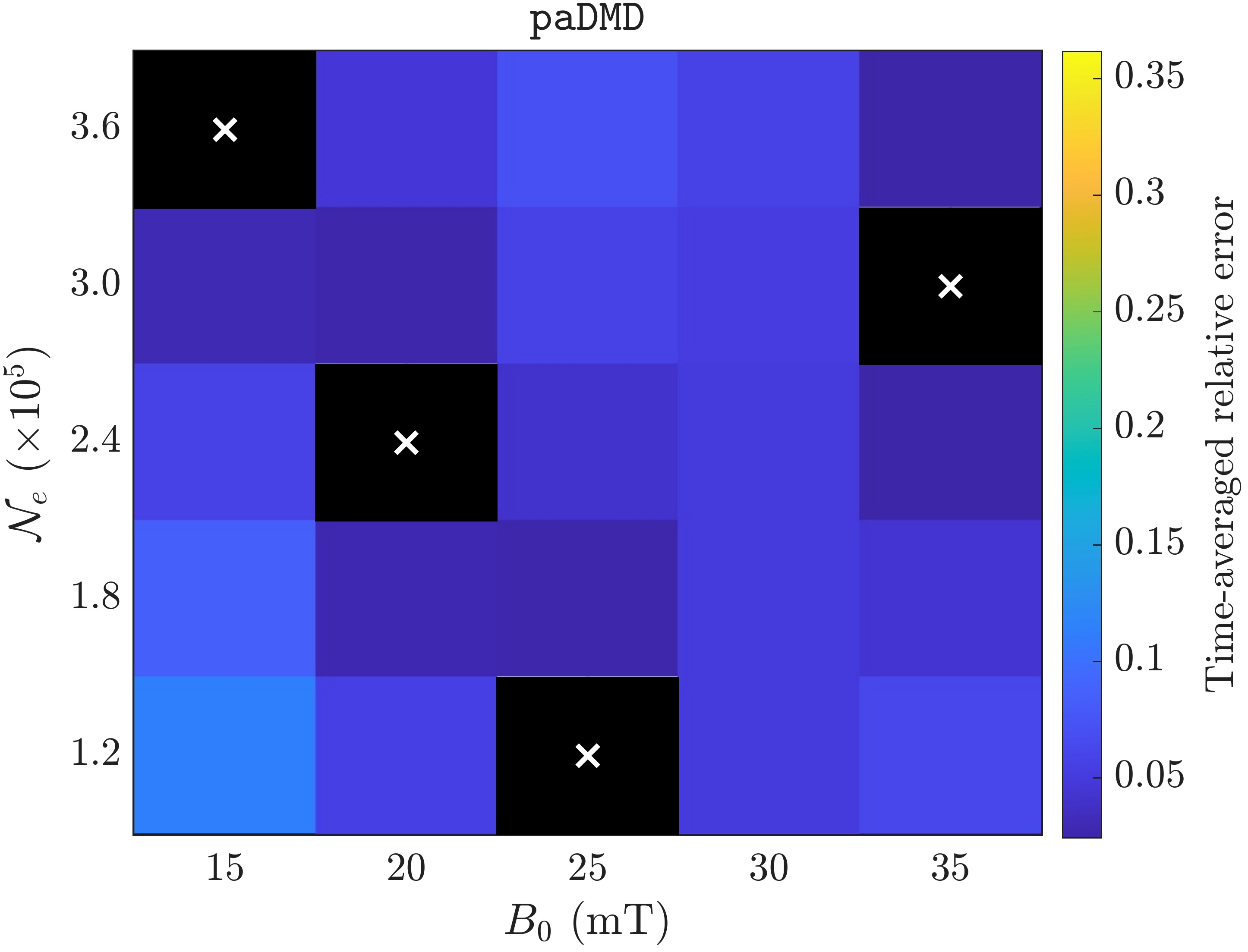}%
        \label{fig:wavy_padmd_err_2param}%
    }\hspace{0.02\textwidth}
    \subfloat[]{%
        \includegraphics[width=0.35\textwidth]{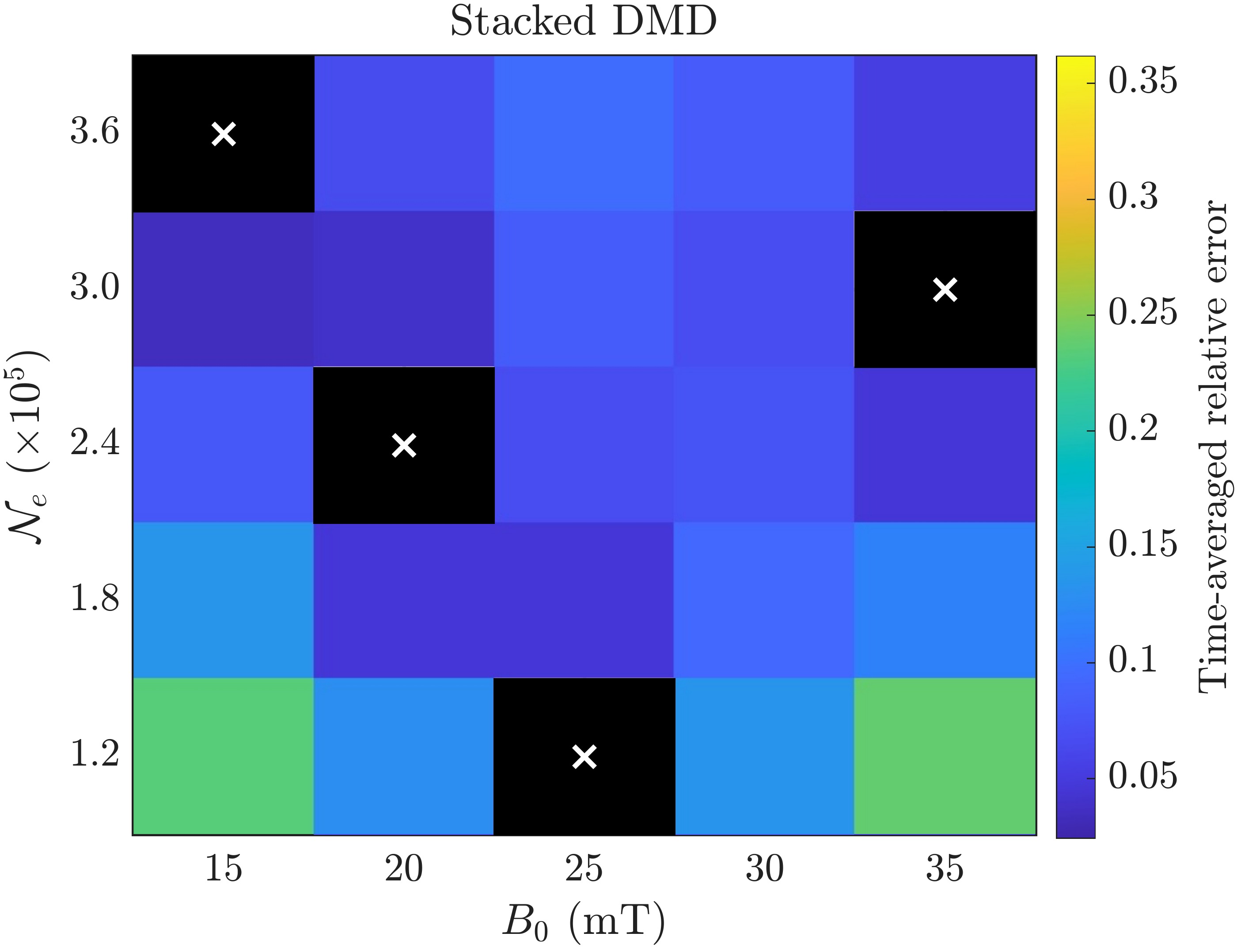}%
        \label{fig:wavy_stacked_err_2param}%
    }\\[0.8em]
    \subfloat[]{%
        \includegraphics[width=0.35\textwidth]{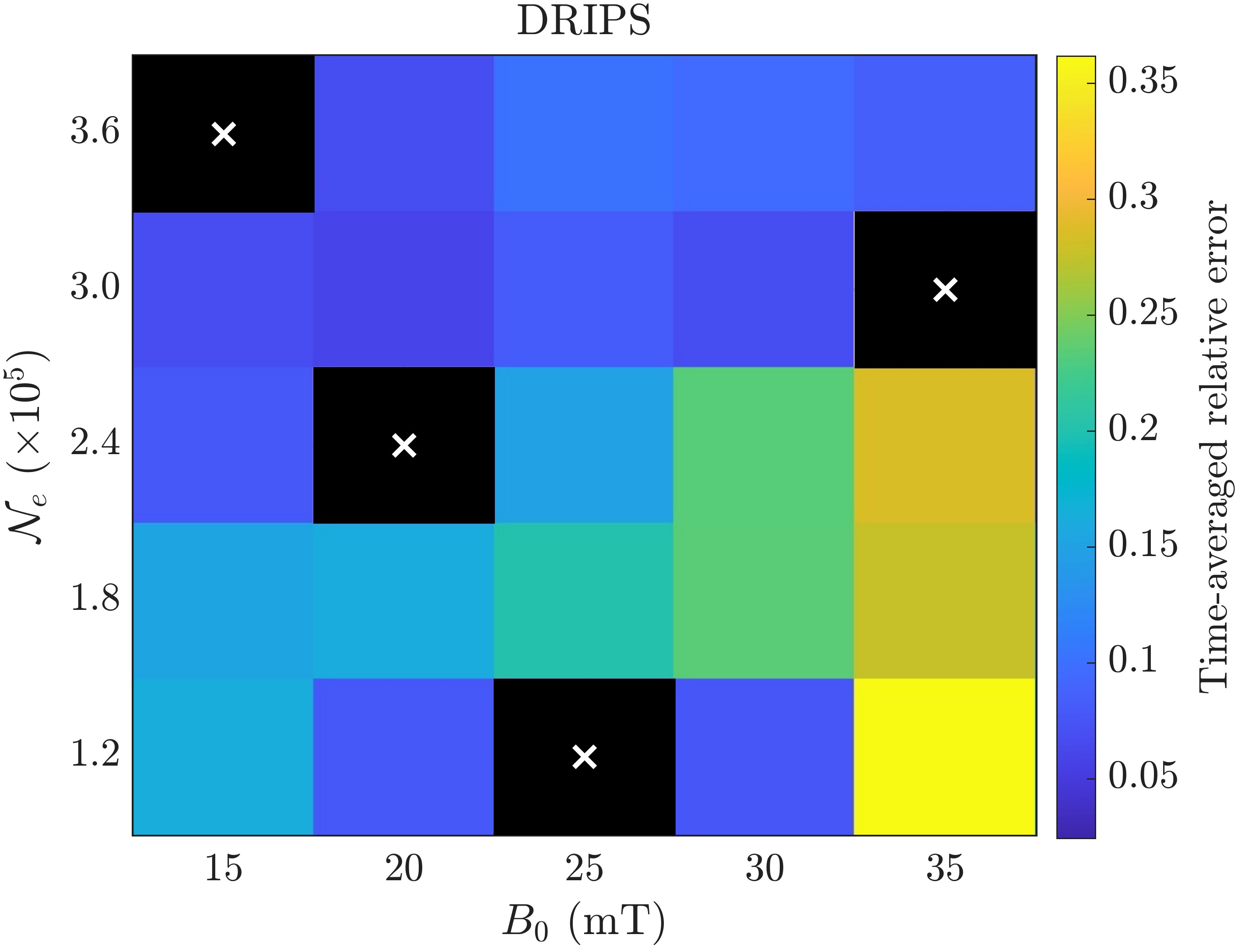}%
        \label{fig:wavy_drips_err_2param}%
    }\hspace{0.02\textwidth}
    \subfloat[]{%
        \includegraphics[width=0.35\textwidth]{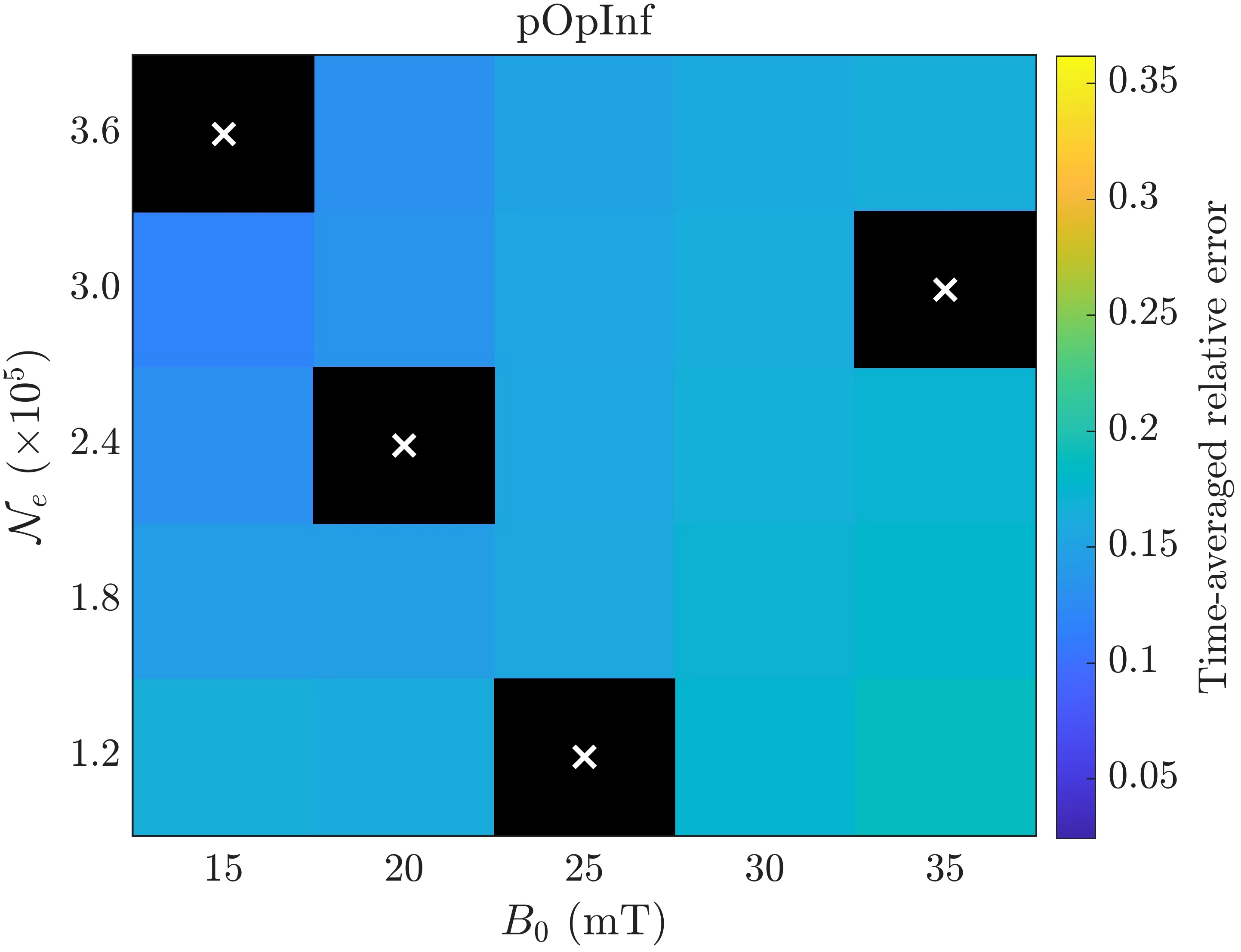}%
        \label{fig:wavy_popinf_err_2param}%
    }
    \caption{Time-averaged relative error distribution for the oscillating electron beam system over the parameter space, for (a) \texttt{paDMD}, (b) stacked DMD, (c) DRIPS and (d) pOpInf, at the ranks with lowest prediction error. Each square corresponds to one two-dimensional parameter sample $\boldsymbol{\theta}=[B_0 \quad \mc{N}_e]^\top$, colored by the time-averaged relative error obtained over a $1000$ time step prediction horizon. Training $\boldsymbol{\theta}$ samples are marked in black with a white $\times$ symbol.}
    \label{fig:wavy_checkerboards_2param}
\end{figure}

\subsection{Virtual cathode oscillations}
\label{subsec:results_vircator}

\begin{figure}[pos=!t]
    \centering
    \includegraphics[width=0.5\textwidth]{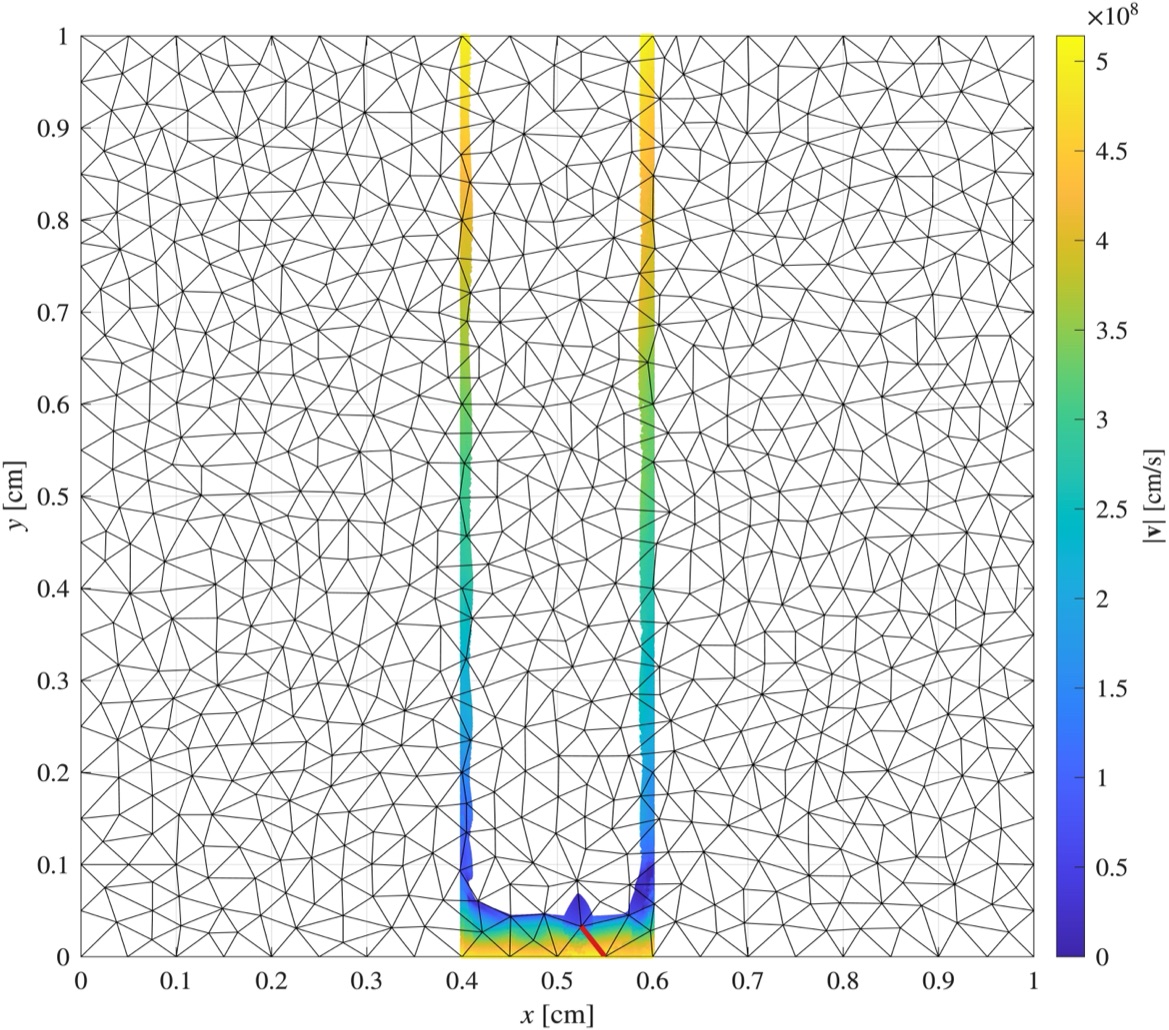}
    \caption{Virtual cathode snapshot at $33.6 ~\text{ns}$ with $r_{\mathrm{sp}}=1.5 \times 10^6$ and superparticle injection rate of 100 per time step. Each superparticle is colored based on its instantaneous speed $| \mf{v} |$. A sample mesh edge is highlighted in red (location: center-bottom).}
    \label{fig:vircator_particle_distrib}
\end{figure}

\begin{table}[pos=!h]
\centering

\footnotesize
\setlength{\tabcolsep}{3pt}
\caption{Mean relative prediction errors over the test $\mc{N}_e$ parameter samples of the virtual cathode oscillator, across various model ranks. All methods use the same delay embedding with $d=6$. The best result for each method is marked in bold, and `--' indicates error blow up.}
\label{tab:vco_ranks}
\begin{tabular}{lccccccccccc}
\toprule
\multicolumn{1}{c}{\multirow{2}{*}{Method}} & \multicolumn{11}{c}{Model rank ($r$)} \\
\cmidrule(lr){2-12}
 & 10 & 15 & 20 & 25 & 30 & 35 & 40 & 45 & 50 & 55 & 60 \\
\midrule
\makecell[l]{\texttt{paDMD}\\[-1pt]{\tiny $(\tilde r = r,~\hat r = \min\{r,dn\})$}} & 0.0291 & 0.0262 & 0.0227 & 0.0229 & 0.0249 & \textbf{0.0194} & 0.0311 & 0.0334 & 0.0296 & 0.0272 & 0.0361 \\
\midrule
DRIPS & -- & -- & -- & -- & -- & -- & -- & -- & -- & -- & -- \\
\midrule
pOpInf & 0.0321 & 0.0351 & 0.0303 & \textbf{0.0303} & 0.0415 & 0.0460 & 0.2306 & 0.3072 & 0.0745 & 0.0327 & 0.0325 \\
\midrule
Stacked DMD & 0.0368 & 0.0422 & 0.0368 & 0.0832 & 0.0431 & 0.0447 & 0.0454 & 0.0449 & 0.0409 & \textbf{0.0308} & 0.0373 \\
\midrule
rKOI & 0.0448 & 0.0411 & \textbf{0.0374} & 0.0430 & 0.0384 & 0.0431 & 0.0460 & 0.0453 & 0.0434 & 0.0442 & 0.0443 \\
\bottomrule
\end{tabular}

\end{table}

\begin{figure}[pos=!h]
    \centering
    \includegraphics[width=0.5\textwidth]{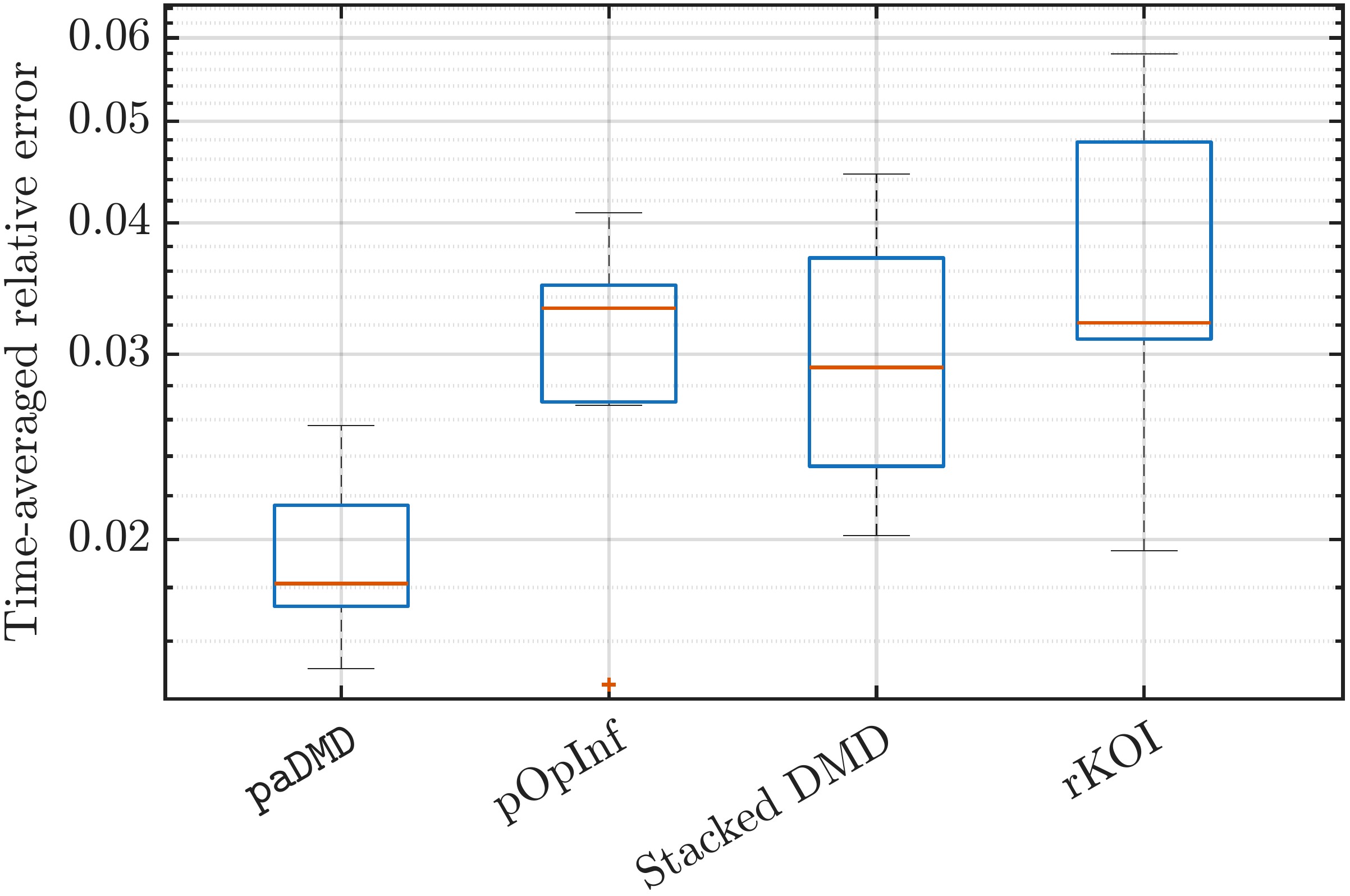}
    \caption{Distribution of the time-averaged relative error over the test $\mc{N}_e$ values for the virtual cathode oscillator. The error is averaged over the prediction horizon at each test $\mc{N}_e$ value, and each converging method is shown at the model rank with its lowest error.}    \label{fig:vircator_err_boxplot}
\end{figure}

\begin{figure}[pos=!htbp]
    \centering
    \includegraphics[width=0.6\textwidth]{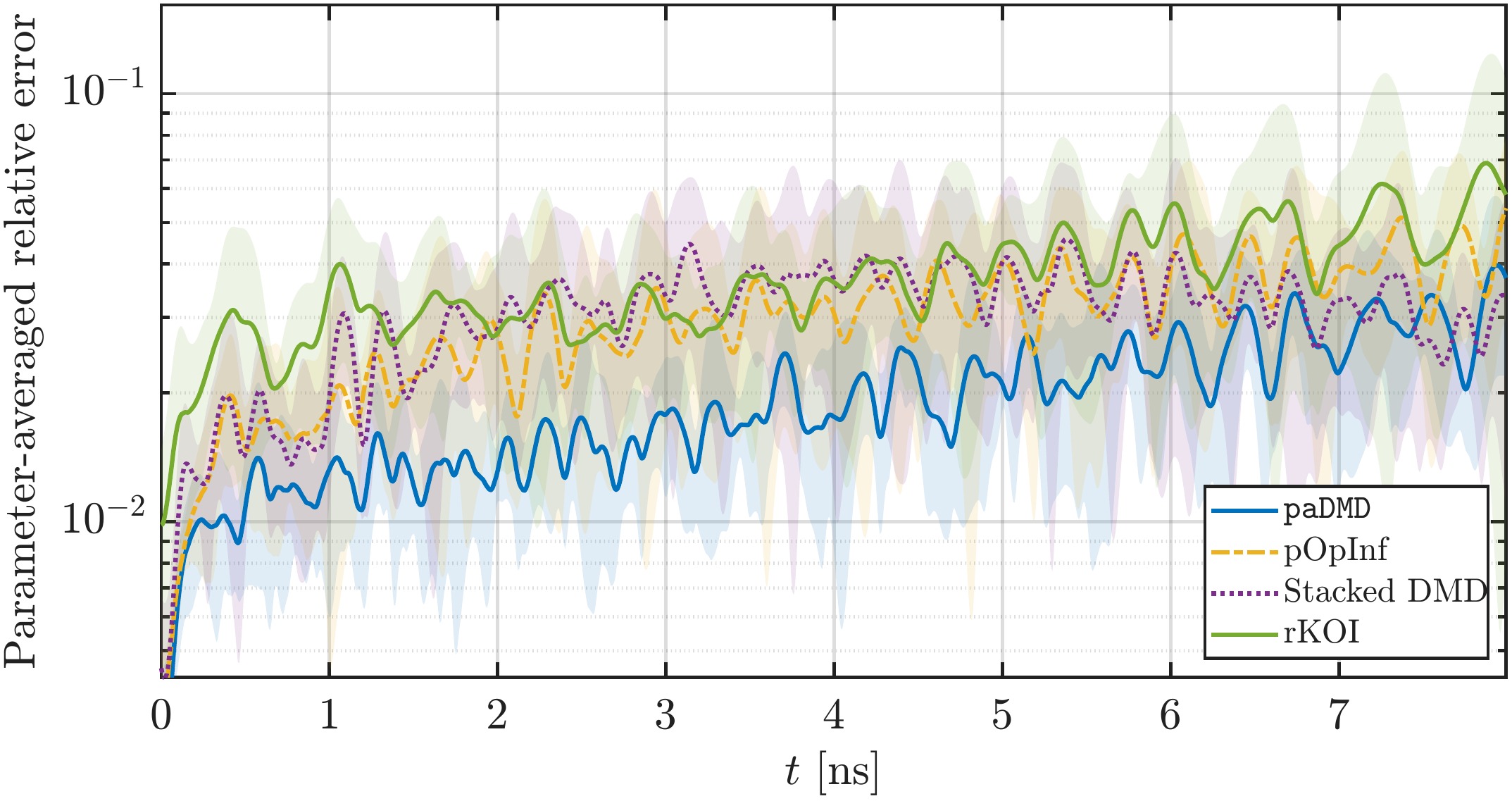}
    \caption{Relative error for the virtual cathode oscillator at the model ranks with lowest error for each converging method, averaged over the test $\mc N_e$ values at each time step in the prediction horizon. The shaded regions indicate the standard deviation across the test $\mc N_e$ values.} \label{fig:vircator_err_time_series}
\end{figure}

\begin{figure}[pos=!h]
\centering
\subfloat[$\boldsymbol{\theta}^*=\mc{N}_e=1.7 \times 10^8$]{%
  \includegraphics[width=0.31\textwidth]{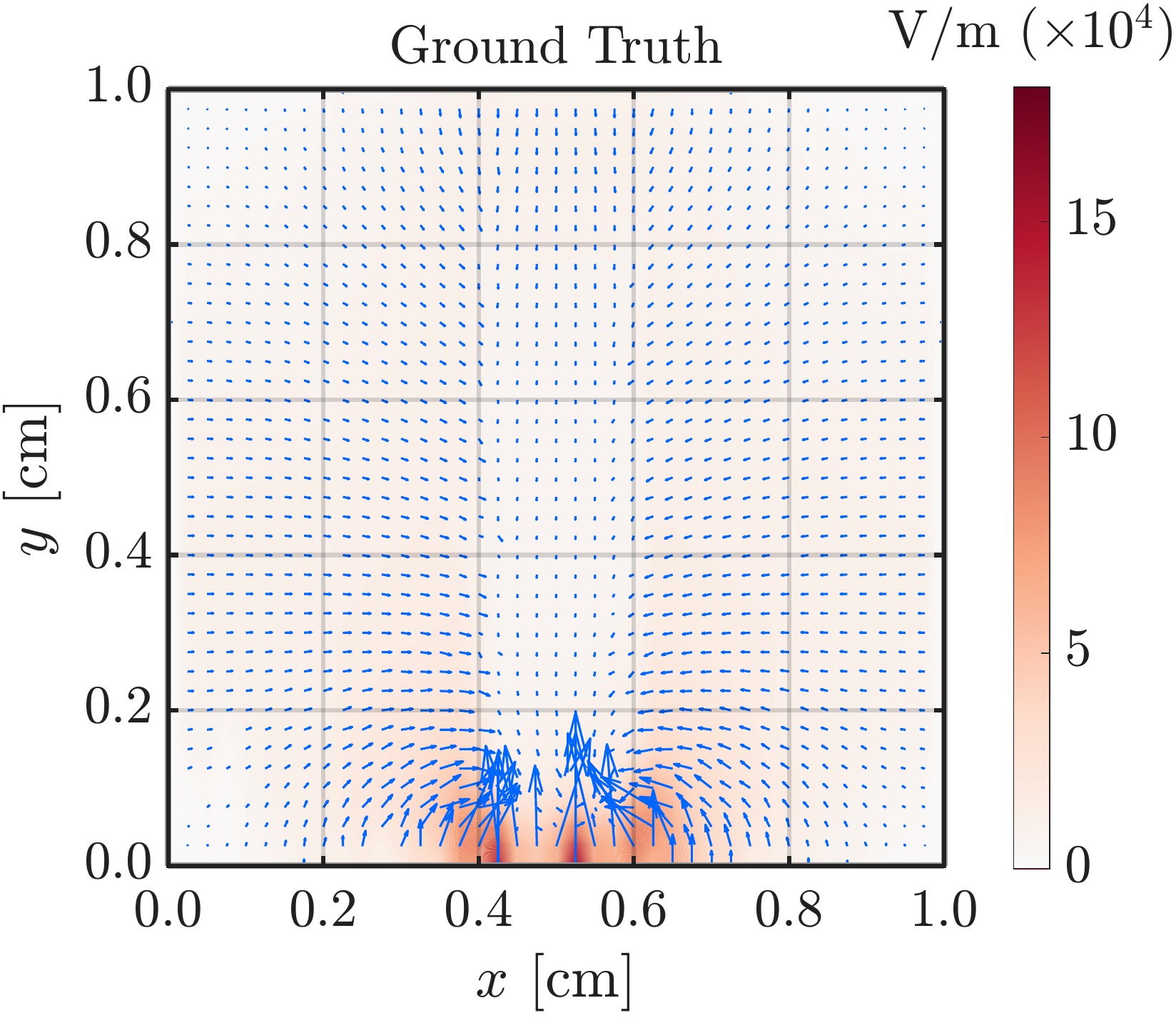}\hspace{0.02\textwidth}%
  \includegraphics[width=0.31\textwidth]{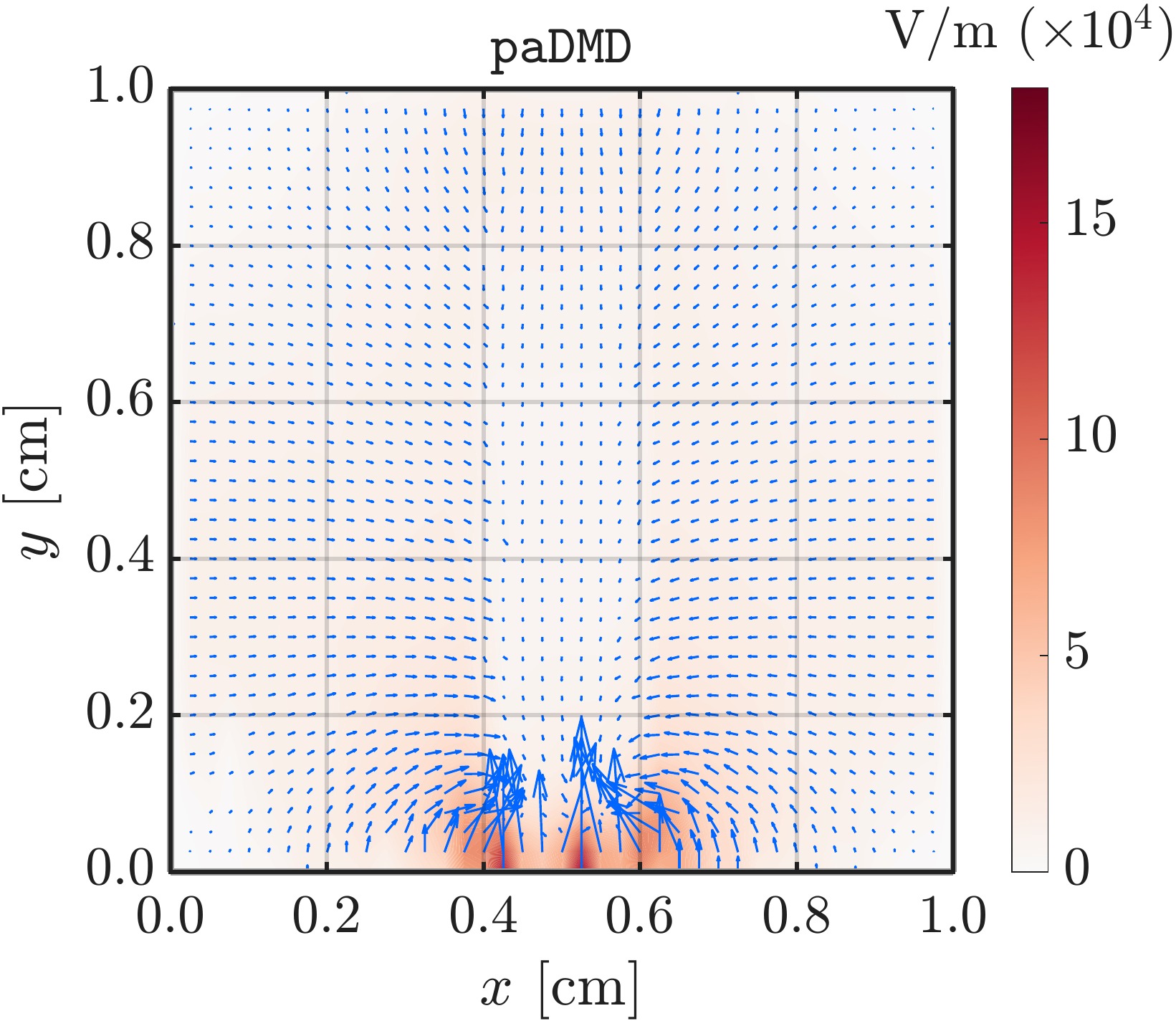}\hspace{0.02\textwidth}%
  \includegraphics[width=0.31\textwidth]{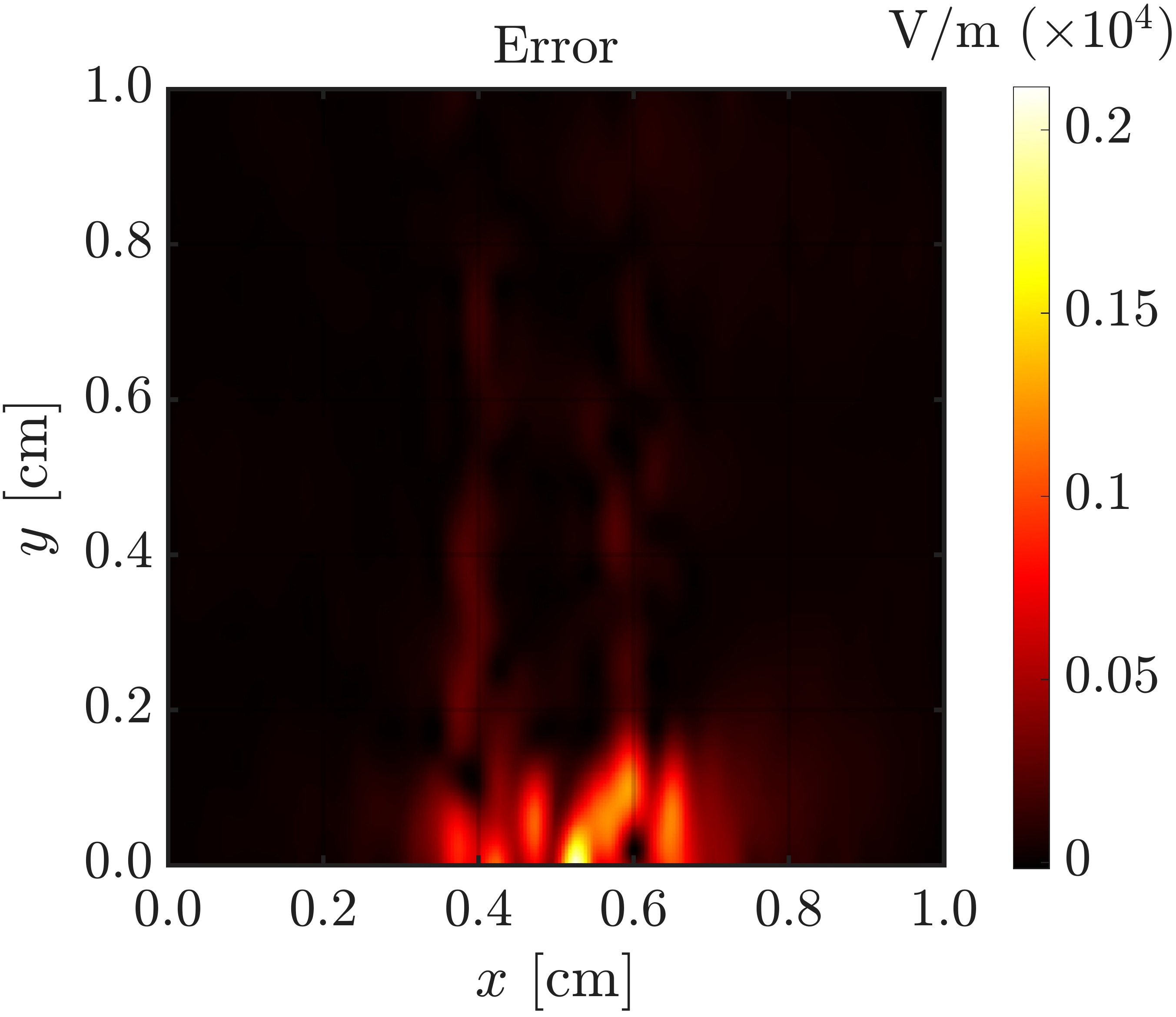}%
}

\vspace{0.8em}
\subfloat[$\boldsymbol{\theta}^*=\mc{N}_e=1.85 \times 10^8$]{%
  \includegraphics[width=0.31\textwidth]{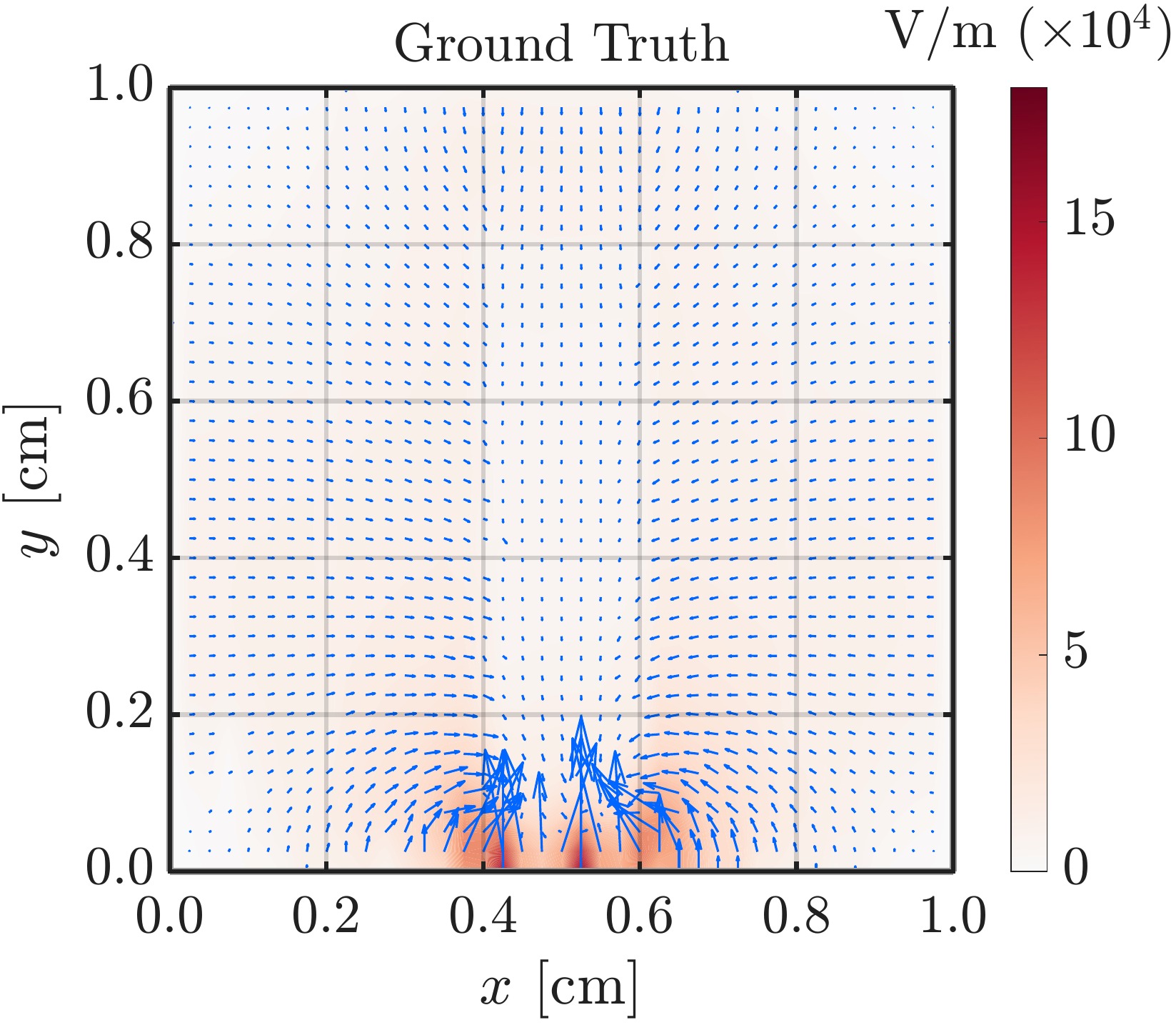}\hspace{0.02\textwidth}%
  \includegraphics[width=0.31\textwidth]{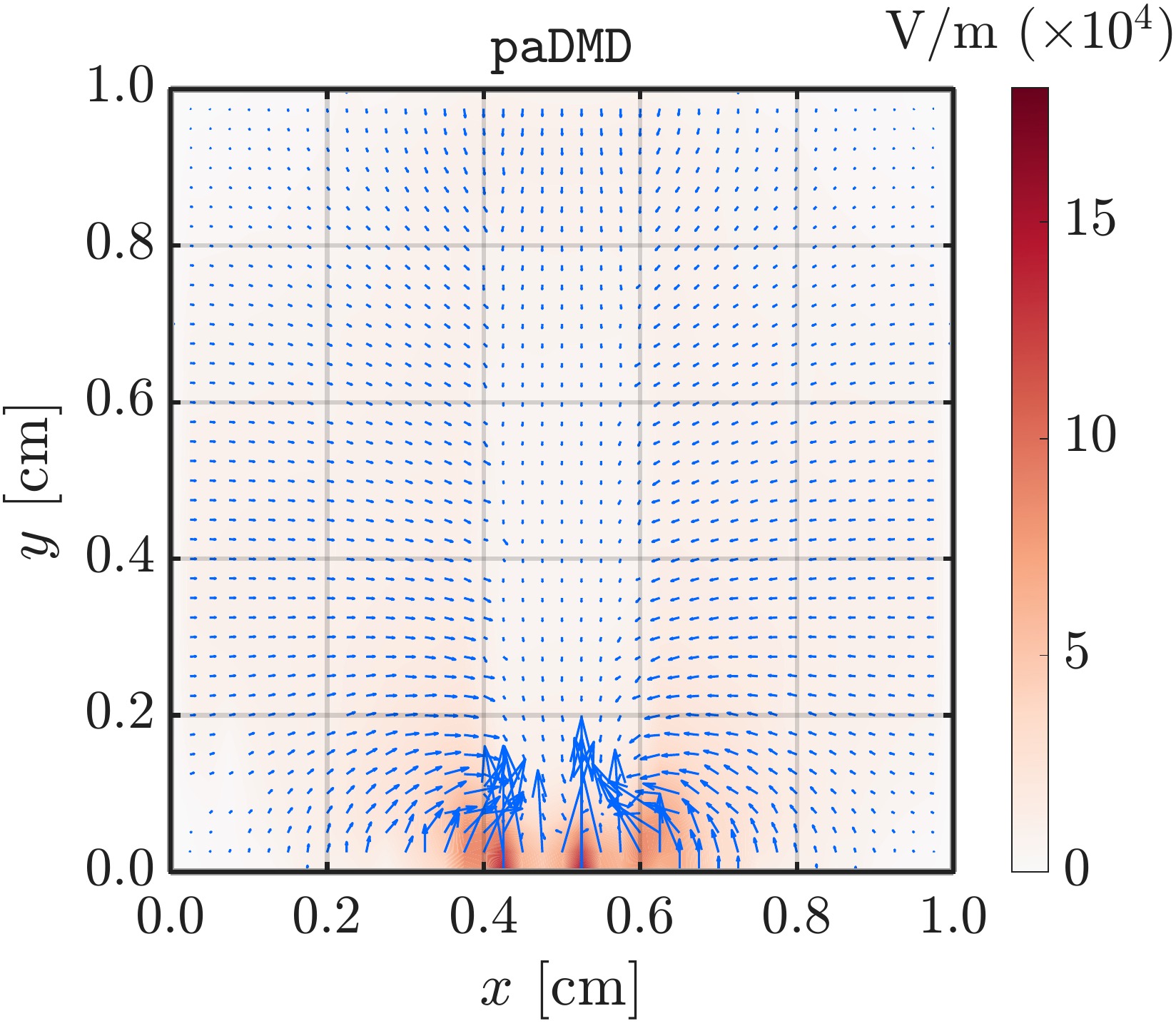}\hspace{0.02\textwidth}%
  \includegraphics[width=0.31\textwidth]{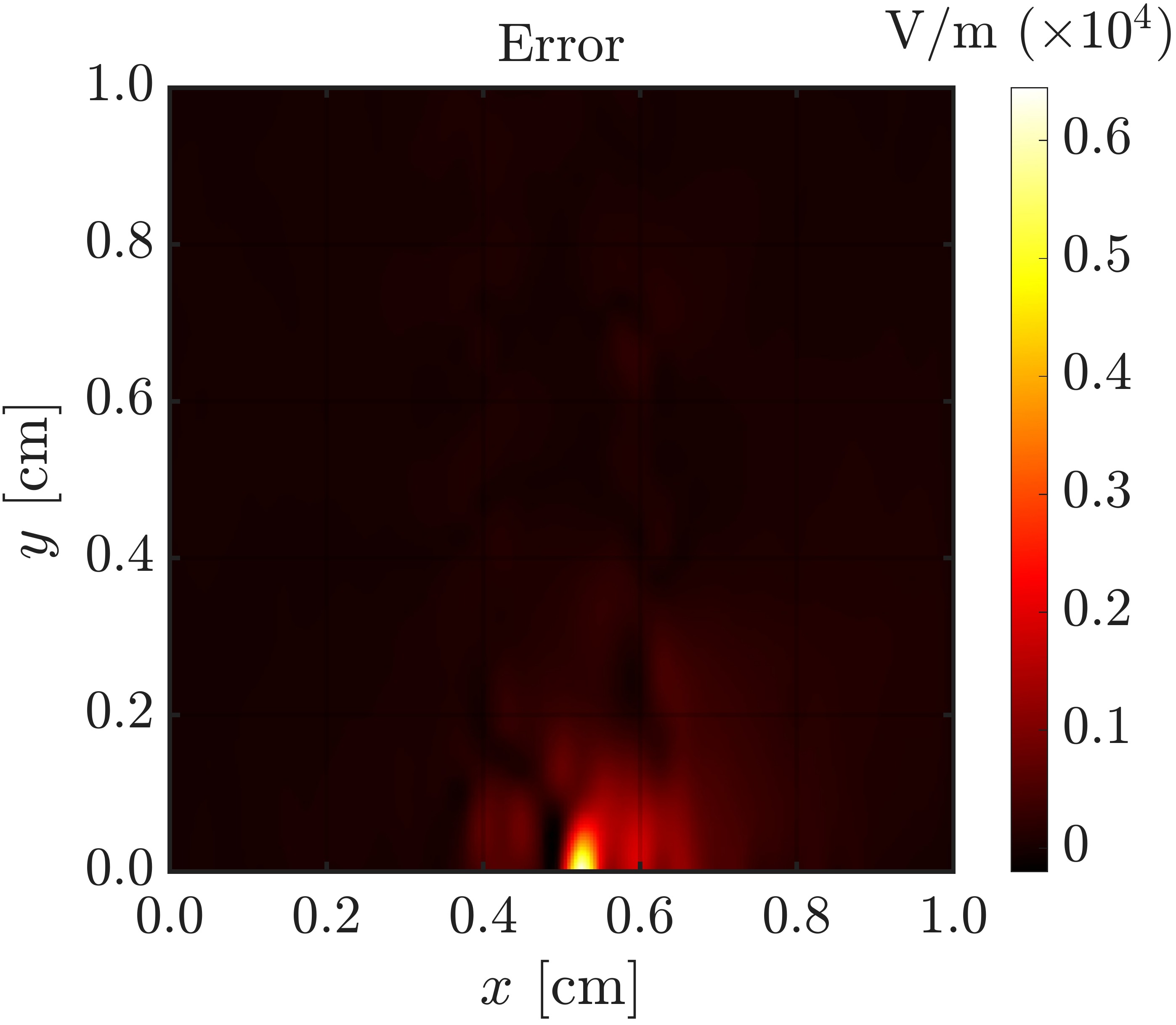}%
}

\caption{Electric field snapshots for virtual cathode oscillations at the final prediction time step for two different test values of $\mc{N}_e$. Each row shows the ground truth, \texttt{paDMD} prediction ($\tilde r = \hat r=35$), and pointwise absolute error.}
\label{fig:vircator_padmd_ground_snaps}
\end{figure}

\begin{figure}[pos=!h]
\centering
\includegraphics[width=0.35\textwidth]{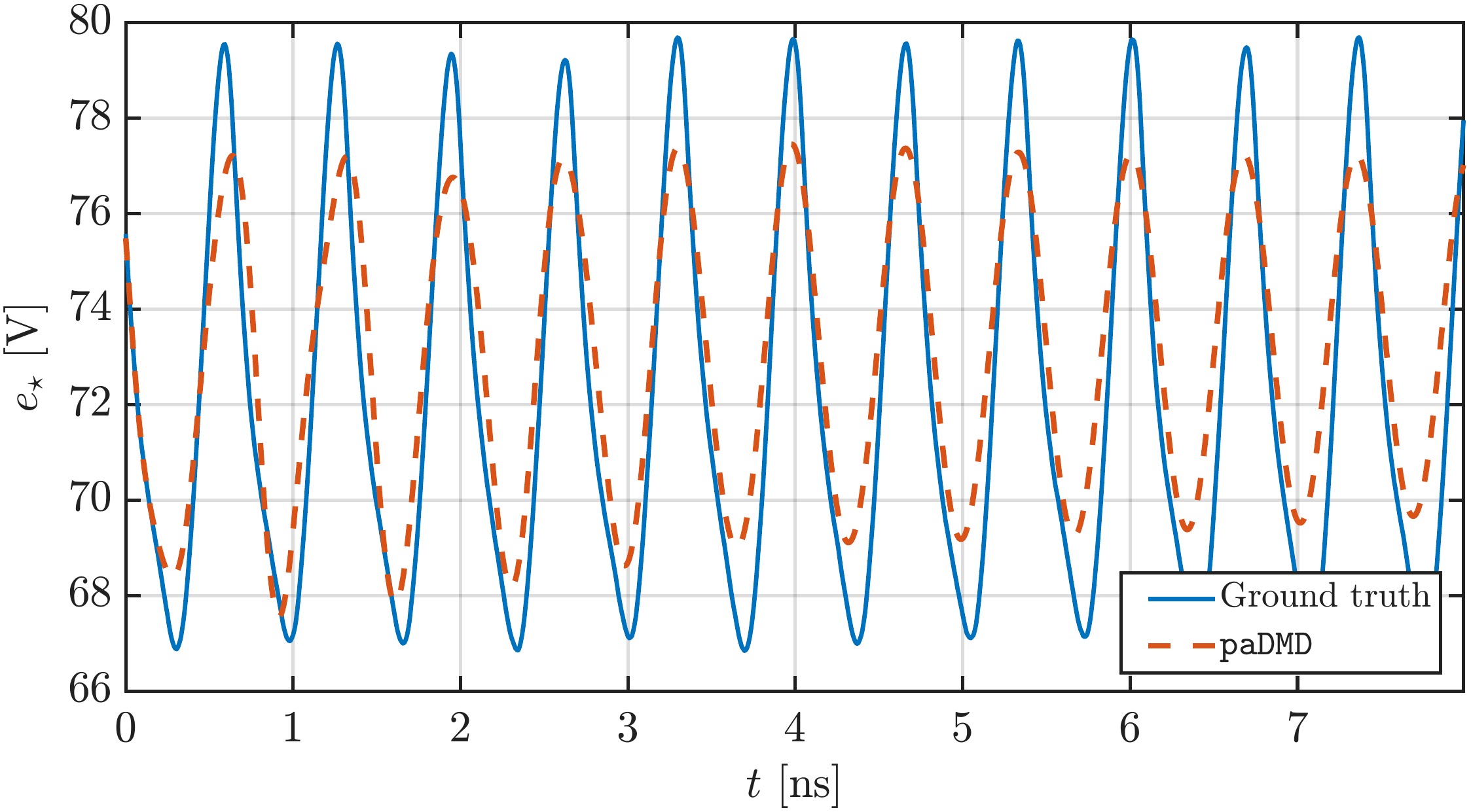}\hspace{0.02\textwidth}%
\includegraphics[width=0.35\textwidth]{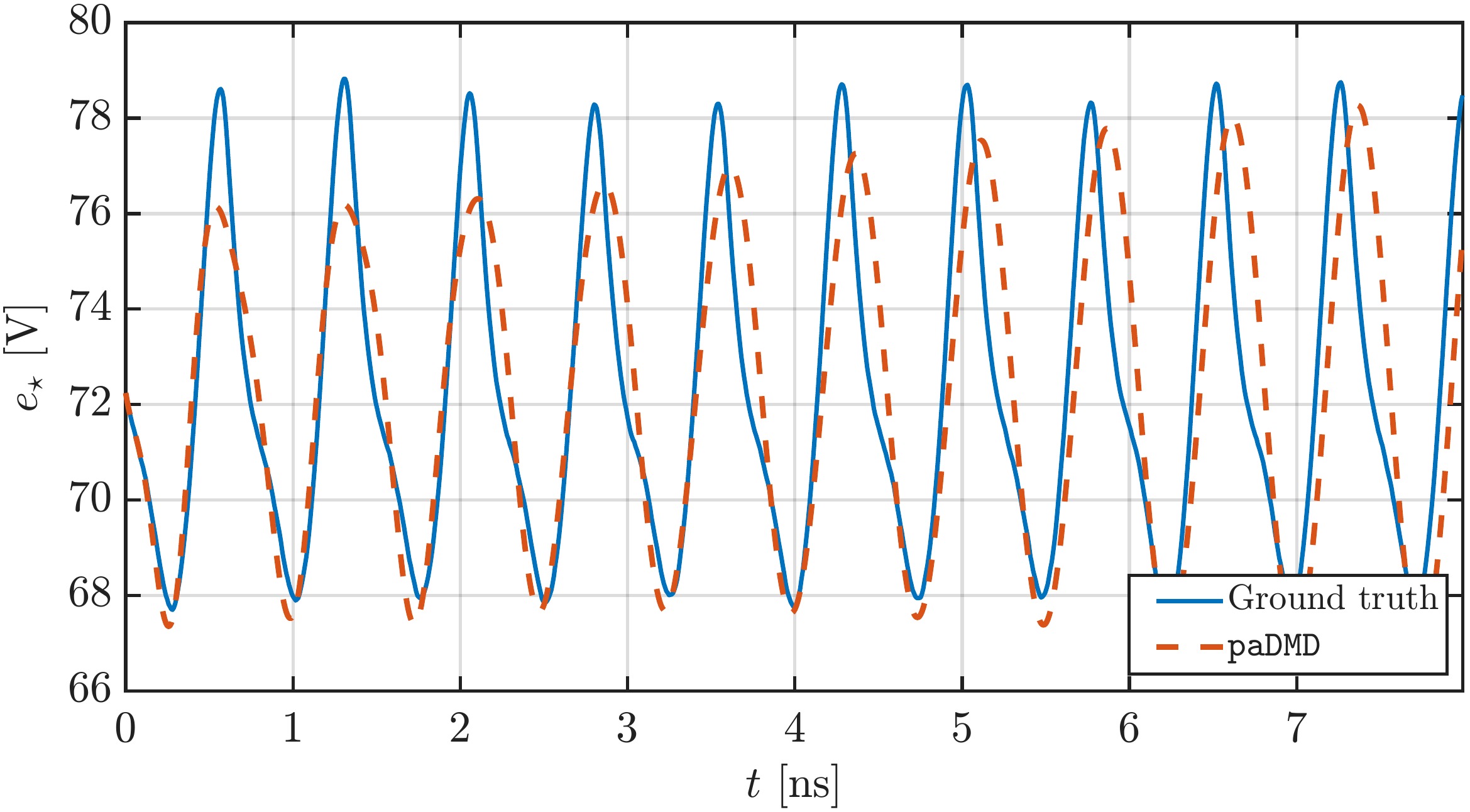}

\vspace{0.6em}
\subfloat[$\boldsymbol{\theta}^*=\mc N_e = 1.7 \times 10^8$]{%
  \includegraphics[width=0.35\textwidth]{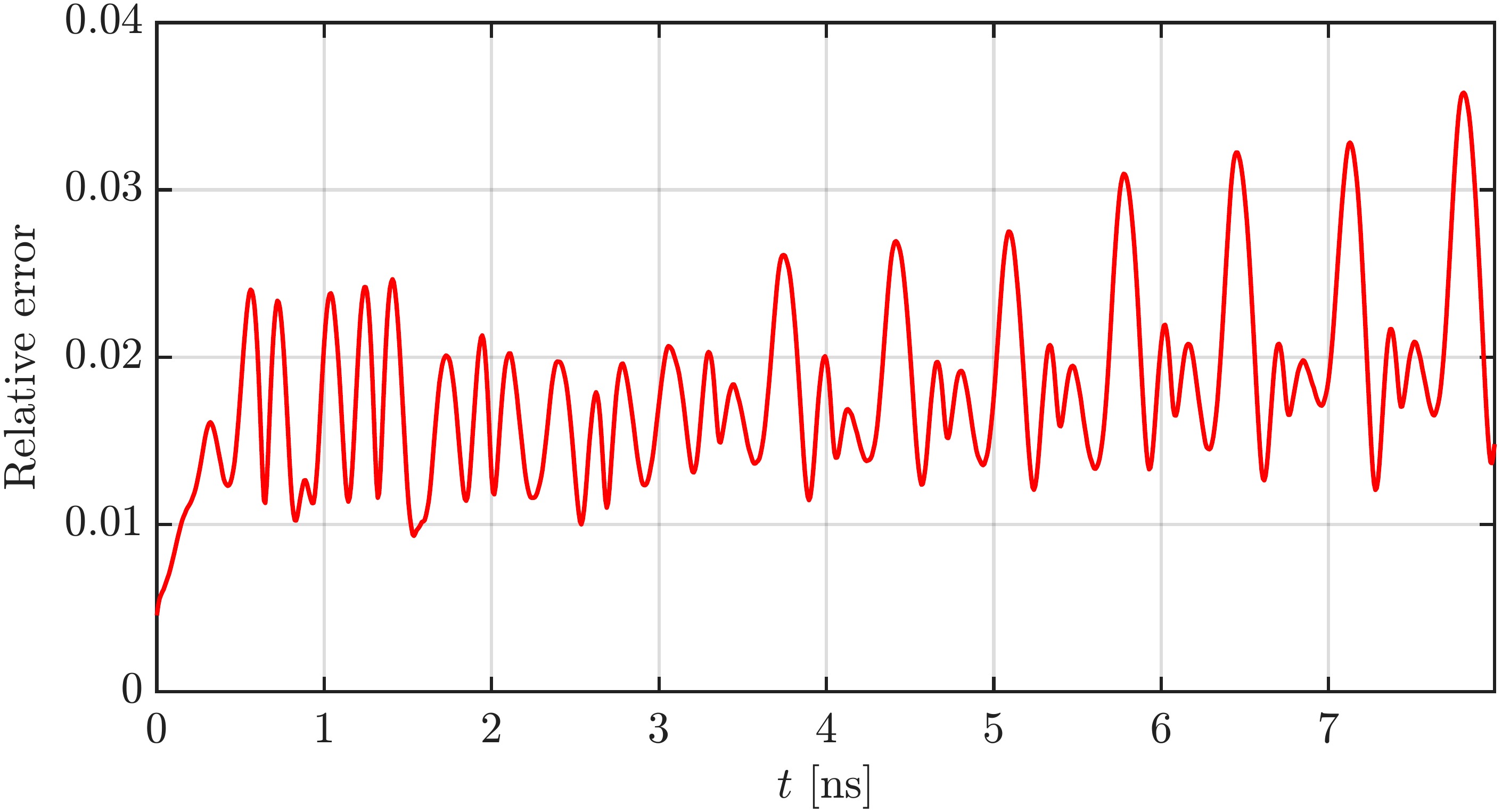}%
}\hspace{0.02\textwidth}
\subfloat[$\boldsymbol{\theta}^*=\mc N_e = 1.85 \times 10^8$]{%
  \includegraphics[width=0.35\textwidth]{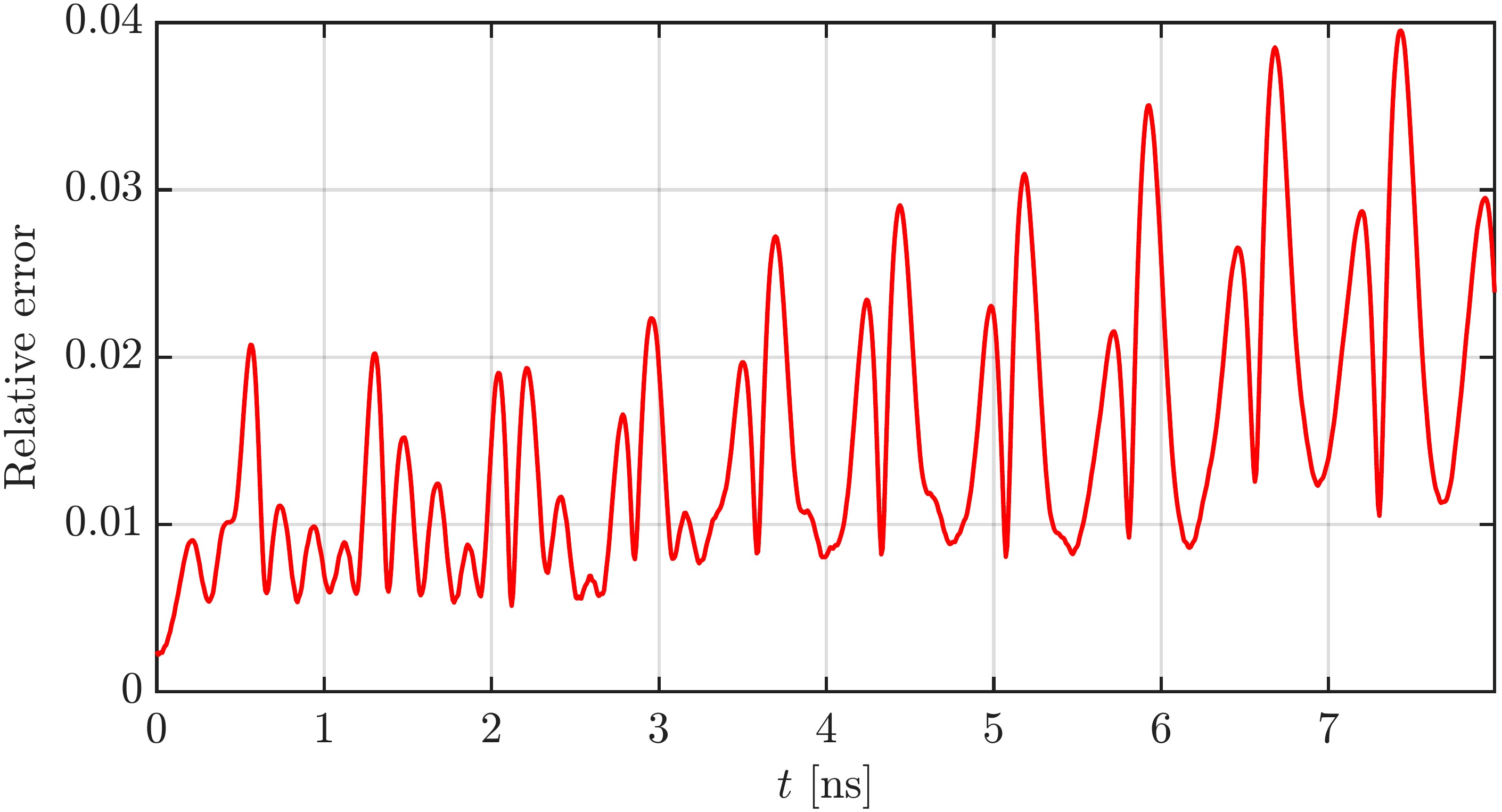}%
}\hspace{0.02\textwidth}

\caption{\emph{Top row}: \texttt{paDMD} prediction results of the electric field for the virtual cathode oscillator at two test values of $\mc N_e$, sampled at the location (mesh edge) highlighted in red in Figure~\ref{fig:vircator_particle_distrib}. \emph{Bottom row}: Corresponding $2$-norm relative errors of the \texttt{paDMD} electric field prediction ($\tilde r = \hat r=35$) over the prediction horizon.}
\label{fig:vircator_dof_err}
\end{figure}

We now consider a more complex plasma example of virtual cathode oscillations. 
A virtual cathode is obtained when the total current of the electron beam exceeds the space-charge limit, causing the formation of a potential barrier. The EMPIC setup of Section~\ref{subsec:results_electron_beam} is adopted with the following modifications. We increase the current injection by increasing the superparticle injection rate to $100$ per time step. The superparticles are injected at the bottom in the region [0.4 cm, 0.6 cm]. Instead of a transverse oscillating magnetic flux, we apply a strong confining magnetic field $\mf{B}=B_y \hat{\mf y}$ along the $y$-direction, with $B_y = 20 ~\text{T}$. We choose $\boldsymbol{\theta}=\mf h(\boldsymbol{\theta})=\mc{N}_e$ as our parameter of interest, and vary it by changing the superparticle ratio $r_{\mathrm{sp}}$ from $1.5 \times 10^6$ to $2 \times 10^6$ in increments of $5 \times 10^4$. Figure~\ref{fig:vircator_particle_distrib} shows the superparticle distribution for the virtual cathode at time 33.6 ns, for $r_{\mathrm{sp}}=1.5 \times 10^6$ and superparticle injection rate of 100 per time step, with each superparticle colored according to its instantaneous speed.

Since superparticles are injected at the rate of $100$ per time step, this results in $\mc{N}_e \in \{1.5 \times 10^8, ~1.55 \times 10^8, ~1.6 \times 10^8, ~\dots, ~2 \times 10^8\}$ per time step. We normalize the $\mc{N}_e$ values to be between 0 and 0.1. Here, as there are no external forces dictating a clear oscillation pattern of the electrons, this system is strongly nonlinear, and we therefore use the Hankel variant of \texttt{paDMD} with $d=6$ delay embeddings to learn an accurate system model. This results in an embedded state of dimension $dn=14682$. We choose $\mc{N}_e=\{1.5 \times 10^8, ~1.65 \times 10^8, ~1.8 \times 10^8, ~2 \times 10^8\}$, as our training parameter samples and the rest for testing (i.e., we have $L=4$ training and 7 test parameter samples). We discard the first $2000$ time steps due to transient behavior, choose the next $500$ time steps from each training parameter (i.e., $T=500$), and predict over a horizon of $1000$ time steps for each test parameter. Note that in this setting each method was evaluated using the delay-embedded state for a fair comparison. Table~\ref{tab:vco_ranks} reports the mean prediction error for every method across several model ranks, with the lowest error for each method marked in bold. Here, we find that DRIPS fails to converge for several test parameter samples for every model rank in the chosen range, and the time-averaged relative error blows up exponentially. This highlights the sensitivity of such interpolation-based parametric ROM methods to the choice of training parameter samples, in addition to the dimension of the parameter space. Figure~\ref{fig:vircator_err_boxplot} plots the time-averaged relative error for the remaining methods at the rank with the lowest error in that range, and Figure~\ref{fig:vircator_err_time_series} shows the time evolution of the relative prediction errors over the prediction horizon, averaged across the test $\mc N_e$ values. The ground truth vs prediction comparisons for \texttt{paDMD} at various test values of $\mc N_e$ are shown in Figures \ref{fig:vircator_padmd_ground_snaps} and \ref{fig:vircator_dof_err}.

\section{Conclusion}
\label{sec:conclusion}
This paper introduced parameter-affine dynamic mode decomposition (\texttt{paDMD}), a parametric DMD framework that learns a parameter-affine Koopman surrogate model of complex nonlinear systems from data. This method embeds known parameter functions directly into the regression step, instead of relying on interpolating modes, eigenvalues, or reduced operators. The \texttt{paDMD} method was evaluated on high-dimensional benchmarks from nonlinear fluid dynamics, namely flow past a cylinder and the fluidic pinball, as well as on kinetic plasma systems, namely oscillating electron beams and virtual cathode oscillations, each with varying governing parameters. Across all of them it demonstrated consistently strong predictive performance and improved robustness in settings with limited training parameter samples and multi-dimensional parameter spaces, where other parametric ROM methods can struggle to learn accurate system models. Despite its advantages, \texttt{paDMD} has its limitations. Like other Koopman-based methods, \texttt{paDMD} is somewhat sensitive to the choice of training data and parameters, particularly in the limited data regime considered here. Furthermore, the assumption that the parameter-affine functions are known \emph{a priori} is restrictive. Several avenues remain for future work. One is the formulation of rigorous error bounds for \texttt{paDMD}, in terms of its sensitivity to the prediction horizon, the model rank and the level of measurement noise. Another is the application of \texttt{paDMD} to more complex systems with heterogeneous, field-valued parameters. A third is the extension of the \texttt{paDMD} framework using neural networks to automate the selection of Koopman observables and the learning of parametrization functions.

\section*{Acknowledgment}
\label{sec:acknowledgment}
This work was funded in part by the Department of Energy Grant DE-SC0022982 through the NSF/DOE Partnership in Basic Plasma Science and Engineering, by the OSU-COE Strategic Research Initiative Grant GR138806, and by the Ohio Supercomputer Center Grants PAS-0061 and PAS-2709.

\appendix

\section{Sensitivity analysis}\label{app:sensitivity}

\subsection{Sensitivity to measurement noise}\label{app:noise}

The results of Section~\ref{sec:results} assume noise-free snapshots for training data. In practice, training data obtained from experiments or from stochastic simulations may be corrupted by measurement noise. We therefore repeat the experiments of Section~\ref{sec:results} with additive noise injected into the data, and measure the prediction error against the noise-free trajectories. Since noisy snapshots can lead to unstable predictions, we can use ridge regression/Tikhonov regularization \cite{golub1999tikhonov}, which can improve numerical stability.

Table~\ref{tab:noise_30db} reports the errors obtained with additive white Gaussian noise at $30$ dB SNR, using the same simulation setups as in Section~\ref{sec:results} except where noted below. The model ranks of every method are re-tuned under noise, over the same ranges swept in Section~\ref{sec:results}, and the regularization parameters of \texttt{paDMD} and pOpInf are likewise re-selected on the \emph{noisy} training data, so that each benchmark is evaluated at its best setting for the noise level rather than at one chosen on clean data. Accuracy generally deteriorates relative to the noise-free case, although a few methods improve slightly once their ranks and regularization are re-selected for the noisy data. In every one of the six systems, \texttt{paDMD} attains the lowest error. At higher noise levels, which we do not report here, predictions deteriorate further, particularly for the more strongly nonlinear systems. This can be mitigated to some extent by using longer delay embeddings. The fluidic pinball system, for instance, requires a deeper embedding of $d=24$ here than in the noise-free case, owing to its low-dimensional state combined with strongly nonlinear, non-parameter-affine dynamics. However, deeper embeddings demand greater computational resources, and the resulting models can still be fragile. Least-squares regression is inherently sensitive to noisy data, even when regularized. A recent work \cite{otto2024learning} addresses the partially observed and noisy setting in Koopman-based modeling explicitly, by treating the Koopman state as a bilinear hidden Markov model within an expectation--maximization framework. Neural network based formulations \cite{nayak2025temporally, lusch2018deep, azencot2020forecasting} offer another route, since their learned observables may improve robustness to noise when trained appropriately. Combining our parametric modeling framework with such methods is an interesting research direction and is the subject of ongoing work.

\begin{table}[pos=!h]
\centering

\footnotesize
\setlength{\tabcolsep}{5pt}
\caption{Mean relative prediction errors under $30$ dB SNR noise, measured against the noise-free trajectories and averaged over $20$ random noise realizations. The best result for each system is marked in bold, and `--' indicates error blow up.}
\label{tab:noise_30db}
\begin{tabular}{ccccccc}
\toprule
\multirow{2}{*}{System} & \multirow{2}{*}{Parameter ($\boldsymbol{\theta}$)} & \multicolumn{5}{c}{Method} \\
\cmidrule(lr){3-7}
 & & \texttt{paDMD} & DRIPS & pOpInf & Stacked DMD & rKOI \\
\midrule
\multirow{2}{*}{Flow past cylinder} & $\nu$ & \textbf{0.0181} & -- & 0.0353 & 0.0572 & 0.1255 \\
\cmidrule(lr){2-7}
 & $D$ & \textbf{0.0326} & -- & 0.0613 & 0.1173 & 0.1595 \\
\midrule
Fluidic pinball & $[v_1~v_2~v_3]^\top$ & \textbf{0.0456} & 0.0694 & 0.0623 & 0.0538 & 0.2219 \\
\midrule
\multirow{2}{*}{Oscillating electron beam} & $B_0$ & \textbf{0.0518} & 0.0810 & 0.1622 & 0.0549 & 0.1885 \\
\cmidrule(lr){2-7}
 & $[B_0~\mc{N}_e]^\top$ & \textbf{0.0642} & 0.5528 & -- & 0.0950 & -- \\
\midrule
Virtual cathode oscillator & $\mc{N}_e$ & \textbf{0.0239} & -- & 0.0272 & 0.0288 & 0.0569 \\
\bottomrule
\end{tabular}

\end{table}

\subsection{Sensitivity to the number of training parameter samples}\label{app:ntrain}

We next examine how the accuracy of each method depends on the number of training parameter samples $L$. These experiments use clean data. For each system, while the number of training parameter samples is varied, the test set is held \emph{fixed} and is disjoint from the training set at every $L$, and the remaining setup and model parameters are kept the same as those of Section~\ref{sec:results}. At each $L$ the training samples are selected to span the parameter range, as in Section~\ref{sec:results}. Table~\ref{tab:ntrain} reports the results. The error of each method, including \texttt{paDMD}, generally decreases as training samples are added, which is as expected. There are some cases, however, where there is a slight increase in error from one value of $L$ to the next (e.g., from $L=4$ to $L=5$). This highlights the sensitivity of all such DMD/ROM based methods to the choice of training parameter samples, especially in the low-data regime. In \texttt{paDMD}, the operators $\{\tilde{\mathbf{B}}_i\}_{i=1}^m$ in \eqref{eq:padmd_closed_form_A_B} are identified from the variation of the dynamics across the training parameter samples, so a denser sampling of the parameter space excites this variation more richly and improves the persistency of excitation of the least-squares problem. We also observe the lowest prediction errors for \texttt{paDMD} for all systems in all experiments. This highlights the advantage of \texttt{paDMD} in settings with limited numbers of training parameter samples.

\begin{table}[pos=!h]
\centering

\footnotesize
\setlength{\tabcolsep}{5pt}
\caption{Mean relative prediction errors for varying number of training parameter samples $L$, on a fixed test set of parameter samples disjoint from every training set. The best result for each system is marked in bold, and `--' indicates error blow up.}
\label{tab:ntrain}
\begin{tabular}{cccccccc}
\toprule
\multirow{2}{*}{System} & \multirow{2}{*}{Parameter ($\boldsymbol{\theta}$)} & \multirow{2}{*}{$L$} & \multicolumn{5}{c}{Method} \\
\cmidrule(lr){4-8}
 & & & \texttt{paDMD} & DRIPS & pOpInf & Stacked DMD & rKOI \\
\midrule
\multirow{6}{*}{Flow past cylinder} & \multirow{3}{*}{$\nu$} & 3 & \textbf{0.0096} & 0.0190 & 0.0277 & 0.0640 & 0.0478 \\
 &  & 4 & \textbf{0.0050} & 0.0906 & 0.0279 & 0.0499 & 0.1297 \\
 &  & 5 & \textbf{0.0045} & 0.0329 & 0.0269 & 0.0381 & 0.1275 \\
\cmidrule(lr){2-8}
 & \multirow{3}{*}{$D$} & 3 & \textbf{0.0281} & 0.0310 & 0.0944 & 0.1359 & 0.1705 \\
 &  & 4 & \textbf{0.0168} & 0.0224 & 0.0397 & 0.1318 & 0.2571 \\
 &  & 5 & \textbf{0.0108} & 0.0168 & 0.0468 & 0.1116 & 0.1141 \\
\midrule
\multirow{3}{*}{Fluidic pinball} & \multirow{3}{*}{$[v_1~v_2~v_3]^\top$} & 10 & \textbf{0.0345} & 0.1066 & 0.0514 & 0.0571 & 0.2425 \\
 &  & 15 & \textbf{0.0312} & 0.0922 & 0.0535 & 0.0691 & 0.2394 \\
 &  & 20 & \textbf{0.0317} & 0.0990 & 0.0535 & 0.0650 & 0.2573 \\
\midrule
\multirow{6}{*}{Oscillating electron beam} & \multirow{3}{*}{$B_0$} & 3 & \textbf{0.0373} & 0.0497 & 0.1503 & 0.0483 & 0.1817 \\
 &  & 4 & \textbf{0.0199} & 0.0331 & 0.1477 & 0.0281 & 0.1642 \\
 &  & 5 & \textbf{0.0193} & 0.0370 & 0.1472 & 0.0230 & 0.1812 \\
\cmidrule(lr){2-8}
 & \multirow{3}{*}{$[B_0~\mc{N}_e]^\top$} & 3 & \textbf{0.0572} & -- & 0.1574 & 0.0904 & -- \\
 &  & 4 & \textbf{0.0523} & 0.1279 & 0.1534 & 0.0924 & -- \\
 &  & 5 & \textbf{0.0293} & 0.1434 & 0.1531 & 0.0706 & -- \\
\midrule
\multirow{3}{*}{Virtual cathode oscillator} & \multirow{3}{*}{$\mc{N}_e$} & 3 & \textbf{0.0258} & -- & 0.0386 & 0.0389 & 0.0471 \\
 &  & 4 & \textbf{0.0184} & -- & 0.0287 & 0.0298 & 0.0344 \\
 &  & 5 & \textbf{0.0192} & -- & 0.0290 & 0.0380 & 0.0832 \\
\bottomrule
\end{tabular}

\end{table}

\section{Computational cost}\label{app:timing}

Besides accuracy, the practical value of a parametric ROM rests on the cost of constructing it from the training data and of evaluating it across the parameter space. We therefore report the measured computation time of every method, together with the analytical cost of \texttt{paDMD}. Table~\ref{tab:timing} reports the computational time of each method on the test systems, at the optimal model ranks and training parameter samples used in Section~\ref{sec:results}. Each entry is the total cost of constructing the model and then predicting at all the test parameter samples, excluding the tuning of model hyperparameters. All computations were carried out in MATLAB R2025b on a MacBook Pro with an Apple M2 processor with $8$ cores and $8$ GB of memory.

The dominant cost of \texttt{paDMD} is the truncated SVD of the augmented regression matrix $\mf\Psi$ assembled from \eqref{eq:padmd_kron_lift}, which has $(m+1)n$ rows rather than the $n$ rows of a standard DMD snapshot matrix, or $(m+1)dn$ rows for the Hankel variant with $d$ delays. This costs $\mathcal{O}\big((m+1)n\,LT\,\tilde r\big)$ for a rank-$\tilde r$ truncated factorization, with $n$ replaced by $dn$ for the Hankel variant, so the cost of learning the operators grows linearly in the number of parametrization functions $m$. Evaluating the learned model at a new parameter $\boldsymbol{\theta}^*$ is considerably cheaper: forming $\tilde{\mf K}^{(\boldsymbol{\theta}^*)}$ in \eqref{eq:padmd_K_theta} from the already-learned operators requires $m$ matrix additions, followed by the projection \eqref{eq:padmd_rom} and the eigendecomposition of the $\hat r\times\hat r$ matrix $\hat{\mf K}_{r}^{(\boldsymbol{\theta}^*)}$, at costs of $\mathcal{O}(n^2\hat r)$ and $\mathcal{O}(\hat r^3)$ respectively, with neither retraining nor interpolation between stored models required. The measured times reflect this: across most systems \texttt{paDMD} stays within a factor of four of the fastest benchmark method, while the virtual cathode oscillator, which uses $d=6$ and therefore the largest $\mf\Psi$, has the largest computation time, at roughly eight times the fastest benchmark. We emphasize that the exact timings in practice will vary depending on the hardware and software environment used.

\begin{table}[pos=!h]
\centering

\footnotesize
\setlength{\tabcolsep}{5pt}
\caption{Computational time (in seconds) for each method on each system, at the optimal model ranks and training parameter samples as described in Section~\ref{sec:results}. Entries marked `--' denote a method whose predictions diverge at every rank in the swept range for that system, so that no optimal rank is available at which to report a time.}
\label{tab:timing}
\begin{tabular}{ccccccc}
\toprule
\multirow{2}{*}{System} & \multirow{2}{*}{Parameter ($\boldsymbol{\theta}$)} & \multicolumn{5}{c}{Method} \\
\cmidrule(lr){3-7}
 & & \texttt{paDMD} & DRIPS & pOpInf & Stacked DMD & rKOI \\
\midrule
\multirow{2}{*}{Flow past cylinder} & $\nu$ & 0.402 & 0.399 & 0.230 & 0.213 & 0.246 \\
\cmidrule(lr){2-7}
 & $D$ & 0.454 & 0.327 & 0.270 & 0.209 & 0.188 \\
\midrule
Fluidic pinball & $[v_1~v_2~v_3]^\top$ & 0.021 & 0.025 & 0.021 & 0.017 & 0.010 \\
\midrule
\multirow{2}{*}{Oscillating electron beam} & $B_0$ & 0.236 & 0.206 & 0.156 & 0.135 & 0.125 \\
\cmidrule(lr){2-7}
 & $[B_0~\mc{N}_e]^\top$ & 1.430 & 0.960 & 0.690 & 0.410 & -- \\
\midrule
Virtual cathode oscillator & $\mc{N}_e$ & 8.871 & -- & 4.536 & 2.229 & 1.082 \\
\bottomrule
\end{tabular}

\end{table}

\bibliographystyle{elsarticle-num}
\bibliography{bilinear_param_DMD}

\end{document}